%% file: main.tex
\documentclass[10pt,journal,final]{IEEEtran}

\usepackage{amsthm}
\usepackage{amsmath}
\usepackage{amssymb}
\usepackage{graphicx}
\usepackage{mathtools}  
\mathtoolsset{showonlyrefs} 

\usepackage{pgf}
\usepackage{pgfplots,tikz,}
\pgfplotsset{compat=newest}
\usetikzlibrary{plotmarks,automata}
\usetikzlibrary{arrows.meta}
\usepgfplotslibrary{patchplots}
\usepackage[font=footnotesize]{caption}
\usepackage{subcaption}

\usetikzlibrary{intersections,calc,patterns,angles,quotes}
\usetikzlibrary{decorations.shapes,arrows,calc,decorations.markings,shapes,decorations.pathreplacing,shapes.arrows,arrows.meta	}
\usepackage{tikz-3dplot}

\usepackage{algorithm}
\usepackage[noend]{algpseudocode}

\renewcommand{\algorithmicrequire}{\textbf{Input: }}
\renewcommand{\algorithmicensure}{\textbf{Output: }}
\newcommand{\R}{\mathbb R}
\newcommand{\C}{\mathbb C}

\newcommand{\G}{\mathcal G}
\newcommand{\V}{\mathcal V}
\newcommand{\E}{\mathcal E}
\newcommand{\nH}{\mathcal H}
\newcommand{\nS}{\mathcal S}
\newcommand{\bB}{\mathbb B}
\DeclareMathOperator*{\re}{Re}
\DeclareMathOperator*{\opt}{opt}
\newcommand{\tol}{\operatorname{tol}}

\DeclareMathOperator*{\init}{init}

\theoremstyle{plain}
\newtheorem{thm}{Theorem}[section]
\newtheorem{prop}[thm]{Proposition}
\newtheorem{lem}[thm]{Lemma}

\theoremstyle{remark}
\newtheorem{rem}[thm]{Remark}
\theoremstyle{definition}

\newtheorem{assmpt}{Assumption}

\newcommand{\longthmtitle}[1]{\mbox{}{\textit{(#1).}}}
\newcommand{\until}[1]{\{1,\dots,#1\}}

\newcommand{\oprocendsymbol}{\hbox{$\blacksquare$}}
\newcommand{\oprocend}{\relax\ifmmode\else\unskip\hfill\fi\oprocendsymbol}

\renewcommand{\baselinestretch}{.985}
\allowdisplaybreaks

\title{Iterative Algorithms for Assessing
  Network Resilience Against Structured Perturbations}

\author{Shenyu~Liu \quad Sonia~Mart{\'\i}nez \quad Jorge~Cort\'es
  \thanks{The authors are with the Department of Aerospace and
    Mechanical Engineering, UC San Diego, {\tt \{shl055,soniamd,
      cortes\}@ucsd.edu}. This work was partially supported by AFOSR
    Award FA9550-19-1-0235.}}

\begin{document}

\maketitle

\begin{abstract}
  This paper studies network resilience against structured additive
  perturbations to its topology. We consider dynamic networks modeled
  as linear time-invariant systems subject to perturbations of bounded
  energy satisfying specific sparsity and entry-wise constraints.
  Given an energy level, the structured pseudospectral abscissa
  captures the worst-possible perturbation an adversary could employ
  to de-stabilize the network, and the structured stability radius is
  the maximum energy in the structured perturbation that the network
  can withstand without becoming unstable.  Building on a novel
  characterization of the worst-case structured perturbation, we
  propose iterative algorithms that efficiently compute the structured
  pseudospectral abscissa and structured stability radius.  We provide
  theoretical guarantees of the local convergence of the algorithms
  and illustrate their efficacy and accuracy on several network
  examples.
\end{abstract}

\section{Introduction}

The resilience of dynamic networks against perturbations and attacks
is key across engineering, scientific, and military domains, including
the operation of cyberphysical infrastructure, the distributed control
of autonomous robots, and time-critical missions.  Despite important
advances in designing distributed coordination, cooperation, and
decision-making algorithms, dynamic networks have proven fragile to
targeted attacks, as local and well-orchestrated actions have rapidly
cascaded into network-wide destructive perturbations.  Because of
this, it is critical to develop techniques and notions that
characterize network resilience and allow us to understand strengths
and vulnerabilities against adversaries and unforeseen failures.
However, obtaining such characterizations is difficult because
resilience is a complex function of the operator's and adversary's
capabilities, knowledge, and resources, the topology of the network,
and the physical limitations on remedial and adversarial actions.
Motivated by these observations, this paper studies the relationship
between network resilience and structured topological perturbations,
with the ultimate goal of enabling the deployment of dynamic networks
with quantifiable resilience guarantees.

\subsubsection*{Literature review}
Network resilience, understood as the ability of the system to carry
out its goals under unexpected failures or malfunctions in its
components, is a rich research area. Multiple layers of network
activity are involved in ensuring resilience (e.g., detection of
failures or attacks, secure communications, injection of false data or
actuation signals), which naturally gets reflected in the variety of
disciplines employed in its study, e.g., computer
security~\cite{CPP-LSP:03,JK-YL:08},
communications~\cite{CK-RP-MS:02}, control~\cite{RP-PMF-RC:89}, and
signal processing~\cite{MB-IVN:93}.  In the context of distributed
control of network systems, the literature has studied diverse topics
including algorithms for function
computation~\cite{SS-CH:11,FP-AB-FB:12} and its robustness against
malicious behavior~\cite{HZ-EF-SS:15}, resilient estimation and
control~\cite{MZ-SM:14,HF-PT-SD:14,QU-DF-YHC-CJT:18}, attack
detection~\cite{FP-FD-FB:13}, and the scaling of robustness with
network size~\cite{TS-MR-MAD:19}, to name a few.


Here, we model the network dynamics using a linear time-invariant
system and analyze its stability properties in the face of additive
perturbations to the entries of the system matrix. This connected our
work with the study of matrix pseudospectra in linear
algebra~\cite{LT-ME:05} and, in particular, the sign of the
pseudospectral abscissa (the real part of the rightmost eigenvalue in
the pseudospectrum).  The work~\cite{NG-CL:13} shows that the
pseudospectral abscissa is associated with a low-rank perturbation to
the matrix.  The works~\cite{JB-AL-MO:03b} and~\cite{DL-BV:17} propose
criss-cross algorithms to numerically compute the value of the complex
and real, respectively, pseudospectral abscissa.  Both algorithms rely
on the method in~\cite{RB:88} to compute the distance to instability,
which is impractical for large-scale systems.  Instead, iterative
algorithms~\cite{NG-MO:11,MR:15} quickly approximate complex or real
pseudospectral abscissa of large matrices with sparse
structures. However, due to their gradient-based nature, these
algorithms possess local convergence properties and are only
guaranteed to yield good approximations of the pseudospectral abscissa
when the magnitude of the system perturbation is sufficiently small.
Another closely-related concept is that of stability radius, which is
the critical value of the magnitude of the perturbation that makes the
pseudospectral abscissa become 0. While the works \cite{NG-MO:11,MR:15,NG:16} propose effecient iterative algorithms for approximating stability radii of sparse matrices, 
the work~\cite{DH-AP:86} provides a formula for directly computing the stability radius
when the perturbation is an arbitrary complex matrix and its magnitude
is measured by its induced 2-norm.  A similar stability radius formula
is given in~\cite{LQ-BB-AR-ED-PY-JD:95} when the perturbation is a
real matrix.

In practice, perturbations to the system matrix might be constrained
by physical modeling, specific cyber vulnerabilities, or sparsity
patterns.  The study of the pseudospectral abcissa and the stability
radius when perturbations are structured is much more limited.  The
fact that structured pseudospectra are closely tied to structured
singular values~\cite{AP-JD:93}, which can be NP-hard to
compute~\cite{SP-JR:93}, explains why no general formula exists for
representing structured pseudospectra.  The
works~\cite{SR:06,MK-EK-DK:10,PB-NG-SN:12} study structured
pseudospectra for specific classes of perturbation matrices such as
Toeplitz, symmetric, Hankel, or circulant. ~\cite{SJ-MW-MZ-RD:18}
studies the problem of characterizing the smallest additive matrix
perturbation to an LTI system so that it loses controllability,
observability, or stability.  The work proposes an algorithm to obtain
a locally optimal perturbation based on the identification of
necessary conditions which, given the problem generality, are
expressed implicitly in terms of abstract linear maps.  Closely
related to our work,~\cite{VK-FP:20} considers bounded-energy
perturbations with sparse structure and studies the structured
stability radius.  The treatment relaxes the sparsity constraints by
incorporating them into a cost function and relies on increasingly
large penalties to satisfy them with increasing accuracy. This
increasing accuracy comes at the expense of reducing the convergence
speed of the proposed gradient-based algorithm, which makes it not
well suited for large-scale networks.  In contrast, our approach here
presents a general characterization of the worst-case structured
perturbation that includes the possibility of element-wise
constraints. This serves as the basis for our design of efficient
iterative algorithms to compute both the structured pseudospectral
abscissa and stability radius that are able to deal with large-scale
systems.


\subsubsection*{Statement of contributions}

We model the network as a linear-time invariant system and study the
effect on stability of additive perturbations to the system matrix of
bounded energy and subject to sparsity and element-wise saturation
constraints.  Our contributions address the questions of whether an
adversary can destabilize the network by employing such perturbations
and characterizing the maximum amount of energy in the perturbation
that the network can withstand without becoming unstable.  Our first
contribution is a novel necessary condition prescribing that the
worst-case structured perturbation to the network must solve an
implicit optimization problem. The implicitness arises because of the
dependence on the right and left eigenvectors associated to the
structured pseudospectral abscissa and, if they were known instead,
the optimization would become explicit and convex. We provide a
complete description of the solution to the explicit optimization,
show it is Lipschitz with respect to the problem parameters, and
provide an incremental method to compute it.  The observation about
the implicit character of the optimization problem is the basis for
our second contribution, which is an iterative algorithm alternating
between finding the right and left eigenvectors given an estimate of
the worst-case structured perturbation and solving the corresponding
explicit optimization problem to refine said estimate.  We show that
the proposed algorithm is guaranteed to converge to the structured
pseudospectral abscissa at a linear rate for sufficiently close
initial conditions. Our final contribution concerns the structured
stability radius, and builds on the fact that this radius corresponds
to the zero-crossing of the structured pseudospectral abscissa when
viewed as a function of the perturbation energy. We establish the
locally Lipschitzness of this function, explicitly compute its
gradient at the points of differentiability, and employ Newton's
method to design an iterative algorithm that provably approximates the
structured stability radius.  We illustrate in simulation the
efficiency and accuracy of the proposed algorithms on several network
examples, including a class of stable large-scale systems.

\section{Preliminaries}\label{sec:prelim}
Here, we introduce the notation and basic notions from linear algebra
used in the paper.

\subsubsection*{Notation}
For any vector $x\in\C^n$ or matrix $A\in\C^{n\times m}$, let $x^*$,
$A^*$ be their conjugate transpose. Let $|\cdot|$ be the 2-norm of
vectors in $\C^n$, that is, $|x|:=\sqrt{x^*x}$. In addition, let
$\Vert\cdot\Vert_2$, $\Vert\cdot\Vert_F$ be the induced 2-norm and
Frobenius norm, respectively, on $\C^{n\times n}$.
We let $\bB_r:=\{\Delta\in\C^{n\times n}:\Vert\Delta\Vert_F\leq r\}$
denote the closed ball of radius $r$ in $\C^{n\times n}$.
We say vectors $x, y \in \C^n$ are \emph{RP-compatible},
cf.~\cite{NG-MO:11}, if $|x| = |y| = 1$ and $x^*y$ is real and
positive. Given any $x$, $y$ with $x^*y \neq 0$, one can obtain a pair
of RP-compatible vectors $\tilde x$, $\tilde y$ by scaling $x$ and $y$
as follows
\[
\tilde x = \frac{x}{|x|}, \quad\tilde y=\frac{y^*x}{|y^*x|}
\frac{y}{|y|}.
\]
The inner product $\langle\cdot,\cdot\rangle:\C^{n\times n}\times
\C^{n\times n}\to \C$ of matrices $A=[a_{ij}]$ and $B=[b_{ij}]$ is $
\langle A, B\rangle: = \textnormal{Tr}(A^*B) =
\sum_{i=1}^n\sum_{j=1}^na^*_{ij}b_{ij}$. Note that $\langle A,
A\rangle=\Vert A\Vert_F^2$. In addition, for any
$x,y\in\C^n,M\in\R^{n\times n}$,
\begin{equation}\label{identity:Frobenius_inner_product}
  \re(x^*My)=\langle M, \re (xy^*)\rangle.
\end{equation}
The \emph{group inverse} of $A$, denoted $A^\#$, is the unique matrix
satisfying $AA^\# = A^\#A$, $A^\#A A^\# = A^\#$, and $A A^\#A = A$,
cf.~\cite{CM-GS:88}. The group inverse is different from the
Moore--Penrose pseudoinverse.
For a function $f:\R \to \R$, define the right derivative of $f$ at $x$ to be
\[
\partial_+f(x):=\lim_{\delta\to
      0^+}\frac{f(x+\delta)-f(x)}{\delta}.
\]
For functions $f,g:\R_{\geq 0} \to\R$, we denote $f(t)=O(g(t))$ if
there exists $k,\delta>0$ such that $|f(t)|\leq kg(t)$ for $t<\delta$.

\subsubsection*{Spectral abscissa and stability radius}
We denote by $\Lambda(A)$ the spectrum (i.e., set of eigenvalues) of a
square matrix $A$. The spectral abscissa of~$A$ is
\begin{equation}\label{def:abscissa}
  \alpha(A):=\max_{\lambda\in\Lambda(A)}\re\lambda.
\end{equation}
We refer to a maximizer $\lambda_{\opt}$ of this function as a
rightmost eigenvalue of $\Lambda(A)$.  For $\epsilon>0$ and a closed
set $\nH\subseteq \C^{n\times n}$, the \emph{structured $\epsilon$-
  pseudospectrum} of $A$ (with respect to the Frobenius norm) is
\begin{equation}
  \Lambda_{\epsilon,\nH}(A):=\{\lambda\in\C:\lambda\in\Lambda(A+\Delta)\mbox{
    for }\Delta\in\nH\cap\bB_\epsilon\}. 
\end{equation}
Note that when $\nH=\C^{n\times n}$ or $\nH=\R^{n\times n}$,
$\Lambda_{\epsilon,\nH}$ reduces to the usual
$\epsilon$-pseudospectrum~\cite{NG-MO:11} or real
$\epsilon$-pseudospectrum~\cite{NG-CL:13}, respectively. Similar to
the spectral abscissa, we also define $\alpha_{\epsilon,\nH}(A)$ as
the \emph{structured $\epsilon$-pseudospectral abscissa} of~$A$,
\begin{equation}\label{def:structured_pseudospectral_abscissa}
  \alpha_{\epsilon,\nH}(A):=\max_{\lambda\in \Lambda_{\epsilon,\nH}(A)}\re\lambda.
\end{equation}
We refer to a maximizer $\lambda_{\opt}$ of this function as a
rightmost eigenvalue of $\Lambda_{\epsilon,\nH}(A)$.
Using~\eqref{def:abscissa}, one can equivalently express the
structured $\epsilon$-pseudospectral abscissa of $A$ as
\begin{equation}\label{def:structured_pseudospectral_abscissa_2}
    \alpha_{\epsilon,\nH}(A)=\max_{\Delta\in\nH\cap\bB_\epsilon} \alpha(A+\Delta).
\end{equation}
We refer to a maximizer $\Delta_{\opt}$ of this function as a
worst-case structured perturbation of energy~$\epsilon$.  The
\emph{structured stability radius} of $A$ is
\begin{equation}\label{def:cosntrained_stability_radius}
  r_{\nH}(A):= \min_{\epsilon\geq 0}\{\epsilon:\alpha_{\epsilon,\nH}(A)\geq 0\}.
\end{equation}
Clearly, if $A$ is not Hurwitz, $r_{\nH}(A)=0$. When $\nH=\C^{n\times
  n}$ or $\nH=\R^{n\times n}$, $r_{\nH}(A)$ coincides with the
definition of stability radius~\cite{DH-AP:86} or real stability
radius~\cite{LQ-BB-AR-ED-PY-JD:95}, respectively. Note also that if
$\nH_1\subseteq\nH_2$, then $r_{\nH_1}(A)\geq r_{\nH_2}(A)$.

\subsubsection*{Matrix perturbation theory}

Here we gather two useful results on matrix perturbations. The first
describes the derivative of a simple eigenvalue of a matrix that
depends linearly on time.

\begin{lem}[{\cite[Lemma 6.3.10 and Theorem
    6.3.12]{RAH-CRJ:85}}]\label{lem:eigenvalue_perturbation}
  Consider a $n \times n$ matrix trajectory $\R \ni t \mapsto C(t) =
  C_0 + tC_1$.  Let $\lambda(t)$ be an eigenvalue of $C(t)$ converging
  to a simple eigenvalue $\lambda_0$ of $C_0$ as $t \to 0$.  Let $x_0$
  and $y_0$ be, respectively, right and left eigenvectors of $C_0$
  corresponding to $\lambda_0$, that is, $(C_0-\lambda_0I)x_0 = 0$ and
  $y^*_0(C_0-\lambda_0I)=0$. Then $y^*_0x_0 \neq 0$ and $\lambda(t)$
  is analytic near $t = 0$ with
  \begin{equation}
    \left.\frac{d\lambda(t)}{dt}\right\vert_{t=0}=\frac{y^*_0C_1x_0}{y^*_0x_0}.
  \end{equation}
\end{lem}


The following result describes the time derivative of the product of
the right and left eigenvectors corresponding to a simple eigenvalue
of a complex-valued time-dependent matrix.

\begin{thm}[{\cite[Theorem 5.2]{NG-MO:11}}]\label{thm:derivative of
    xy*}
  Consider a $n\times n$ complex-analytic matrix trajectory $\C \ni t
  \mapsto C(t) = C_0 + tC_1+O(t^2)$.  Let $\lambda(t)$ be a simple
  eigenvalue of $C(t)$ in a neighborhood $\mathcal N$ of $t = 0$, with
  corresponding RP-compatible right and left eigenvectors $x(t)$ and
  $y(t)$. Then $Q(t)=x(t)y(t)^*$ is $C^\infty$ on~$\mathcal N$ and its
  derivative at~$t = 0$ is
  \begin{equation}\label{eqn:thm:derivative of xy*}
    \frac{dQ(t)}{dt}\Big|_{t=0} = \re(\beta + \gamma)Q(0) - GC_1Q(0) - Q(0)C_1G,
  \end{equation}
  where $G = (C_0 - \lambda I)^\#$, $\beta = x^*GC_1x$, $\gamma =
  y^*C_1Gy$, $\lambda = \lambda(0)$, $x = x(0)$, and $y = y(0)$.
\end{thm}

\section{Problem statement}\label{subsec:problem_formulation}

We provide here a formal mathematical description of the problem of
interest.  Let $\G:=\{\V,\E\}$ denote a \emph{network graph}, where
$\V = \{1,\dots, n\}$ is the set of nodes and $\E\subseteq \V\times\V $
is the set of edges. The network dynamics is described by a linear
differential equation
\begin{equation}\label{def:sys}
  \dot x= Ax,
\end{equation}
where the components of $x\in\R^n$ correspond to the states of the
nodes, and $A=[a_{ij}]\in\R^{n\times n}$, with $a_{ij}=0$ for all
$(i,j)\not\in\E$, is the weighted adjacency matrix. We assume the
matrix $A$ is Hurwitz.  An adversary seeks to destabilize the dynamics
by attacking the network interconnections. Such attacks are
structured, in the sense that the adversary is limited to perturbing
only certain edges and within some budget.  Formally, let $\E_p\subseteq
\E$, denote the \emph{perturbation edge set} and for all
$(i,j)\in\E_p$, let $\underline\Delta_{ij}\in\R_{\leq
  0}\cup\{-\infty\}$, $\overline\Delta_{ij} \in \R_{\geq 0}\cup\{+\infty\}$
be parameters specifying \emph{saturation constraints}. We denote the
set of allowed perturbations whose sparsity pattern is compatible with
$\E_p$ by
\begin{multline}\label{def:nH_a}
  \nH := \{\Delta=[\Delta_{ij}] \in \R^{n\times n}: \Delta_{ij} = 0
  \mbox{ if }(i, j)\not\in\E_p,
  \\
  \Delta_{ij}\in[\underline\Delta_{ij},\overline\Delta_{ij}]\mbox{ if
  }(i,j)\in\E_p\} .
\end{multline}
Note that we always have $0\in\nH$. When
$\underline\Delta_{ij}=-\infty$ or $\overline\Delta_{ij}=+\infty$, then
there is no lower or upper bound constraint on the perturbation size
of the corresponding edge. We define the energy of an attack $\Delta$
to be its Frobenius norm.

After an attack $\Delta\in\nH$, the network dynamics changes to
\begin{equation}\label{def:attacked_sys}
  \dot x= (A+\Delta)x.
\end{equation}
We are interested in answering the following questions:
\begin{enumerate}
\item Given the network dynamics~\eqref{def:sys}, can an adversary
  destabilize it by employing perturbations of bounded
  energy~$\epsilon>0$ in $\nH$?
\item What edges are most important to protect against perturbations
  of bounded energy~$\epsilon>0$ in order to preserve stability?
\item If by resilience of the network dynamics we understand the
  maximum amount of energy in the perturbation that it can withstand
  without becoming unstable, what is the network resilience against
  the adversary?
\end{enumerate}
Each of these questions can be transcribed into a mathematically
precise statement. In fact, keeping in mind
that~\eqref{def:attacked_sys} is GAS if $A+\Delta$ is Hurwitz, we can
equivalently say that
\begin{itemize}
\item question 1) refers to determining whether the structured
  $\epsilon$-pseudospectral abscissa $\alpha_{\epsilon,\nH}(A)$ of $A$
  is positive;
\item with regards to question 2), assume the worst-case structured
  perturbation $\Delta_{\opt}(\epsilon)$ of energy~$\epsilon$ is
  unique.  Notice that the larger $|(\Delta_{\opt})_{ij}(\epsilon)|$
  is, the larger proportion of weight modification is done on the edge
  $(i,j)$ to destabilize the system. Hence, the magnitude of the
  elements in the perturbation matrix provides an ordering of the
  relative importance of edges against the attack; and
\item question 3) refers to determining the value of the structured
  stability radius $r_{\nH}(A)$.
\end{itemize}

In our ensuing discussion, we address questions 1) and 2) concurrently
by introducing an iterative algorithm that finds both the value of
$\alpha_{\epsilon,\nH}(A)$ and a maximizer $\Delta_{\opt}$
of~\eqref{def:structured_pseudospectral_abscissa_2}. We then answer
question 3) by designing another iterative algorithm which finds the
value of~$r_{\nH}(A)$.

\section{Network stability against perturbations: structured
  pseudospectral abscissa}\label{sec:main}

In this section, we study the answers to questions 1) and 2) of our
problem statement. We begin by characterizing the worst-case
structured perturbation of a given energy. We build on this
characterization later to propose an algorithm that computes
iteratively the structured pseudospectral abscissa.

\subsection{Characterization of the worst-case structured
  perturbation}\label{subsec:optimality}

Our first result states that any worst-case structured perturbation
which gives rise to a simple rightmost eigenvalue is a solution to an
implicit optimization problem.

\begin{thm}[First-order necessary condition for
  optimality]\label{thm:1}
  Let $\nH' \subseteq\C^{n\times n}$ be compact and convex, and let
  \begin{equation}\label{def:Delta_opt}
    \Delta_{\opt} \in \arg\max_{\Delta\in \nH'}\alpha(A+\Delta).    
  \end{equation}
  Let $\lambda_{\opt}(A+\Delta_{\opt})$ be a rightmost eigenvalue of
  $A+\Delta_{\opt}$ and suppose it is simple.  Then $\Delta_{\opt}$
  must satisfy
  \begin{equation}\label{eqn:1st_order_necessary_condition}
    \Delta_{\opt} \in \arg\max_{\Delta\in \nH'}\langle\Delta,\re(yx^*)\rangle,
  \end{equation}
  where $x,y \in \C^n$ are the RP-compatible right and left
  eigenvectors associated with $\lambda_{\opt}(A+\Delta_{\opt})$.
\end{thm}
\begin{proof}
  To study the first-order necessary condition of optimality, we
  consider the set $\mathcal F$ of feasible directions at $\Delta \in
  \nH'$,
  \[
  \mathcal F:=\{C \in\C^{n\times n}:\exists \tau>0 \mbox{ s.t. }
  \Delta+t C\in \nH' , \; \forall t\in[0,\tau]\}.
  \]
  The condition for optimality states that, if $\Delta_{\opt}$ is the
  optimizer, then
  $\frac{d}{dt} \alpha(A+\Delta_{\opt}+tC) \big|_{t=0}\leq 0$, for all
  $C\in \mathcal F$. Applying Lemma~\ref{lem:eigenvalue_perturbation},
  we deduce
  \begin{multline*}
    0\geq
    \re\left(\frac{d\lambda_{\opt}(A+\Delta_{\opt}+tC)}{dt}\Big|_{t=0}\right)
    \\
    =\re\left(\frac{y^*Cx}{y^*x}\right)=\frac{\re(y^*Cx)}{y^*x}
    =\frac{\langle C,\re(yx^*)\rangle}{y^*x},
  \end{multline*}
  where we have used the RP-compatibility for the second last equality
  and the identity \eqref{identity:Frobenius_inner_product} for the
  last equality. Now, one can see that $\langle C,\re(yx^*)\rangle\leq
  0$ for all $C\in \mathcal F$ corresponds indeed to the first-order
  necessary condition for optimality of the maximization
  problem~\eqref{eqn:1st_order_necessary_condition}. Since this
  problem is convex, the condition is also sufficient to characterize
  an optimizer. Thus, satisfying the condition
  of~\eqref{eqn:1st_order_necessary_condition} is a necessary
  condition for being a maximizer of~\eqref{def:Delta_opt}.
\end{proof}

It is worth pointing out that,
since~\eqref{eqn:1st_order_necessary_condition} is a necessary
condition, it needs to be satisfied by any worst-case perturbation. We
make use of this fact to design an algorithm based on
Theorem~\ref{thm:1} to find a worst-case perturbation and compute the
structured pseudospectral abscissa.

Comparing with \eqref{def:structured_pseudospectral_abscissa_2}, we
observe that \eqref{def:Delta_opt} with $\nH'= \nH \cap\bB_\epsilon$
corresponds exactly to a worst-case structured perturbation of
energy~$\epsilon$ (in fact, $\nH$ as defined in \eqref{def:nH_a} is
closed and convex, and hence $\nH '$ is compact and convex). The
optimization problem in~\eqref{def:Delta_opt} is nonconvex and hence
difficult to solve in general.  Instead, Theorem~\ref{thm:1}
facilitates finding the structured pseudospectral abscissa by
providing a characterization~\eqref{eqn:1st_order_necessary_condition}
of the worst-case structured perturbations.

For known $x,y\in\C^n$, the optimization
in~\eqref{eqn:1st_order_necessary_condition} is a convex problem of
the form
\begin{alignat}{2}\label{Optimization_problem}
  & \textnormal{maximize } && \quad \langle\Delta,M\rangle
  \\
  & \textnormal{subject to } && \quad \Delta\in \nH\cap\bB_\epsilon ,
\end{alignat}
when we set $M=[m_{ij}]:=\re(yx^*)$, and can therefore be solved
efficiently.  However, we should note that the vectors $x$ and $y$
in~\eqref{eqn:1st_order_necessary_condition} are unknown, since they
are the eigenvectors of $A+\Delta_{\opt}$, making
equation~\eqref{eqn:1st_order_necessary_condition} implicit
in~$\Delta_{\opt}$.  Before we address this issue, we finish the
exposition here describing the properties of the solution
of~\eqref{Optimization_problem} for a given known~$M$.

We make the next assumption regarding the worst-case
perturbation.

\begin{assmpt}\longthmtitle{Non-saturation at optimizers}\label{ass:1}
  No optimizer of~\eqref{Optimization_problem} is fully saturated,
  i.e., if $\Delta_{\opt}$ is an optimizer
  of~\eqref{Optimization_problem}, then there exists $(i,j)\in\E_p$
  such that $(\Delta_{\opt})_{ij}\in(\underline \Delta_{ij},\overline
  \Delta_{ij})$. Furthermore,~$m_{ij}\neq 0$. 
\end{assmpt}

Assumption~\ref{ass:1} is reasonable: in case it does not hold, i.e.,
a worst-case perturbation is fully saturated, then it must be at a
vertex of $\nH$ and since this constraint set has finitely many
vertices, the worst-case perturbation can be found by exhaustion.
Meanwhile, if $m_{ij}=0$, then the value of $(\Delta_{\opt})_{ij}$
does not affect the optimal value and, consequently, one can construct
other optimizers that are saturated at edge $(i,j)$. This is the
reason why we explicitly require $m_{ij}\neq 0$ in
Assumption~\ref{ass:1}. Since \eqref{Optimization_problem} has a
non-trivial linear objective function with a convex constraint set,
the optimum is achieved on its boundary and hence, if the maximizer
$\Delta_{\opt}$ is not fully saturated, it must verify $\Vert
\Delta_{\opt}\Vert_F=\epsilon$.  In addition, by employing the KKT
conditions for optimality, we are able to characterize the solution
of~\eqref{Optimization_problem}.

\begin{prop}\longthmtitle{Characterization of solution
    of~\eqref{Optimization_problem}}\label{prop:solution_to_optimization}
  Let $M \in \C^{n\times n}$ and $\epsilon>0$.  Under
  Assumption~\ref{ass:1}, the
  optimization~\eqref{Optimization_problem} has a unique optimizer
  $\Delta_{\opt}$ given by
  \begin{equation}\label{assignment_of_Delta_opt}
    (\Delta_{\opt})_{ij}=\left\{\begin{array}{cc} m_{ij}\theta_{\opt}, &
        \mbox{if } (i,j)\in \E_p\backslash(\overline\nS\cup\underline\nS),
        \\
        \overline \Delta_{ij}, & \mbox{if } (i,j)\in\overline\nS,
        \\
        \underline \Delta_{ij}, & \mbox{if } (i,j)\in\underline\nS,
        \\
        0,&\mbox{if }(i,j)\not\in\E_p
      \end{array}\right.
  \end{equation}
  where
  $\overline\nS:=\overline\nS(\epsilon,M),\underline\nS:=\underline\nS(\epsilon,M)$
  are the unique subsets of $\E_p$ such that
  \begin{subequations}\label{def:saturation_set}
    \begin{align}
      m_{ij}\theta_{\opt}\in (\underline \Delta_{ij},\overline \Delta_{ij}) &
      \quad\forall (i,j)\in\E_p\backslash(\overline\nS\cup\underline\nS),
      \label{saturation_set_1}
      \\
      m_{ij}\theta_{\opt}\geq\overline
      \Delta_{ij} & \quad\forall (i,j)\in
      \overline\nS,\label{saturation_set_2}
      \\
      m_{ij}\theta_{\opt}\leq\underline \Delta_{ij} & \quad\forall
      (i,j)\in \underline\nS\label{saturation_set_3}
    \end{align}
  \end{subequations}
  and $\theta_{\opt}$ is shorthand notation for
  $\theta_{\opt}(\epsilon,M) :=
  \theta(\epsilon,M,\overline\nS(\epsilon,M),\underline\nS(\epsilon,M))$, where
  the function $\theta$ is
  \begin{equation}\label{def:theta}
    \theta(\epsilon,M,\overline\nS,\underline\nS) \! := \!
    \sqrt{\frac{\epsilon^2 -
        \sum_{(i,j)\in\overline\nS}\overline\Delta_{ij}^2-\sum_{(i,j)\in\underline\nS}
        \underline
        \Delta_{ij}^2}{\sum_{(i,j)\in\E_p\backslash(\overline\nS\cup\underline\nS)}m_{ij}^2}}.  
  \end{equation}
  In addition, there is a neighborhood $D$ around
  $(\epsilon,M)\in\R_{\geq 0}\times \R^{n\times n}$ such that
  $\theta_{\opt}$ is Lipschitz on~$D$.
\end{prop}

The proof of Proposition~\ref{prop:solution_to_optimization} is in the
Appendix.  According to this result, the element of the optimizer
$\Delta_{\opt}$ corresponding to $(i,j)\in \E_p$ is either saturated
or proportional to $m_{ij}$, with ratio given by $\theta_{\opt}$.  We
refer to $\overline\nS,\underline\nS$ as the \emph{index sets of
  saturation} as $(\Delta_{\opt})_{ij}$ attains either of its boundary
values for all $(i,j)\in\overline\nS\cup\underline\nS$.  Note that, by
Assumption~\ref{ass:1}, $\nS(\epsilon,M)\subsetneq\E_p$ and
$\theta_{\opt}$ given by~\eqref{def:theta} is well defined since there
exists $(i,j)\in \E_p\backslash(\overline\nS\cup\underline\nS)$ such that
$m_{ij}\neq 0$.

\begin{rem}\longthmtitle{Comparison with the literature}
  Proposition~\ref{prop:solution_to_optimization} is a generalization
  of the results available in the
  literature~\cite{NG-CL:13,VK-FP:20}. When there are neither sparsity
  constraints nor saturation constraints, $\nH=\R^{n\times n}$ and
  $\theta_{\opt}=\frac{\epsilon}{\Vert M\Vert_F}$, so
  \begin{align*}
    \Delta_{\opt}=\frac{\epsilon \re(yx^*)}{\Vert \re(yx^*)\Vert_F} ,
  \end{align*}
  as stated in \cite[Theorem 2.2]{NG-CL:13}. On the other hand, when
  $\nH$ contains sparsity constraints but no saturation constraints,
  we deduce that there exists $c\geq 0$ such that
  \begin{align*}
    (\Delta_{\opt})_{ij}=\left\{\begin{array}{cc}
        c\re(yx^*), & \mbox{ if } (i,j)\in\E_p, \\
        0, & \mbox{ otherwise, }
      \end{array}\right.
  \end{align*}
  as stated in~\cite[Theorem 3.2]{VK-FP:20}.  \oprocend
\end{rem}

We note that $\theta_{\opt}$ and $\overline\nS$, $\underline\nS$ are
inter-dependent, which means that
Proposition~\ref{prop:solution_to_optimization} does not provide an
explicit expression of $\Delta_{\opt}$. However, the result provides
the basis for a simple method, which we summarize in
Algorithm~\ref{alg:solving_optimization_problem}, to find the solution
of \eqref{Optimization_problem} by growing the index sets of
saturation $\overline\nS,\underline\nS$ if they do not meet the
conditions~\eqref{saturation_set_1}--\eqref{saturation_set_3} for the
corresponding value of $\theta_{\opt}$ determined
by~\eqref{def:theta}.
\begin{algorithm}[h]
  \caption{Incremental construction of index
    sets}\label{alg:solving_optimization_problem}
  \algorithmicrequire $\epsilon,M,\nH$
  \\
  \algorithmicensure $\theta_{\opt},\Delta_{\opt}$
  \begin{algorithmic}[1]
    \State $\overline\nS\gets\emptyset, \underline\nS\gets\emptyset$,
    \State NotDone $\gets$ false\label{Step:clearflag}
    \If{$\overline\nS\cup\underline\nS=\E_p$} break with error ``the optimizer is fully saturated, Assumption~\ref{ass:1} is violated"
    \EndIf
    \State Compute $\theta_{\opt}$ as in \eqref{def:theta}
    \ForAll{$(i,j)\in\E_p\backslash(\overline\nS\cup\underline\nS)$}
        \If{$m_{ij}\theta_{\opt}\geq\overline\Delta_{ij}$} 
        \State $\overline\nS\gets\overline\nS\cup\{(i,j)\}$
        \State $(\Delta_{\opt})_{ij}\gets\overline\Delta_{ij}$
        \State NotDone $\gets$ true
        \ElsIf{$m_{ij}\theta_{\opt}\leq\underline\Delta_{ij}$}
        \State $\underline\nS\gets\underline\nS\cup\{(i,j)\}$
        \State $(\Delta_{\opt})_{ij}\gets\underline\Delta_{ij}$
        \State NotDone $\gets$ true
        \EndIf
    \EndFor
    \If{NotDone $=$ true} go back to Step~\ref{Step:clearflag}
    \Else{} $(\Delta_{\opt})_{ij}\gets m_{ij}\theta_{\opt}$ for all $(i,j)\in\E_p\backslash(\overline\nS\cup\underline\nS)$
    \EndIf
  \end{algorithmic}
\end{algorithm}
The next result shows that
Algorithm~\ref{alg:solving_optimization_problem} finds the solution
of~\eqref{Optimization_problem}.

\begin{lem}\longthmtitle{Algorithm~\ref{alg:solving_optimization_problem}
    solves~\eqref{Optimization_problem}}\label{lem:why_optimization_alg_works}
  Under Assumption~\ref{ass:1},
  Algorithm~\ref{alg:solving_optimization_problem} finds the solution
  of the optimization~\eqref{Optimization_problem} in a finite number
  of steps.
\end{lem}

The proof of Lemma~\ref{lem:why_optimization_alg_works} is provided in
the Appendix.  In contrast to generic convex optimization solvers,
Algorithm~\ref{alg:solving_optimization_problem} is tailored to
problem~\eqref{Optimization_problem} and takes advantage of the
characterization~\eqref{assignment_of_Delta_opt} of its optimizer.  We
use later the ratio $\theta_{\opt}$ in
Algorithm~\ref{alg:solving_optimization_problem} to compute the
structured stability radius.

We conclude this section by presenting a result which shows that the
optimizer of~\eqref{Optimization_problem} is locally Lipschitz when
viewed as a function of the matrix defining the objective function. To
establish this, we use the Frobenius norm and show that the Lipschitz
constant is proportional to the energy of matrix perturbations.  The
proof is given in the Appendix.

\begin{lem}\longthmtitle{Sensitivity of the optimizer of
    \eqref{Optimization_problem} with respect to
    parameters}\label{lem:continuous_dependence}
  Let $\epsilon>0,M_1,M_2\in\R^{n\times n}$ and suppose $\Delta_{k}$
  are the optimizers of~\eqref{Optimization_problem} with parameters
  $(\epsilon,M_k),k=1,2$. Also assume that Assumption~\ref{ass:1}
  holds for $(\epsilon,M_1)$. Then, there exist $\delta=\delta(M_1)>0,
  \ell=\ell(M_1)>0$ such that
  %
  %
  \begin{align*}
    \Vert \Delta_{1}-\Delta_{2} \Vert_F\leq \ell\epsilon\Vert
     M_1-M_2\Vert_F
  \end{align*}
  as long as $\Vert M_1-M_2\Vert_F\leq \delta$.
\end{lem}

\subsection{Iterative computation of structured pseudospectral
  abscissa}\label{subsec:alg_pa}

What we have unveiled about the optimization
problem~\eqref{Optimization_problem} and the structure of its solution
in Section~\ref{subsec:optimality} is not directly applicable to the
determination of the worst-case structured perturbation and the
structured pseudospectral abscissa.  This is because, as we mentioned
earlier, the
characterization~\eqref{eqn:1st_order_necessary_condition} is implicit
in $\Delta_{\opt}$, i.e., the matrix $M =\re(yx^*)$ required to set
up~\eqref{Optimization_problem} is not a priori known, and in fact
depends on the optimizer itself. To address this obstacle, we propose
in Algorithm~\ref{alg:pseudospectral_abscissa} an iterative strategy
that proceeds by repeatedly solving instances of
problem~\eqref{Optimization_problem}, in each case taking the right
and left eigenvectors corresponding to the previous iterate.

\begin{algorithm}[h]
  \caption{Computation of worst-case
    perturbation}\label{alg:pseudospectral_abscissa}
  \algorithmicrequire
  $A,\epsilon,\nH,\Delta_0,\tol_{\Delta}$
  \\
  \algorithmicensure $\Delta,\alpha, x, y,\theta$
  \begin{algorithmic}[1]
    \State $\nH '\gets \nH\cap\bB_\epsilon$\label{alg:spa_H'}
    \State $A_0\gets A+\Delta_0$
    \Repeat{ $k=0,1,\cdots$}
    \State Compute a rightmost eigenvalue $\lambda_{k}$ of $A_{k}$ and
    its corresponding RP-compatible right and left eigenvectors
    $x_{k},y_{k}$\label{alg:spa_xy}
    \State $M\gets\re(y_kx_k^*)$
    \State Run Algorithm~\ref{alg:solving_optimization_problem} with
    inputs $\epsilon, M,\nH$, set the outputs $\theta_{k+1},\Delta_{k+1}$
    %
    \State $A_{k+1}\gets A+\Delta_{k+1}$
    \Until{$\Vert\Delta_{k+1}-\Delta_{k}\Vert_2\leq
      \tol_{\Delta}$}
    \State $(\Delta, \alpha, x,y,\theta) \gets (\Delta_{k+1},
    \re\lambda_{k+1},  x_{k+1},  y_{k+1},\theta_{k+1})$
  \end{algorithmic}
\end{algorithm}

The logic of Algorithm~\ref{alg:pseudospectral_abscissa} can be
described as follows.  At each step, we consider a candidate
worst-case perturbation $\Delta_k$, followed by computing the
RP-compatible right and left eigenvalues $x_k,y_k$ of a rightmost
eigenvalue of $A+\Delta_k$. We then solve the optimization
problem~\eqref{Optimization_problem} using $M=\re(y_kx_k^*)$, and set
the new optimizer to be $\Delta_{k+1}$. This process is repeated until
the sequence of possible worst-case perturbations converges.  From its
design, it is clear that a fixed point of
Algorithm~\ref{alg:pseudospectral_abscissa} is a solution to the
maximization problem~\eqref{eqn:1st_order_necessary_condition} (and
hence, by Theorem~\ref{thm:1}, satisfies the first-order necessary
condition for being the maximizer of~\eqref{def:Delta_opt}).

The next result establishes the local asymptotic convergence of
Algorithm~\ref{alg:pseudospectral_abscissa} to the structured
pseudospectral abscissa.

\begin{thm}[Local convergence of
  Algorithm~\ref{alg:pseudospectral_abscissa}]\label{thm:local_convergence}
  Let $\Delta_{\opt}$ be a worst-case structured perturbation (i.e.,
  $\Delta_{\opt}$ satisfies \eqref{def:Delta_opt} with $\nH'=
  \nH\cap\bB_\epsilon$) and assume the rightmost eigenvalue $\lambda$
  of $A+\Delta_{\opt}$ is simple, with RP-compatible right and left
  eigenvector pair $x$, $y$.  Let Assumption~\ref{ass:1} hold for
  $\Delta_{\opt}$ and $M=\re(yx^*)$. Define
  \begin{equation}\label{def:r}
    r = \frac{4\sqrt{n}\ell\epsilon}{\sigma_{n-1}(A-\lambda
      I)(y^*x)^2}  ,
  \end{equation}
  where $\ell=\ell(\re(yx^*))$ is given in Lemma~\ref{lem:continuous_dependence} and
  $\sigma_{n-1}$ denotes the second smallest singular value. Let
  $\Delta_k,\re\lambda_k$ be the sequences generated by
  Algorithm~\ref{alg:pseudospectral_abscissa}. If $r < 1$,
  $r^\dagger\in(r,1)$ is arbitrary and
  $\Vert\Delta_0-\Delta_{\opt}\Vert_2$ is sufficiently small, then 
  \begin{equation}\label{convegence_of_Delta_k}
    \Vert \Delta_k-\Delta_{\opt}\Vert_2 \leq (r^\dagger)^k\Vert \Delta_0-\Delta_{\opt}\Vert_2
  \end{equation}
  for all $k=0,1,\dots$ and $\re\lambda_k$ converges to the
  structured pseudospectral abscissa $\alpha_{\epsilon,\nH}(A)$. In
  addition, the output of Algorithm~\ref{alg:pseudospectral_abscissa}
  satisfies
  \begin{equation}\label{eqn:main_thm_last}
    |\alpha-\alpha_{\epsilon,\nH}(A)| = \frac{\sqrt n
      r^\dagger\Vert\re(yx^*)\Vert_F}{(1-r^\dagger)y^*x}\tol_{\Delta}+
    O(\tol_{\Delta}^2). 
  \end{equation}
\end{thm}
\begin{proof}
  Let $L:=yx^*$ and $L_{k}:=y_{k}x_{k}^*$ for $k=0,1,\dots$, where
  $x_k,y_k$ come from Step~\ref{alg:spa_xy} in the $k$-th iteration of
  Algorithm~\ref{alg:pseudospectral_abscissa}.  In addition, define
    \begin{align*}
      E_{k}:=\Delta_{k}-\Delta_{\opt}, 
      \quad
      F_{k}:=L_k-L 
    \end{align*}
    We first find a relation between $\Vert F_k\Vert_2$ and $\Vert
    E_k\Vert_2$. Consider the matrix trajectory
  \begin{equation}\label{matrix_trajectory}
      C(t)=A+\Delta_{\opt}+t\frac{E_k}{\Vert E_k\Vert_2}.
  \end{equation}
  Let $x(t)$, $y(t)$ be RP-compatible right and left
  eigenvectors of $C(t)$ associated with its rightmost eigenvalue
  $\lambda(t)$ such that $\lambda(0)=0$, $x(0)=x$ and $y(0)=y$. Define
  $Q(t):=y(t)x(t)^*$.  Invoking Theorem~\ref{thm:derivative of xy*}
  and taking the conjugate of~\eqref{eqn:thm:derivative of xy*},
  \begin{align*}
    \frac{dQ(t)}{dt}\Big|_{t=0} =
    \re(\beta+\gamma)Q(0)-Q(0)C_1^*G^*-G^*C_1^*Q(0),
  \end{align*}
  where $G = (A+\Delta_{\opt} - \lambda I)^\#$, $\beta = x^*GC_1x$
  with $C_1=E_k / \Vert E_k\Vert_2$ and $\gamma = y^*C_1Gy$. Note
  that, since $C(0)=A+\Delta_{\opt}$ and $C(\Vert
  E_k\Vert_2)=A+\Delta_k$, we have $Q(0)=L$ and $Q(\Vert
  E_k\Vert_2)=L_k$. Therefore using Taylor's expansion, we have
  \begin{align}\label{before_norm}
    F_k& =Q(\Vert E_k\Vert_2)-Q(0)=\frac{dQ(t)}{dt}\Big|_{t=0}\Vert
    E_k\Vert_2+R(\Vert E_k\Vert_2) \notag
    \\
    & =\re(x^*GE_kx+y^*E_kGy)L-LE_k^*G^*-G^*E_k^*L 
    \notag
    \\
    & \quad  +R(\Vert E_k\Vert_2) ,
  \end{align}
  where $R(\Vert E_k\Vert_2)$ is the Taylor remainder of order $1$
  such that $\Vert R(t)\Vert_2=O(t^2)$, which implies the existence of
  a class $\mathcal K$ function $\xi$ such that $\Vert R(t)\Vert_2\leq
  st$ for any $s>0$ and $t\in[0,\xi(s)]$.  Now notice that because
  $x,y$ are RP-compatible and $L=yx^*$, $|x|=|y|=1$ and $\Vert L
  \Vert_2\leq \Vert L \Vert_F=1$. In addition, from \cite[Theorem
  5.5]{NG-MO:11}, $\Vert G\Vert_2\leq \frac{1}{\sigma_{n-1}(A-\lambda
    I)(y^*x)^2}$. Hence it follows from \eqref{before_norm} that
  \begin{align}
    \Vert F_k\Vert_2&\leq 4\Vert G\Vert_2\Vert E_k\Vert_2\Vert L\Vert_2+ \Vert R(\Vert
    E_k\Vert_2)\Vert_2\nonumber\\
    &\leq\frac{4\Vert E_k\Vert_2}{\sigma_{n-1}(A-\lambda
      I)(y^*x)^2}+\Vert R(\Vert
    E_k\Vert_2)\Vert_2\nonumber\\
    &=\frac{r}{\sqrt n \ell \epsilon}\Vert E_k\Vert_2+\Vert R(\Vert
    E_k\Vert_2)\Vert_2\label{norm_F_k}
  \end{align}
  where we have used the definition~\eqref{def:r} of $r$ in the last equality. 
  
  We are ready to establish~\eqref{convegence_of_Delta_k} by
  induction. Define $r_1:=\frac{r^\dagger-r}{\sqrt n
    \ell\epsilon}$ and let $\Vert E_0\Vert_2<\min\big\{\frac{\sqrt
    n\ell\epsilon}{r^\dagger}\delta,\xi(r_1)\big\}$, where $\delta$ is
  given in Lemma~\ref{lem:continuous_dependence}.
  The base case $k=0$ for \eqref{convegence_of_Delta_k} is trivially
  true. Since $r^\dagger<1$, the induction assumption for index $k$
  implies $\Vert E_k\Vert_2<\min\big\{\frac{\sqrt
    n\ell\epsilon}{r^\dagger}\delta,\xi(r_1)\big\} $ as well. Hence,
  from~\eqref{norm_F_k},
  \begin{equation}\label{F<E}
    \Vert F_k\Vert_2 \leq \big(\frac{r}{\sqrt
      n \ell\epsilon}+r_1\big)\Vert E_k\Vert_2=\frac{r^\dagger}{\sqrt n
      \ell\epsilon}\Vert E_k\Vert_2<\delta.     
  \end{equation}
  Using Lemma~\ref{lem:continuous_dependence}, we deduce
  \begin{multline*}
    \Vert E_{k+1}\Vert_2\leq
    \Vert E_{k+1}\Vert_F=\Vert\Delta_{k+1}-\Delta_{\opt}\Vert_F
    \\
    \leq \ell\epsilon\Vert\re(y_kx_k^*)-\re(yx^*)\Vert_F\leq \ell\epsilon
    \Vert F_k \Vert_F\leq \sqrt{n}\ell\epsilon \Vert F_k \Vert_2,
  \end{multline*} 
  where we have used $\Vert M\Vert_2\leq \Vert M\Vert_F\leq
  \sqrt{n}\Vert M\Vert _2$.  Combining this bound with the first
  inequality in~\eqref{F<E} and using the induction hypothesis,
  \begin{equation}\label{recursive_E}
    \Vert E_{k+1}\Vert_2\leq  r^\dagger\Vert E_{k}\Vert_2\leq
    (r^\dagger)^{k+1}\Vert E_0\Vert_2  ,
  \end{equation}
  establishing~\eqref{convegence_of_Delta_k}.  It follows that
  $\alpha(A+\Delta_k)=\re \lambda_k$ converges to
  $\alpha(A+\Delta_{\opt})=\alpha_{\epsilon,\nH}(A)$. To prove
  \eqref{eqn:main_thm_last}, we apply
  Lemma~\ref{lem:eigenvalue_perturbation} to the matrix trajectory
  \eqref{matrix_trajectory} and conclude
  \[
  \frac{d}{dt}\alpha(C(t))\Big|_{t=0}=\frac{\langle E_k,\re(yx^*)\rangle}{\Vert
    E_k\Vert_2 y^*x}.
  \]
  Again because $C(0)=A+\Delta_{\opt}$ and $C(\Vert
  E_k\Vert_2)=A+\Delta_k$, we conclude from the Taylor expansion that
  \begin{multline*}
    |\alpha(A+\Delta_k)-\alpha_{\epsilon,\nH}(A)|=\frac{\langle
      E_k,\re(yx^*)\rangle}{y^*x} +O(\Vert E_k \Vert_2^2)
    \\
    \leq \frac{\sqrt n\Vert
      E_k\Vert_2\Vert\re(yx^*)\Vert_F}{y^*x}+O(\Vert E_k \Vert_2^2) .
  \end{multline*}
  On the other hand, when the loop in
  Algorithm~\ref{alg:pseudospectral_abscissa} terminates at the $k$-th
  iteration,
  \begin{multline*}
    (1-r^\dagger) \Vert E_{k+1}\Vert_2\leq r^\dagger(\Vert
    E_k\Vert_2-\Vert E_{k+1}\Vert_2)
    \\
    \leq r^\dagger(\Vert \Delta_{k+1}-\Delta_{k}\Vert_2)\leq
    r^\dagger\tol_{\Delta},
  \end{multline*}
  where we use the first inequality in \eqref{recursive_E}, the
  triangle inequality and the stopping criteria of
  Algorithm~\ref{alg:pseudospectral_abscissa}. Hence $\Vert
  E_{k+1}\Vert_2\leq \frac{r^\dagger}{1-r^\dagger}\tol_{\Delta}$,
  completing the proof.
\end{proof}

According to Theorem~\ref{thm:local_convergence}, the algorithm's
convergence requires $r$ to be smaller than~1: the smaller $r$ is, the
smaller $r^\dagger$ can be, leading to faster convergence of
Algorithm~\ref{alg:pseudospectral_abscissa} and a smaller
approximation error $|\alpha-\alpha_{\epsilon,H}(A)|$. The value of
$r$ depends on various system parameters: it is proportional to the
square root of the dimension of $A$, inversely proportional to the
second smallest singular value of $A-\lambda I$, and increases as the
left and right eigenvectors of $A+\Delta_{\opt}$ associated with
$\lambda$ get closer to being orthogonal. Using arguments similar to
those in the proof of \cite[Theorem 5.7]{NG-MO:11}, we can conclude
that $r=O(\epsilon)$ and hence, as long as the energy of the
structured perturbation is small enough, $r<1$ is ensured.  A smaller
value of $\tol_{\Delta}$ results in a more accurate approximation of
the structured pseudospectral abscissa, cf.~\eqref{eqn:main_thm_last},
at the cost of more iterations in
Algorithm~\ref{alg:pseudospectral_abscissa}.


\section{Measuring network resilience: structured stability
  radius}\label{subsec:alg_sr}

Here we introduce an iterative algorithm to find the structured
stability radius $r_{\nH}(A)$, providing a metric of network
resilience, corresponding to question 3) of our problem statement. Our
strategy makes repeated use of
Algorithm~\ref{alg:pseudospectral_abscissa} to find the structured
pseudospectral abscissa $\alpha_{\epsilon,\nH}(A)$ for a given
energy~$\epsilon$, since the zero-crossing of this function
corresponds to the structured stability radius.

\subsection{Structured pseudospectral abscissa as a function of
  perturbation energy}

From the definition~\eqref{def:structured_pseudospectral_abscissa}, we
observe that $\epsilon\mapsto \alpha_{\epsilon,\nH}(A)$ is an
increasing function and $r_{\nH}(A)$ is its zero-crossing. We claim
that the map $\epsilon\mapsto \alpha_{\epsilon,\nH}(A)$ is locally
Lipschitz and hence differentiable almost everywhere by Rademacher's
theorem~\cite{FHC:83}. To show Lipschitzness, we note that eigenvalues
are Lipschitz functions with respect to perturbations of matrix
entries, cf.~\cite{AL:99} (and in fact differentiable when the
eigenvalue is simple).  Then, in view of
definition~\eqref{def:structured_pseudospectral_abscissa} and because
the maximum of Lipschitz functions is Lipschitz, we conclude that
$\epsilon\mapsto \alpha_{\epsilon,\nH}(A)$ is locally Lipschitz. The
following result is useful in computing derivative of $\epsilon\mapsto
\alpha_{\epsilon,\nH}(A)$ when it exists. 


\begin{lem}[Derivative of optimal value
  of~\eqref{Optimization_problem} with respect to energy of the
  perturbation]\label{lem:derivative}
  Let $\epsilon>0$ and $M\in\R^{n\times n}$.  Suppose
  Assumption~\ref{ass:1} holds for
  $\Delta_{\opt}(\epsilon)=[(\Delta_{\opt}(\epsilon))_{ij}]\in\nH\cap
  \bB_\epsilon$ and let $\eta(\epsilon,M)$ be the optimal value of the
  optimization problem~\eqref{Optimization_problem}. Let $\overline\nS: =
  \overline\nS(\epsilon,M)$, $\underline\nS: = \underline\nS(\epsilon,M)$
  be the index sets of saturation as in
  Proposition~\ref{prop:solution_to_optimization}.  Then $\epsilon
  \mapsto \eta(\epsilon,M)$ is Lipschitz and wherever it is
  differentiable, its derivative~is
  \begin{equation}
    \frac{\partial }{\partial \epsilon}\eta(\epsilon,M)
    =\frac{\epsilon}{\theta_{\opt}(\epsilon,M)} ,
  \end{equation}
  where $\theta_{\opt}$ is given by~\eqref{def:theta}.
\end{lem}

The proof is in the Appendix. If the perturbation is unstructured and
without saturation constraints, Lemma~\ref{lem:derivative} simplifies
to $ \frac{\partial}{\partial\epsilon}\eta(\epsilon,M) = \Vert
M\Vert_F$, recovering~\cite[Lemma 4]{NG:16}.
Using Lemma~\ref{lem:eigenvalue_perturbation}, when
$\Lambda_{\epsilon,\nH}(A)$ has a simple, unique rightmost eigenvalue
associated with the worst-case perturbation $\Delta_{\opt}(\epsilon)$
with RP-compatible right and left eigenvectors
$x(\epsilon),y(\epsilon)$,
\begin{align*}
  \frac{d}{d\epsilon}\alpha_{\epsilon,\nH}(A) &=
  \frac{d}{dt}\alpha(A+\Delta_{\opt}(\epsilon+t))\big\vert_{t=0}
  \\
  &=\re\Big(\frac{y(\epsilon)^*(\frac{d}{dt}
    \Delta_{\opt}(t)|_{t=\epsilon})x(\epsilon)}{y(\epsilon)^*x(\epsilon)}\Big)
  \\
  &=\frac{1}{y(\epsilon)^*x(\epsilon)}\langle \frac{d}{dt}
  \Delta_{\opt}(t)|_{t=\epsilon}, \re(y(\epsilon)x(\epsilon)^*)\rangle
  \\
  &=\frac{1}{y(\epsilon)^*x(\epsilon)}\frac{d}{dt}\langle
  \Delta_{\opt}(t), \re(y(\epsilon)x(\epsilon)^*)\rangle|_{t=\epsilon}.
\end{align*}
From Theorem~\ref{thm:1},
$\Delta_{\opt}(\epsilon)$ is the maximizer
of~\eqref{Optimization_problem} with $M =
\re(y(\epsilon)x(\epsilon)^*)$. Thus, using
Lemma~\ref{lem:derivative}, we conclude
\begin{equation}\label{eqn:derivative}
  \frac{d}{d\epsilon}\alpha_{\epsilon,\nH}(A) = \frac{1}{y(\epsilon)^*x(\epsilon)}
  \frac{\epsilon}{\theta_{\opt}(\epsilon,\re(y(\epsilon)x(\epsilon)^*))}.  
\end{equation}

\subsection{Iterative computation of structured stability radius}

We use Newton's method~\cite{NG:16} 
\begin{equation}\label{newton}
\epsilon_{l+1}=\epsilon_{l} -
  \frac{\alpha_{\epsilon_l,\nH}(A)}{\frac{d}{d\epsilon}
    \alpha_{\epsilon,\nH}(A)\big|_{\epsilon=\epsilon_l}}   
\end{equation}
to find the zero-crossing of
$\epsilon \mapsto \alpha_{\epsilon,\nH}(A)$, which by
definition~\eqref{def:cosntrained_stability_radius} is the stability
radius $r_{\nH}(A)$. Note from our discussion above that
$\frac{d}{d\epsilon}\alpha_{\epsilon,\nH}(A)$ is not well defined when
$\Lambda_{\epsilon,\nH}(A)$ has multiple rightmost eigenvalues. In
this case, we compute the value of the right-hand of
\eqref{eqn:derivative} for each rightmost eigenvalue and take the
minimum to be the ``gradient'' (which is in fact the subgradient of
$\epsilon\mapsto\alpha_{\epsilon,\nH}(A)$ with smallest
norm). Consequently, substituting~\eqref{eqn:derivative} into
\eqref{newton}, we update $\epsilon$ using
\begin{equation}\label{eqn:Newton}
  \epsilon_{l+1} = \epsilon_{l} -
  \frac{(y_l^*x_l)\theta_{\opt}(\epsilon_l,\re(y_lx_l^*))\alpha_{\epsilon_l,\nH}(A)}{\epsilon_l} 
\end{equation}
where $x_l,y_l$ denote the right and left eigenvectors associated with
the rightmost eigenvalue in $\Lambda_{\epsilon,\nH}(A)$ giving the
smallest value of right-hand side of \eqref{eqn:derivative}.  We also
observe from \eqref{eqn:Newton} that a fixed point of the iteration
corresponds to either $y^*_lx_l=0$, which is ruled out if the
corresponding rightmost eigenvalue is simple, or
$\alpha_{\epsilon_l,\nH}(A)=0$, in which case the iteration has found
the structured stability radius.  Algorithm~\ref{alg:stability_radius}
summarizes the procedure written in pseudocode.
\begin{algorithm}[htb]
  \caption{Computation of structured stability
    radius}\label{alg:stability_radius}
  \algorithmicrequire $A,\epsilon_0,\nH,\tol_{\Delta},\tol_{\alpha},\zeta$\\
  \algorithmicensure $r$
  \begin{algorithmic}[1]
    \Repeat{ $l=0,1,\cdots$}
    \State Pick $\Delta_{\textnormal{init}}\in\nH\cap\bB_{\epsilon_l}$
    \State\label{alg:alg1_in_alg2} Run
    Algorithm~\ref{alg:pseudospectral_abscissa} with inputs $A$,
    $\epsilon_l$, $\nH$, $\Delta_{\init}$, $\tol_{\alpha}$ and set the outputs $\Delta_l,\alpha_l,x_l,y_l,\theta_l$
    %
    %
    \State\label{alg:step_conditioned_newton}
    $\epsilon_{l+1}\gets\max\left\{\epsilon_{l}-\frac{(y_l^*x_l)\theta_l\alpha_l}{\epsilon_l},\zeta
      \epsilon_l\right\}$
    \Until{$|\alpha_l|\leq \tol_{\alpha}$}
    \State $r\gets\epsilon_l$
  \end{algorithmic}
\end{algorithm}
Notice that the matrix $\Delta_{\init}\in\nH\cap\bB_{\epsilon_l}$
selected in Step~2 is used as the initial guess of the worst-case
perturbation for Algorithm~\ref{alg:pseudospectral_abscissa}
(cf. Step~3). In our simulations, cf. Section~\ref{sec:example}, we
use either the zero matrix, a random matrix taken from the constraint
set or the result from the previous algorithm iteration.
In addition, in Step~\ref{alg:step_conditioned_newton}, we let
$\epsilon_{l+1}$ be lower bounded by $\zeta\epsilon_l$, for some
$\zeta\in(0,1)$, in order to prevent $\epsilon_{l+1}$ from becoming
negative. 


The next result characterizes the output of
Algorithm~\ref{alg:stability_radius}.

\begin{thm}[Error bound for
  Algorithm~\ref{alg:stability_radius}]\label{thm:iterative_error}
  Let $\lambda$ be a simple rightmost eigenvalue for the structured
  pseudospectrum corresponding to the structured stability radius,
  with RP-compatible right and left eigenvector pair $x$,~$y$.  Let
  $r^\dagger\in(0,1)$ and suppose that for each iteration of
  Algorithm~\ref{alg:stability_radius}, it holds that $r<r^\dagger$,
  where the parameter $r$ is defined by \eqref{def:r}. Further assume
  that for every iteration~$l$, the rightmost eigenvalue of
  $\Lambda_{\epsilon_l,\nH}(A)$ is simple and the initial guess
  $\Delta_{\init}$ in Step 2 is close enough to
  $\Delta_{\opt}(\epsilon_l)$ so that it falls in the region of
  convergence of Algorithm~\ref{alg:pseudospectral_abscissa}.
  Then, if Algorithm~\ref{alg:stability_radius} terminates at the
  $l_f$-th iteration, it holds that
  \begin{equation}\label{eqn:final_error_estimate}
    |\epsilon_{l_f}-r_{\nH}(A)| = O\left(\frac{\sqrt n
        r^\dagger\Vert\re(yx^*)\Vert_F}{(1-r^\dagger)y^*x}\tol_{\Delta}
      +\tol_{\alpha}\right). 
  \end{equation}
\end{thm}
\begin{proof}
  Since the rightmost eigenvalue is simple,
  $\frac{d}{d\epsilon}\alpha_{\epsilon,\nH}(A)\big|_{\epsilon=r_{\nH}(A)}$
  exists and is given by~\eqref{eqn:derivative}. Hence, using a
  first-order approximation on the inverse map of $\epsilon \mapsto
  \alpha_{\epsilon,\nH}(A)$,
  \begin{align}\label{eq:auxxx}
    \epsilon_{l_f} = r_{\nH}(A) + \frac{y^*x\theta^\dagger
      \alpha_{\epsilon_{l_f},\nH}(A)}{r_{\nH}(A)}+O(\alpha_{\epsilon_{l_f},\nH}(A)^2)
    ,
  \end{align}
  where $\theta^\dagger:=\theta_{\opt}(r_{\nH}(A),\re(yx^*))$ is given
  by \eqref{def:theta}. In addition,
  $|\alpha_{\epsilon_{l_f},\nH}(A)|\leq
  |\alpha_{\epsilon_{l_f},\nH}(A)-\alpha_{l_f}|+|\alpha_{l_f}|=\frac{\sqrt
    n r^\dagger\Vert\re(yx^*)\Vert_F}{(1-r^\dagger)y^*x}\tol_{\Delta}+
  O(\tol_{\Delta}^2)+\tol_{\alpha}$, where we have used the bound
  \eqref{eqn:main_thm_last} from Theorem~\ref{thm:local_convergence}.
  The result follows from using this fact in equation~\eqref{eq:auxxx}.
\end{proof}

%
Notice that $\l_f$ does not appear on the right-hand side
of~\eqref{eqn:final_error_estimate}, which indicates that the final
error does not depend on the total number of iterations. In other
words, the error introduced by using the approximate solution computed
by Algorithm~\ref{alg:pseudospectral_abscissa} at each iteration does
not accumulate.  We also make a final remark here that while in most
cases Newton's method has quadratic convergence rate, for some
particular initial guesses it may yield cyclic orbits and not
converge. Nevertheless, such cyclic orbits are unstable and in all our
simulations we observe Algorithm~\ref{alg:stability_radius} terminates
after a finite number of timesteps.


\section{Examples}\label{sec:example}
We illustrate here the use of the proposed algorithms to find the
structured pseudospectral abscissa and structured stability radius. In
all examples, we use the parameters $\epsilon_0 =1$,
$\tol_{\Delta}=\tol_{\alpha}=10^{-3}$, and $\zeta=0.1$. Our algorithms
are implemented in \texttt{MATLAB} on a personal computer with 4 cores
at 2.71GHz.

\subsection{Perturbation to edges of single node of sparse network
}\label{ex:polynomial}
Consider a 5-node network system with graph given in
Fig.~\ref{subfig:example_A}. Let system matrix in~\eqref{def:sys} be
given by
\begin{equation}
    A=
   \begin{pmatrix}
     0&     1&     0&     0&     0\\
     0&     0&     1&     0&     0\\
     0&     0&     0&     1&     0\\
     0&     0&     0&     0&     1\\
  -150&  -260&  -187&   -69&   -13
    \end{pmatrix},
\end{equation}
whose eigenvalues are $-2\pm i,-3,-3\pm i$. 
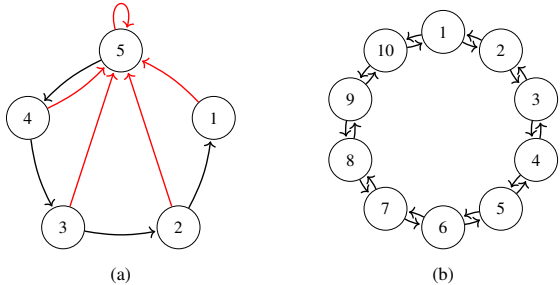
\begin{figure}[htb]
  \centering
  \tikzset{every picture/.style={scale=1}}
  \tikzset{every edge/.append style = {thick}}
  \begin{subfigure}[t]{0.47\columnwidth}
    \centering
    \scalebox{0.65}{
    \begin{tikzpicture}
        \node[state] at (0, 2)  (5) {5};
        \node[state] at (-1.9, .62)  (4) {4};
        \node[state] at (-1.18, -1.62)  (3) {3};
        \node[state] at (1.18, -1.62)  (2) {2};
        \node[state] at (1.9, .62)  (1) {1};

        \draw[every loop]
            (2) edge[bend right=10] node {} (1)
            (3) edge[bend right=10] node {} (2)
            (4) edge[bend right=10] node {} (3)
            (5) edge[bend right=10] node {} (4)
            (1) edge[color=red,bend right=10] node {} (5)
            (2) edge[color=red] node {} (5)
            (3) edge[color=red] node {} (5)
            (4) edge[color=red,bend right=10] node {} (5)
            (5) edge[color=red,loop above] node {} (5);
    \end{tikzpicture}
    }\caption{}\label{subfig:example_A}
  \end{subfigure}
  \begin{subfigure}[t]{0.47\columnwidth}
    \centering
    \scalebox{0.65}{
        \begin{tikzpicture}
        \node[state] at (0, 2)  (1) {1};
        \node[state] at (1.18, 1.62)  (2) {2};
        \node[state] at (1.9, .62)  (3) {3};
        \node[state] at (1.9, -.62)  (4) {4};
        \node[state] at (1.18, -1.62)  (5) {5};
        \node[state] at (0, -2)  (6) {6};
        \node[state] at (-1.18, -1.62)  (7) {7};
        \node[state] at (-1.9, -.62)  (8) {8};
        \node[state] at (-1.9, .62)  (9) {9};
        \node[state] at (-1.18, 1.62)  (10) {10};
        \draw[every loop]
            (1) edge[bend right=10] node {} (2)
            (2) edge[bend right=10] node {} (3)
            (3) edge[bend right=10] node {} (4)
            (4) edge[bend right=10] node {} (5)
            (5) edge[bend right=10] node {} (6)
            (6) edge[bend right=10] node {} (7)
            (7) edge[bend right=10] node {} (8)
            (8) edge[bend right=10] node {} (9)
            (9) edge[bend right=10] node {} (10)
            (10) edge[bend right=10] node {} (1)
            (1) edge[bend right=10] node {} (10)
            (2) edge[bend right=10] node {} (1)
            (3) edge[bend right=10] node {} (2)
            (4) edge[bend right=10] node {} (3)
            (5) edge[bend right=10] node {} (4)
            (6) edge[bend right=10] node {} (5)
            (7) edge[bend right=10] node {} (6)
            (8) edge[bend right=10] node {} (7)
            (9) edge[bend right=10] node {} (8)
            (10) edge[bend right=10] node {} (9);
    \end{tikzpicture}
    }\caption{}\label{subfig:example_B}
  \end{subfigure}
  \caption{Network graph of (a) Section~\ref{ex:polynomial} (only red
    edges are subject to perturbation) and (b)
    Section~\ref{sec:circulant} (without self-loops
    displayed).}\label{fig:example}
\end{figure}    
We consider additive perturbations to the edges of the node $5$, i.e.,
$\E_p=\{(5,i)\,:\, i = 1,\dots,5\}$. All rows of the perturbation
matrix are then zero, except for the last one which has entries
$\delta=(\delta_5,\dots, \delta_1)$.  We consider three scenarios: (i)
the entries $\delta_i$ are unconstrained, (ii) $\delta_i\leq 0$ for
all $i=1,\dots,5$; and (iii) $\delta_i\geq 0$, for all $i=1,\dots,5$.
Since $A$ is in controllable canonical form, its eigenvalues are the
roots of the 5-degree polynomial
$p_a(x)=x^5+a_1x^4+a_2x^3+a_3x^2+a_2x^1+a_0$, with $a=(a_1,\cdots,
a_5)=(13,\ 69,\ 187,\ 260,\ 150)$. Hence, the problem of determining
the structured stability radius of $A$ is equivalent to finding the
smallest perturbation on the non-leading coefficients of $p_a(x)$ such
that $p_{a-\delta}(x)$ becomes non-Hurwitz.  Even though the roots of
a high-degree polynomial are in general sensitive to its coefficients,
cf.~\cite{JW:63b}, the proposed algorithms efficiently compute their
outputs for~$A$.

\subsubsection{No constraints on perturbation}

Consider the case when $\nH=\{\Delta\in\R^{n\times
  n}:\Delta_{ij}=0\mbox{ if }(i,j)\not\in\E_p\}$.
Algorithm~\ref{alg:stability_radius}, using $\Delta_{\init}=0^{5\times
  5}$ in Step~2, finds $10.1465$ as the structured stability radius,
with worst-case perturbation $\delta_{\opt}=(10.1341,\ 0.4743,\
-0.1571,\ -0.0074,\ 0.0024)$.  Fig.~\ref{subfig:abnormal_ps_10} shows
the corresponding structured $\epsilon$-pseudospectrum of $A$ using
random samples from $\nH\cap\bB_\epsilon$.
Note that structured pseudospectrum is indeed touching the imaginary
axis, showing that the computed value is the true structured stability radius. We also note that the first element
$(\delta_{\opt})_1$ has the largest magnitude in the worst-case
perturbation, indicating that the element $A_{5,5}$ is the most
critical for preserving network stability (equivalently, the
coefficient $a_1$ in $p_a(x)$ is the most critical for preserving the
Hurwitzness of the polynomial).  If we eliminate the possibility of
perturbing $A_{5,5}$, our algorithm computes the new structured
stability radius $94.3512$, which is significantly larger.  
Fig.~\ref{subfig:CCF_csr1} shows the locally Lipschitzness nature of
$\epsilon \mapsto  \alpha_{\epsilon,\nH}(A)$.
The non-smooth corners in the map $\epsilon\to\alpha_{\epsilon,\nH}(A)$
corresponds to case when $\Lambda_{\epsilon,\nH}(A)$ has multiple rightmost eigenvalues. The evolution of $\Lambda_{\epsilon,\nH}(A)$ with respect to $\epsilon$ is appreciated in
Figs.~\ref{subfig:abnormal_ps_5} (with $\epsilon =5$) through
\ref{subfig:abnormal_ps_35} (with $\epsilon = 35$).


\begin{figure}[htb]
  \tikzset{every picture/.style={scale=0.8}}
  \begin{subfigure}[t]{0.47\columnwidth}
    \scalebox{0.5}{\input{figure/CCF_csr1}}\caption{}\label{subfig:CCF_csr1}
  \end{subfigure}
  \begin{subfigure}[t]{0.47\columnwidth}
    \scalebox{0.5}{\input{figure/CCF_cpa_5}}\caption{}\label{subfig:abnormal_ps_5}
  \end{subfigure}
  
  \begin{subfigure}[t]{0.47\columnwidth}
    \scalebox{0.5}{\input{figure/CCF_cpa_10}}\caption{}\label{subfig:abnormal_ps_10}
  \end{subfigure}
  \begin{subfigure}[t]{0.47\columnwidth}
    \scalebox{0.5}{\input{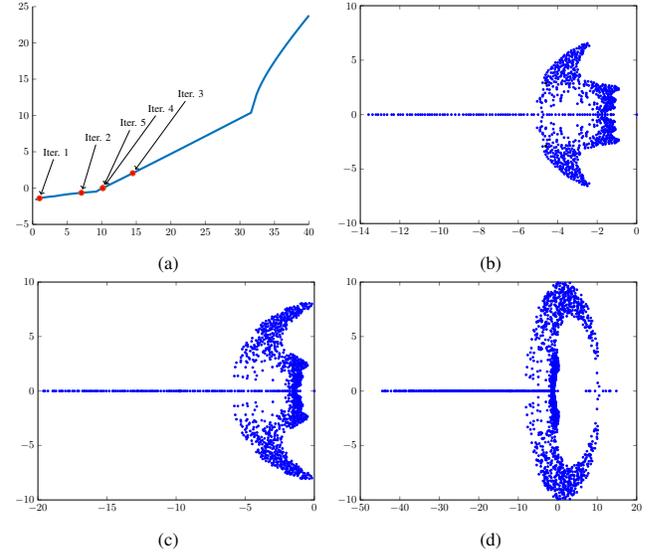}}\caption{}\label{subfig:abnormal_ps_35}
  \end{subfigure}
  \caption{Structured stability radius without constraints on
    perturbation. (a) Plot of $\epsilon \mapsto
    \alpha_{\epsilon,\nH}(A)$ generated by uniformly sampling
    $\epsilon$ in $[0,40]$ with stepsize $0.1$ and running
    Algorithm~\ref{alg:pseudospectral_abscissa} with
    $\Delta_{\init}=0^{5\times 5}$ for each value. The red dots
    correspond to the values $\alpha_l$ computed by
    Algorithm~\ref{alg:stability_radius}, which terminates after 5
    iterations. Generated with 1000 random samples of
    $\Delta\in\nH\cap\bB_\epsilon$ in each case, the figures (b) (c)
    and (d) are the estimated $\Lambda_{\epsilon,\nH}(A)$ with
    $\epsilon=5,10.1465$ (the structured stability radius computed by
    Algorithm~\ref{alg:stability_radius}) and $35$,
    respectively.}\label{fig:CCF_1}
\end{figure}

\subsubsection{Non-positive constraints on perturbation}

We next consider the case when $\nH=\{\Delta\in\R^{n\times
  n}:\Delta_{ij}=0 \mbox{ if }(i,j)\not\in\E_p \text{ and }
\Delta_{5j} \le 0 \, , j \in \until{5}\}$. This corresponds to
increasing the value of the coefficients of $p_a(x)$.
Algorithm~\ref{alg:stability_radius} finds $24.1733$ as the structured
stability radius, with a worst-case perturbation
$\delta_{\opt}=(-23.9583,\ 0,\ 0,\ -2.6928,\ -1.4657)$.
Fig.~\ref{subfig:constrained_abnormal_sr} shows the corresponding
structured $\epsilon$-pseudospectrum of $A$ using random samples from
$\nH\cap\bB_\epsilon$.
In this case, we observe that using $\Delta_{\init}=0^{5\times 5}$ in
Step 2 of Algorithm~\ref{alg:stability_radius} does not always lead to
convergence to a global optimizer of
problem~\eqref{eqn:1st_order_necessary_condition} when executing
Algorithm~\ref{alg:pseudospectral_abscissa} (cf.  blue curve in
Fig.~\ref{fig:CCF_2}). Instead, we 
make the selection
\begin{equation}\label{Delta_init_selection}
  \Delta_{\init}= \min_{\Delta\in
    \nH\cap\bB_{\epsilon_l}}\Vert\Delta-\Delta_{l-1}\Vert_F ,
\end{equation}
which takes the previous estimated worst-case perturbation as an
initial guess corresponding to the updated value of structured
pseudospectral abscissa.
In Fig.~\ref{fig:CCF_2}, we observe that with the
selection~\eqref{Delta_init_selection}, the first and second
iterations of Algorithm~\ref{alg:stability_radius} correspond to local
rightmost eigenvalues, but to global ones from the 3rd iteration
onwards, thereby allowing the algorithm to determine the structured
stability radius.

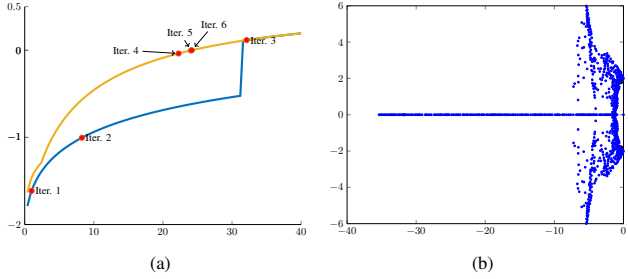
\begin{figure}[htb]
  \tikzset{every picture/.style={scale=0.8}}
  \begin{subfigure}[t]{0.47\columnwidth}
    \scalebox{.5}{\input{figure/CCF_csr2}}\caption{}\label{subfig:CCF_csr2}
  \end{subfigure}
  \begin{subfigure}[t]{0.47\columnwidth} 
  \scalebox{.5}{\input{figure/CCF_np_cpa}}\caption{}\label{subfig:constrained_abnormal_sr}
  \end{subfigure}
  \caption{Structured stability radius with non-positive
    constraints on perturbation. (a) The blue (respectively, yellow)
    curve corresponds to the output of
    Algorithm~\ref{alg:pseudospectral_abscissa} with uniform samples
    of $\epsilon$ in $[0,40]$ with stepsize $0.1$ and the initial
    guess $\Delta_{\init}=0^{5\times 5}$ (respectively, 1000 random
    $\Delta_{\init}\in\nH\cap\bB_\epsilon$ and taking the maximum).
    The red dots correspond to the values $\alpha_l$ of the stability
    radius computed by Algorithm~\ref{alg:stability_radius}, which
    terminates after 6 iterations; (b) Estimated
    $\Lambda_{\epsilon,\nH}(A)$ with $\epsilon=24.1733$, generated
    with 1000 random samples of
    $\Delta\in\nH\cap\bB_\epsilon$.}\label{fig:CCF_2}
\end{figure}

\subsubsection{Non-negative constraints on perturbation}

Lastly, the case when $\nH=\{\Delta\in\R^{n\times n}:\Delta_{ij}=0
\mbox{ if }(i,j)\not\in\E_p \text{ and } \Delta_{5j} \ge 0 \, , j \in
\until{5}\}$ yields almost identical result as the scenario without
constraints, with $10.1478$ as the structured stability radius and
$\delta_{\opt}=(10.1367,\ 0.474,\ 0,\ 0,\ 0.0024)$ as worst-case
perturbation.

\subsection{Circulant network}\label{sec:circulant}
Consider a network system of $10$ agents with circulant graph given in
Fig.~\ref{subfig:example_A}.  Each agent's state is 1-dimensional and
obeys the simple dynamics $\dot x_i=u_i$, where $u_i$ is the control
input.  Here, we consider
\begin{equation}\label{circular_u}
  u_i=x_{(i+1)}-x_{(i-1)}-0.1x_i,
\end{equation}
where we abuse notation by identifying $(i+1)=1$ if $i=10$ and
$(i-1)=10$ if $i=1$.  The closed-loop system is then
\begin{equation}
  \dot x_i=-0.1x_i+x_{(i+1)}-x_{(i-1)},
\end{equation}
which is a typical formation control problem of agents running in a
circle when the state $x_i$ corresponds to the angle of agent $i$ with
respect to the center of rotation.  The corresponding $10\times 10$
circulant matrix $A$ is
\begin{equation}
  A=\begin{pmatrix}
    -0.1&1 && &-1\\
    -1&-0.1&1\\
    &\ddots&\ddots&\ddots\\
    & & -1&-0.1& 1\\
    1 & & &-1&-0.1
  \end{pmatrix}    .
\end{equation}
and the network graph is given in Fig.~\ref{subfig:example_B}.

Suppose the first agent is compromised by an adversary so that instead
of~\eqref{circular_u}, it implements $u_1 =
(1+\delta_1)x_2-(1-\delta_2)x_{10}-(0.1-\delta_3)x_1$. This
corresponds to perturbations in the non-zero elements of the first row
of~$A$. If $\delta=(\delta_1,\delta_2,\delta_3)$ is unconstrained,
Algorithm~\ref{alg:stability_radius} initialized at
$\Delta_{\init}=0^{10\times 10}$ finds the structured stability radius
$0.4727$ and a worst-case perturbation
$\delta_{\opt}=(-0.0889,0.0889,0.4556)$. Fig.~\ref{subfig:1_agent_ps}
shows the structured pseudospectrum.

If the perturbation is constrained with $\overline\Delta_{1,1}=0.1$,
$\underline\Delta_{1,2}=-1$, and $\overline\Delta_{1,10}=1$, the sign of
each term in the control law~\eqref{circular_u} is preserved and
Algorithm~\ref{alg:stability_radius} initialized with
$\Delta_{\init}=0^{10\times 10}$ finds the structured stability radius
$1.4177$ and a worst-case perturbation
$\delta_{\opt}=(0.1,-1,1)$. With this worst-case perturbation, the
control $u_1$ vanishes and the network becomes marginally stable.
This can also be seen in Fig.~\ref{subfig:1_agent_ps_c}, where the
structured pseudospectrum touches the imaginary axis.

\begin{figure}[htb]
  \tikzset{every picture/.style={scale=0.8}}
\begin{subfigure}[t]{0.47\columnwidth}
    \scalebox{0.5}{\input{figure/1_agent_ps}}\caption{}\label{subfig:1_agent_ps}
  \end{subfigure}
  \begin{subfigure}[t]{0.47\columnwidth}
    \scalebox{0.5}{\input{figure/1_agent_ps_c}}\caption{}\label{subfig:1_agent_ps_c}
  \end{subfigure}
  \caption{Structured pseudospectra of the circulant network when one
    agent is subject to adversaries, generated by 2000 random
    sampling of $\Delta\in\nH\cap\bB_\epsilon$. (a) $\epsilon=0.4727$
    and the perturbation is unconstrained. (b) $\epsilon=1.4177$ and the
    perturbation is sign preserving.}\label{fig:circulant_1_agent}
\end{figure}
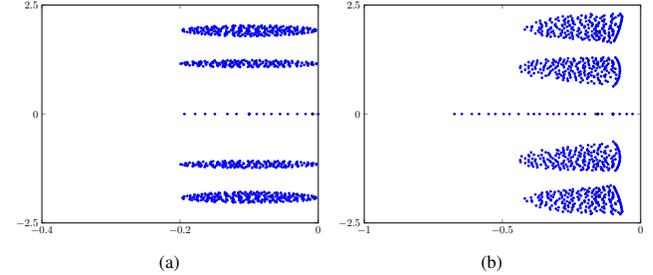

In scenarios where several, instead of just one, consecutive agents
are compromised, Fig.~\ref{subfig:agents_vs_sr} shows the stability
radius as a function of the number of compromised agents.  As this
number grows, the set $\nH$ becomes larger and, for a fixed
$\epsilon$, $\alpha_{\epsilon,\nH}(A)$ increases. As a result,
$r_{\nH}(A)$ decreases, which is correctly captured in
Fig.~\ref{subfig:agents_vs_sr}. Notice also that when $5$ or more
consecutive agents are subject to adversaries, the structured
stability radii are the same independently of whether the perturbation
is unconstrained or sign preserving. This is because the worst-case
perturbation is distributed over all compromised edges, with entries
on each edge small enough that they do not violate the boundary
constraints.  It is also worth mentioning that when all agents are
subject to adversaries, the estimated pseudospectrum generated by
randomly sampling of perturbations in $\nH\cap\bB_{\epsilon}$ does not
accurately reflect the true pseudospectrum, cf.
Fig.~\ref{subfig:10_agent_ps},
because of the sensitivity of the rightmost eigenvalue with respect to
the perturbation. Nevertheless, our iterative algorithms are still
able to correctly find the pseudospectral abscissa and stability
radius.

\begin{figure}[htb]
  \tikzset{every picture/.style={scale=0.8}}
  \begin{subfigure}[t]{0.47\columnwidth}
    \scalebox{0.5}{\input{figure/agents_vs_sr}}
    \caption{}\label{subfig:agents_vs_sr}
  \end{subfigure}
    \begin{subfigure}[t]{0.47\columnwidth}
    \scalebox{0.5}{\input{figure/10_agent_ps}}\caption{}\label{subfig:10_agent_ps}
\end{subfigure}
\caption{(a) Stability radius vs. number of consecutive compromised
  agents. (b) Estimated $\Lambda_{\epsilon,\nH}(A)$ when all agents
  are subject to unconstrained adversaries, generated by 2000 random
  sampling of $\Delta\in\nH\cap\bB_\epsilon$ with $\epsilon=0.1826$.}
\end{figure}
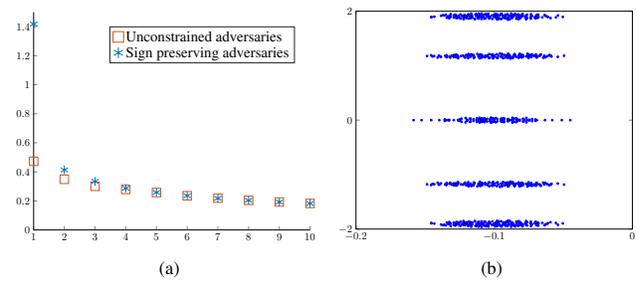

\subsection{Tolosa networks}\label{sec:Tolosa}
Lastly we consider large-scale systems where the network matrix
in~\eqref{def:sys} is given by a Tolosa
matrix~\cite{RB-RP-KR-RB-JD:97}. These matrices are sparse,
asymmetric, and Hurwitz, and here we consider the cases of dimension
$n=90$, $340$, and $1090$.  For each network, we consider additive
perturbations to the entries of the 19th to the 36th rows
($\E_p=\{(i,j):19\leq i\leq 36\}$), since these are the ones with the
most non-zero elements.  We consider two scenarios for the constraints
corresponding, respectively, to non-negative and non-positive
perturbations, in both cases bounded in magnitude by 10,
\begin{align*}
  \nH_1 & =\{\Delta \in [0,10]^{n \times n}: \Delta_{ij}=0\ \forall
  (i,j)\not\in\E_p\},
  \\
  \nH_2 &=\{\Delta \in [-10,0]^{n \times n} : \Delta_{ij}=0\ \forall
  (i,j)\not\in\E_p\}.
\end{align*}
We run Algorithm~\ref{alg:stability_radius}, where at each iteration
we take the initial condition $\Delta_{\init}=0^{n\times n}$ and 10
random initial guesses $\Delta_{\init} \in \nH\cap\bB_{\epsilon_l}$,
and select the one that results in the maximum structured
pseudospectral abscissa.  Tables~\ref{table:1} and~\ref{table:2}
summarize the results for each scenario.

\begin{table}[htb]
  \centering
  \begin{tabular}{c|c|c|c|c}
    Matrix & Dimension & $r_{\nH}$ & Iterations & Time (secs)\\
    \hline
    Tolosa90 & 90 & 0.50251 & 2 & 0.915406\\
    Tolosa340 & 340 & 0.066404 & 13 & 58.732560\\
    Tolosa1090 & 1090 & 0.15919 & 7 & 312.649106\\
  \end{tabular}
  \caption{Structured 
    stability radii under perturbations in~$\nH_1$ obtained by
    Algorithm~\ref{alg:stability_radius}. The  columns, from 
    left to right, correspond to  the matrix name, its  dimension, the
    structured stability radius, the number of algorithm iterations, and the total
    computation time.}\label{table:1}
\end{table}

\begin{table}[htb]
    \centering
    \begin{tabular}{c|c|c|c|c}
       Matrix & Dimension & $r_{\nH}$ & Iterations & Time (secs)\\
       \hline
        Tolosa90 & 90 & 4.6737 & 3 & 2.896226\\
        Tolosa340 & 340 & 0.35407 & 17 & 45.508739\\
        Tolosa1090 & 1090 & 0.19163 & 7 & 213.961736\\
    \end{tabular}
    \caption{Structured 
      stability radii under perturbations in~$\nH_2$ obtained by
      Algorithm~\ref{alg:stability_radius}. The 
      adescription of the columns is the same as in 
      Table~\ref{table:1}.}\label{table:2}
\end{table}

In our simulations, we notice that the worst-case perturbations have a
structure that is even sparser than the one specified by~$\E_p$. For
example, for the case of Tolosa90 with perturbation in~$\nH_1$, the
worst-case perturbation $\Delta_{\opt}$ only has 8 elements with
magnitude larger than $0.001$, and the largest element is
$(\Delta_{\opt})_{21,21}=0.5021$, which takes about $99.84\%$ of the
total energy ($r_{\nH}(A)=0.50251$) of $\Delta_{\opt}$, suggesting the
importance of protecting edge $(21,21)\in\E_p$ against structured
additive topological perturbations.  We also observe that for these
large-scale systems, around $90\%$ of the total computation time
corresponds to the computation of the left and right eigenvectors
using the \texttt{MATLAB} command~\texttt{eigs}.  Hence, there is
significant room for improving the computation time by optimizing the
computation of eigenvectors of large-scale sparse matrices.

\section{Conclusions}\label{sec:conclusion}
We have studied the stability
of linear dynamical systems systems against additive perturbations of
the system matrix.  We have formalized questions about whether an
adversary can destabilize the network with perturbations of a given
energy and determining what is the maximum amount of perturbation
energy the network can withstand without becoming unstable using the
concepts of structured pseudospectral abscissa and structured
stability radius.
We have proposed iterative algorithms that asymptotically compute both
quantities for a given network along with the corresponding worst-case
structured perturbations.  Future work will study the global
convergence of the algorithms and their extension to consider
arbitrary values of the perturbation energy, develop distributed
strategies for the computation of the structured pseudospectral
abscissa and the corresponding worst-case perturbation to help
individual agents assess their relative value in ensuring network
stability, and examine resilience in scenarios where adversaries only
have partial knowledge of the network structure.

\appendix

\begin{proof}[Proof of Proposition~\ref{prop:solution_to_optimization}]
  Let $\lambda_0\geq 0, \underline\lambda_{ij}\geq
  0,\overline\lambda_{ij}\geq 0$ be the Lagrange multipliers for the
  constraints. Writing~\eqref{Optimization_problem} element-wise, the
  Lagrangian is
  \begin{multline*}
    L(\Delta,\lambda_0,\underline\lambda_{ij},\overline\lambda_{ij}) =
    -\sum_{(i,j)\in\E_p}\Delta_{ij}m_{ij}+\lambda_0(\sum_{(i,j)\in\E_p}\Delta_{ij}^2-\epsilon^2)
    \\ 
    +\sum_{(i,j)\in\E_p}
    \underline\lambda_{ij}(\underline\Delta_{ij}-\Delta_{ij}) +
    \sum_{(i,j)\in\E_p}\overline\lambda_{ij}(\Delta_{ij}-\overline\Delta_{ij}).
  \end{multline*}
  Define the sets $ \overline\nS
  :=\{(i,j)\in\E_p:(\Delta_{\opt})_{ij}=\overline\Delta_{ij}\}$ and $
  \underline\nS :=
  \{(i,j)\in\E_p:(\Delta_{\opt})_{ij}=\underline\Delta_{ij}\}$.  Note
  that since the sets are defined for the worst-case perturbation
  $\Delta_{\opt}$, they are exactly the index sets of saturation as
  defined in Proposition~\ref{prop:solution_to_optimization}.  At the
  optimum
  $(\Delta_{\opt},\lambda_{0,\opt},\underline\lambda_{\opt},\overline\lambda_{\opt})$,
  the Lagrangian is maximized.  Because $L$ is linear in
  $\underline\lambda_{ij},\overline\lambda_{ij}$, and
  $(\Delta_{\opt})_{ij}$'s do not reach boundary values for any
  $(i,j)\in\E_p\backslash(\overline\nS\cup\underline\nS)$, we have
  $\overline\lambda_{\opt,ij}=\underline\lambda_{\opt,ij}=0$ for all
  those edges. Thus, for all
  $(i,j)\in\E_p\backslash(\overline\nS\cup\underline\nS)$,
  \[
  0=\frac{\partial}{\partial
    \Delta_{ij}}L(\Delta_{\opt},\lambda_{\opt,0},\underline\lambda_{\opt},
  \overline\lambda_{\opt})
  \\
  =-m_{ij}+2\lambda_{0,\opt}(\Delta_{\opt})_{ij}.
  \]
  Because of Assumption~\ref{ass:1}, $\lambda_{0,\opt}\neq 0$ or
  otherwise the above equation does not hold for that particular
  unsaturated element with $m_{ij}\neq 0$. This implies
  $\sum_{(i,j)\in\E_p}(\Delta_{\opt})_{ij}^2=\epsilon^2$ and
  \begin{equation}\label{eqn:proof_formula_Delta_opt}
    (\Delta_{\opt})_{ij}=\frac{m_{ij}}{2\lambda_{0,\opt}}
  \end{equation}
  for all $(i,j)\in\E_p\backslash\nS$, which yields the formula for
  $\Delta_{\opt}$ given in
  Proposition~\ref{prop:solution_to_optimization} if
  $\theta_{\opt}=\frac{1}{2\lambda_{0,\opt}}$. The latter is verified
  by squaring both sides of~\eqref{eqn:proof_formula_Delta_opt} and
  summing all the terms whose indices are not in the saturation
  index set,
  which gives 
  \begin{multline*}
    \epsilon^2-\sum_{(i,j)\in\overline\nS}
    \overline\Delta_{ij}^2-\sum_{(i,j)\in\underline\nS}\underline\Delta_{ij}^2
    \\
    =\sum_{(i,j)\in\E_p\backslash(\overline\nS\cup\underline\nS)}(\Delta_{\opt})_{ij}^2
    = \frac{\sum_{(i,j)\in
        \E_p\backslash(\overline\nS\cup\underline\nS)}m_{ij}^2}{4\lambda_{0,\opt}^2}.
  \end{multline*}
  Rearranging the terms and comparing with \eqref{def:theta}, we
  conclude
  \begin{align*}
    \theta_{\opt}= \sqrt{\frac{\epsilon^2-\sum_{(i,j)\in\overline\nS}
        \overline\Delta_{ij}^2-\sum_{(i,j)\in\underline\nS} \underline
        \Delta_{ij}^2}{\sum_{(i,j)\in\E_p\backslash\nS}m_{ij}^2}}=\frac{1}{2\lambda_{0,\opt}}.
  \end{align*}
  Finally, we show the local Lipschitzness of $\theta_{\opt}$ over
  some neighborhood $D\ni(\epsilon,M)$. Let $\overline\epsilon>
  \epsilon,\overline m>\Vert M\Vert_F$, and $\delta>0$. Define
  $D:=\{(\epsilon',M')\in \R_{\geq 0}\times \R^{n\times
    n}:\epsilon\in[0,\overline\epsilon], \Vert M\Vert_F\leq \overline
  m,|\epsilon'-\epsilon|+\Vert M'-M\Vert\leq \delta\}$. Let
  $(\epsilon_1,M_1)=(\epsilon,M)$ and pick arbitrary
  $(\epsilon_2,M_2)\in D$. Denote $\overline
  \nS_k:=\overline\nS(\epsilon_k,M_k),\underline\nS_k:=\overline\nS(\epsilon_k,M_k),
  \theta_k:=\theta_{\opt}(\epsilon_k,M_k),k=1,2$. We make the
  following abbreviations:
  \begin{align*}
    a
    &:=\sum_{(i,j)\in\overline\nS_1\cap\overline\nS_1}\overline\Delta_{ij}^2
    +
    \sum_{(i,j)\in\underline\nS_1\cap\overline\nS_1}\underline\Delta_{ij}^2,
    \\
    b
    &:=\sum_{(i,j)\in\overline\nS_1\backslash\overline\nS_2}\overline\Delta_{ij}^2
    +
    \sum_{(i,j)\in\underline\nS_1\backslash\underline\nS_2}\underline\Delta_{ij}^2,
    \\
    c
    &:=\sum_{(i,j)\in\overline\nS_2\backslash\overline\nS_1}\overline\Delta_{ij}^2
    +
    \sum_{(i,j)\in\underline\nS_2\backslash\underline\nS_1}\underline\Delta_{ij}^2
  \end{align*}
  and for $k=1,2$, let
  \begin{align*}
    d_k &:=\sum_{(i,j)\in\E_p}(M_k)_{ij}^2,
    \\
    e_k &:=\sum_{(i,j)\in\overline\nS_1\cap\overline\nS_2}(M_k)_{ij}^2
    + \sum_{(i,j)\in\underline\nS_1\cap\underline\nS_2}(M_k)_{ij}^2,
    \\
    f_k&:=\sum_{(i,j)\in\overline\nS_1\backslash\overline\nS_2}(M_k)_{ij}^2
    +
    \sum_{(i,j)\in\underline\nS_1\backslash\underline\nS_2}(M_k)_{ij}^2,
    \\
    g_k&:=\sum_{(i,j)\in\overline\nS_2\backslash\overline\nS_1}(M_k)_{ij}^2
    +
    \sum_{(i,j)\in\underline\nS_2\backslash\underline\nS_1}(M_k)_{ij}^2.
  \end{align*}
  Using the formula \eqref{def:theta}, we have
  \[
  \theta_1^2=\frac{\epsilon_1^2-a-b}{d_1-e_1-f_1},\quad
  \theta_2^2=\frac{\epsilon_2^2-a-c}{d_2-e_2-g_2}.
  \]
  By making $\delta$ sufficiently small, $(\epsilon_2,M_2)$ is close
  enough to $(\epsilon_1,M_1)$ and hence the denominators of the
  expressions for $\theta_1^2,\theta_2^2$ are close and have a common
  positive lower bound. In other words, there exists $\underline m>0$
  only depending on $\delta$ and $M_1$ such that $ d_1-e_1-f_1\geq
  \underline m^2$ and $d_2-e_2-g_2\geq \underline m^2$.
  In addition, by properties of index sets of saturation
  \eqref{def:saturation_set}, $ a\leq \theta_1^2 e_1$, $a\leq
  \theta_2^2 e_2$, and
  \begin{align}
    &\theta_2^2f_2\leq b\leq \theta_1^2f_1,\label{set_2}
    \\
    &\theta_1^2g_1\leq c \leq \theta_2^2g_2,\label{set_3}
  \end{align}
  which further implies that
  \begin{equation}\label{relation_between_epsilon_and_overline_m}
    \epsilon_k\leq\theta_kd_k\leq\theta_k\Vert M_k\Vert_F\leq \theta_k\overline m
  \end{equation}
  for both $k=1,2$.  We show the Lipschitzness of
  $\theta_{\opt}(\eta,M)$ by considering two cases, 1) $f_1=0$ or
  $g_2=0$, and 2) both $f_1,g_2>0$. In case 1), when $f_1=0$, the
  inequality \eqref{set_2} implies $b=f_2=0$. Meanwhile, it can also
  be deduced from \eqref{set_3} that
  \begin{align*}
    (\epsilon_1^2-a)g_1&\leq c(d_1-e_1),\\
    (\epsilon_2^2-a)(d_2-e_2-g_2)&\leq (\epsilon_2^2-a-c)(d_2-e_2).
  \end{align*}
  Denote $\kappa_1:=(d_1-e_1)(d_2-e_2-g_2) \geq \underline m^4$. Note
  that
  \begin{align*}
    \theta_2^2&-\theta_1^2
    \\
    &=\kappa_1^{-1}\big((\epsilon_2^2-a-c)(d_1-e_1)-(\epsilon_1^2-a)(d_2-e_2-g_2)\big)\\
    &\leq \kappa_1^{-1}\big((\epsilon_2^2-a)(d_1-e_1)-(\epsilon_1^2-a)g_1\\
    &\qquad-(\epsilon_1^2-a)(d_2-e_2-g_2)\big)\\
    &= \kappa_1^{-1}\big((\epsilon_2^2-\epsilon_1^2)(d_1-e_1)\\
    &\qquad+(\epsilon_1^2-a)((d_1-e_1-g_1)-(d_2-e_2-g_2))\big)\\
    &\leq \kappa_1^{-1}\big(|\epsilon_2^2-\epsilon_1^2|\overline
    m^2+\epsilon_1^2|(d_1-e_1-g_1)-(d_2-e_2-g_2)|\big). 
    \end{align*}
    Note that \eqref{relation_between_epsilon_and_overline_m} implies
    that
    $\frac{\epsilon_1}{\theta_1+\theta_2}\leq\frac{\epsilon_1+\epsilon_2}{
      \theta_1+\theta_2}\leq \overline m$. Hence
    \begin{align*}
      \theta_2-\theta_1&\leq
      \frac{1}{(\theta_1+\theta_2)\kappa_1}\Big(\overline
      m^2(\epsilon_1+\epsilon_2)|\epsilon_2-\epsilon_1|
      \\
      &\qquad+\epsilon_1^2(\sqrt{d_1-e_1-g_1}+\sqrt{d_2-e_2-g_2})
      \\
      &\qquad\cdot|\sqrt{d_1-e_1-g_1}-\sqrt{d_2-e_2-g_2}|\Big)
      \\
      &\leq\frac{\overline
        m^2}{\kappa_1}\frac{\epsilon_1+\epsilon_2}{\theta_1+\theta_2}|\epsilon_2-\epsilon_1|
      \\
      &\qquad+
      \frac{\overline\epsilon}{\kappa_1}\frac{\epsilon_1}{\theta_1+\theta_2}(\sqrt{d_1-e_1-g_1}
      + \sqrt{d_2-e_2-g_2})
      \\
      &\qquad\cdot|\sqrt{d_1-e_1-g_1}-\sqrt{d_2-e_2-g_2}|\Big)
      \\
      &\leq \frac{\overline m^3}{\underline
        m^4}|\epsilon_2-\epsilon_1|+\frac{\overline\epsilon\overline
        m^2}{\underline m^4}\Vert M_2-M_1\Vert_F\Big).
    \end{align*}
    On the other hand,
    \begin{align*}
      \theta_1^2&-\theta_2^2\\
      &=\kappa_1^{-1}\big((\epsilon_1^2-a)(d_2-e_2-g_2)-(\epsilon_2^2-a-c)(d_1-e_1)\big)\\
      &=\kappa_1^{-1}\big((\epsilon_1^2-\epsilon_2^2)(d_2-e_2-g_2)\\
      &\qquad+(\epsilon_2^2-a)(d_2-e_2-g_2)-(\epsilon_2^2-a-c)(d_1-e_1)\big)\\
      &\leq \kappa_1^{-1}\big((\epsilon_1^2-\epsilon_2^2)(d_2-e_2-g_2)\\
      &\qquad+(\epsilon_2^2-a-c)((d_2-e_2)-(d_1-e_1))\big)\\
      &\leq \kappa_1^{-1}\big(|\epsilon_1^2-\epsilon_2^2|\overline
      m^2+\epsilon_2^2|(d_1-e_1)-(d_2-e_2)|\big).
    \end{align*}
    and again we can conclude the same upper bound on
    $\theta_1-\theta_2$. Therefore $\theta_{\opt}$ is Lipschitz on
    $D$. Similar arguments hold when $g_2=0$.

    In case (2), we have $f_2>0,g_1>0$ by picking $\delta$ small
    enough. Hence there exist $m_f,m_g>0$ only depending on $\delta$
    and $M_1$ such that $f_k\geq m_f^2,g_k\geq m_g^2$ for both
    $k=1,2$. Consequently
    \[
    | \frac{1}{\sqrt{f_1}}-\frac{1}{\sqrt{f_2}}| =
    \frac{|\sqrt{f_1}-\sqrt{f_2}|}{\sqrt{f_1f_2}}\leq\frac{\Vert
      M_1-M_2\Vert_F}{m_f^2}
    \]
    where we  have used the triangle inequality. Similarly we also have
    $|\frac{1}{\sqrt{g_1}}-\frac{1}{\sqrt{g_2}}|\leq \frac{\Vert
      M_1-M_2\Vert_F}{m_g}$. Meanwhile, it can be directly concluded
    from \eqref{set_2}, \eqref{set_3} that
    \begin{align*}
      |\theta_1-\theta_2|&\leq
      \max\{\theta_1-\theta_2,\theta_2-\theta_1\}
      \\
      &\leq\max\big\{\big|\sqrt{\frac{b}{f_1}}-\sqrt{\frac{b}{f_2}}\big|,\big|\sqrt{\frac{c}{g_1}}-\sqrt{\frac{c}{g_2}}\big|\big\}
      \\
      &\leq \max\{\frac{\sqrt{b}}{m_f^2},\frac{\sqrt{c}}{m_g^2}\}\Vert M_1-M_2\Vert_F\\
      &\leq \frac{\overline \epsilon}{(\min\{m_f,m_g\})^2}\Vert
      M_1-M_2\Vert_F ,
    \end{align*}
    completing the proof.
\end{proof}

\begin{proof}[Proof of Lemma~\ref{lem:why_optimization_alg_works}]
  Note that Algorithm~\ref{alg:solving_optimization_problem}
  terminates in a finite number of steps since the sequence of
  index sets of saturation $\overline\nS$, $\underline\nS$ is monotonically
  non-decreasing and uniformly bounded. In fact, the algorithm stops
  if the sets $\overline\nS$, $\underline\nS$ are unchanged or their union
  satisfies $\overline\nS\cup\underline\nS=\E_p$. Under
  Assumption~\ref{ass:1},
  Algorithm~\ref{alg:solving_optimization_problem} terminates if and
  only if NotDone = false, i.e., when \eqref{saturation_set_1}
  holds. In addition, when
  Algorithm~\ref{alg:solving_optimization_problem} terminates,
  $\Delta_{\opt}$ is given by \eqref{assignment_of_Delta_opt} for the
  computed $\overline\nS$, $\underline\nS$.  Thus, to prove the statement,
  we are left to show that \eqref{saturation_set_2} and
  \eqref{saturation_set_3} hold. Let
  $\theta_{\opt,1},\theta_{\opt,2},\cdots,\theta_{\opt,k}$ be the
  sequence of $\theta_{\opt}$ generated by
  Algorithm~\ref{alg:solving_optimization_problem}.
  From~\eqref{def:theta}, it is clear that these variables are
  non-negative.
 We next show that the sequence is non-decreasing.
  Take any consecutive terms $\theta_{\opt,l}$, $\theta_{\opt,l+1}$
  and let $\delta\overline\nS$ (resp. $\delta\underline\nS$) be the
  difference between the set $\overline\nS$ (resp. $\underline\nS$)
  computed in the $l$-th iteration and the one computed in the next
  iteration. In other words, $m_{ij}\theta_{\opt,l}\geq
  \overline\Delta_{ij}$ for all $(i,j)\in\delta\overline\nS$ and
  $m_{ij}\theta_{\opt,l}\leq \underline\Delta_{ij}$ for all
  $(i,j)\in\delta\underline\nS$. To simplify the presentation, for the
  $l$-th iteration, let
  \begin{align*}
    a &:=\sum_{(i,j)\in\overline\nS}\overline\Delta_{ij}^2 +
    \sum_{(i,j)\in\underline\nS}\underline\Delta_{ij}^2,
    \\
    b &:=\sum_{(i,j)\in\delta\overline\nS}\overline\Delta_{ij}^2 +
    \sum_{(i,j)\in\delta\underline\nS}\underline\Delta_{ij}^2, \;
    c :=\sum_{(i,j)\in\E_p}m_{ij}^2,
    \\
    d
    &:=\sum_{(i,j)\in\overline\nS}m_{ij}^2+\sum_{(i,j)\in\underline\nS}m_{ij}^2,
    \\
    e
    &:=\sum_{(i,j)\in\delta\overline\nS}m_{ij}^2+\sum_{(i,j)\in\delta\underline\nS}m_{ij}^2.
  \end{align*}
  Note that \eqref{def:theta} implies
  \begin{subequations}
    \begin{align}
      \theta_{\opt,l}^2 &
      =\frac{\epsilon^2-a}{c-d}, \label{eqn:pf_of_sol_to_opt_2}
      \\
      \theta_{\opt,l+1}^2 &=
      \frac{\epsilon^2-a-b}{c-d-e}. \label{eqn:pf_of_sol_to_opt_3}
    \end{align}
  \end{subequations}
  Plugging~\eqref{eqn:pf_of_sol_to_opt_2} into $ b\leq
  \theta_{\opt,l}^2e$, we have $b(c-d)\leq(\epsilon^2-a)e$.
  Subtracting \eqref{eqn:pf_of_sol_to_opt_2} from
  \eqref{eqn:pf_of_sol_to_opt_3},
  \begin{align*}
    \theta&_{\opt,l+1}^2-\theta_{\opt,l}^2
    =\frac{(\epsilon^2-a)e-b(c-d)}{(c-d-e)(c-d)}\geq 0 ,
  \end{align*}
  as claimed.  Observe that in the execution of
  Algorithm~\ref{alg:solving_optimization_problem}, an edge
  $(i,j)\in\overline\nS$ (resp. $\underline\nS$) is added to the index set
  of saturation at some iteration $l\le k$, when
  $m_{ij}\theta_{\opt,l}\geq\overline\Delta_{ij}\geq 0$
  (resp. $m_{ij}\theta_{\opt,l}\leq\underline\Delta_{ij}\leq
  0$). Since $\{\theta_{\opt,l}\}$ is non-decreasing, we deduce
  $m_{ij}\theta_{\opt,k}\geq\overline\Delta_{ij}$
  (resp. $m_{ij}\theta_{\opt,k}\leq\underline\Delta_{ij}$), thereby
  verifying~\eqref{saturation_set_2} and~\eqref{saturation_set_3}.
\end{proof}

\begin{proof}[Proof of Lemma~\ref{lem:continuous_dependence}]
  Let $\theta_{k}=\theta_{\opt}(\epsilon,M_k),\quad k=1,2$. Notice
  that for some $(i,j)\in\E_p$, if $(i,j)$ belongs to neither index
  sets of saturation for the two optimization problems, then
  $(\Delta_k)_{ij}=(M_k)_{ij}\theta_k$ for both $k=1,2$ so
  \begin{equation}\label{component_wise_inequality}
      |(\Delta_1)_{ij}-(\Delta_2)_{ij}|\leq|(M_1)_{ij}\theta_1-(M_2)_{ij}\theta_2|
  \end{equation}
  holds with equality. If $(i,j)$ only belongs to one index set of
  saturation, say $(i,j)\in\overline\nS(\epsilon,M_1)$ but
  $(i,j)\not\in\overline\nS(\epsilon,M_2)$, then
  $(M_1)_{ij}\theta_1\geq (\Delta_1)_{ij}= \overline\Delta_{ij}\geq
  (\Delta_2)_{ij}=(M_2)_{ij}\theta_2$ and hence again the inequality
  \eqref{component_wise_inequality} holds. This is also true if
  $(i,j)$ only belongs to one of the other index sets of saturation. If
  $(i,j)$ belongs to the index sets of saturation for both
  optimization problems, then $(\Delta_1)_{ij}=(\Delta_2)_{ij}$ and
  \eqref{component_wise_inequality} holds again.
  
  Let $\delta>0$ be the one picked in the proof of
  Proposition~\ref{prop:solution_to_optimization} for defining the
  neighborhood $D\ni(\epsilon,M_1)$. We have $(\epsilon,M_2)\in D$ as
  well and using the Lipschitzness of $\theta_{\opt}$,
  \begin{align*}
    |(&\Delta_{1})_{ij}-(\Delta_{2})_{ij}| \leq
    |(M_1)_{ij}\theta_{1}-(M_{2})_{ij}\theta_{2}|
    \\
    &\leq
    |(M_1)_{ij}||\theta_{1}-\theta_{2}|+|(M_1)_{ij}-(M_{2})_{ij}|\theta_{2}
    \\
    &\leq\kappa\epsilon\Vert M_1-M_2\Vert_F+\frac{\epsilon}{\underline
      m}|(M_1){ij}-(M_{2})_{ij}|,
  \end{align*}
  where $\kappa=\max\{\frac{\overline m^2}{\underline
    m^4},\frac{1}{(\min\{m_f,m_g\})^2}\}$ and $\overline m, \underline
  m,m_f,m_g$ come from the proof of
  Proposition~\ref{prop:solution_to_optimization} and only depend on
  $\delta,M_1$.  As a result,
  \begin{align*}
    \Vert \Delta_{1}&-\Delta_{2}\Vert_F =
    \sqrt{\sum_{(i,j)\in\E_p}|(\Delta_{1})_{ij}-(\Delta_{2})_{ij}|^2}
    \\
    &\leq \left(\sqrt{|\E_p|}\kappa+\frac{1}{\underline m}\right)\epsilon\Vert M_1-M_2\Vert_F
  \end{align*}
  and hence the statement holds with
  $\ell:=\left(\sqrt{|\E_p|}\kappa+\frac{1}{\underline m}\right)$.
\end{proof}

\begin{proof}[Proof of Lemma~\ref{lem:derivative}]
   %
  For fixed $M\in\R^{n\times n}$, 
  $\eta(\epsilon,M)=\langle\Delta_{\opt},M\rangle$ is linear in
  $\Delta_{\opt}$, and hence Lipschitz. In addition,
  $\theta_{\opt}(\epsilon,M)\mapsto \Delta_{\opt}$, given by
  \eqref{assignment_of_Delta_opt}-\eqref{def:saturation_set} is also
  Lipschitz. Lastly, $\epsilon\mapsto \theta_{\opt}(\epsilon,M)$ is
  Lipschitz by Proposition~\ref{prop:solution_to_optimization}. Hence,
  the composition $\epsilon\mapsto\eta(\epsilon,M)$ is locally Lipschitz.

  To find the derivative of $\epsilon\mapsto\eta(\epsilon,M)$ when it
  exists, we first conclude from the continuity of
  $\epsilon\mapsto\theta_{\opt}(\epsilon,M)$ and the criteria for
  index sets of saturation \eqref{def:saturation_set} that for
  $\delta\in\R$ with sufficiently small $|\delta|$,
  $\overline\nS(\epsilon+\delta,M)\subseteq\overline\nS(\epsilon,M)$
  and
  $\underline\nS(\epsilon+\delta,M)\subseteq\underline\nS(\epsilon,M)$. Meanwhile,
  the linearity of the objective function \eqref{Optimization_problem}
  implies that when $\epsilon$ grows, the saturated elements in the
  optimizer remain saturated. In other words, if $\delta\geq 0$, then
  $\overline\nS(\epsilon+\delta,M)\supseteq\overline\nS(\epsilon,M)$
  and
  $\underline\nS(\epsilon+\delta,M)\supseteq\underline\nS(\epsilon,M)$
  Therefore, the index sets of saturation are the same for
  sufficiently small $\delta>0$; i.e.,
  $\overline\nS(\epsilon+\delta,M) =
  \overline\nS(\epsilon,M)=:\overline\nS$ and
  $\underline\nS(\epsilon+\delta,M)=\underline\nS(\epsilon,M)=:\underline\nS$. Thus
  the difference between $\eta(\epsilon+\delta,M)$ and
  $\eta(\epsilon,M)$ can be expressed as
  \begin{multline}\label{difference_eta}
    \eta(\epsilon+\delta,M)-\eta(\epsilon,M)
    \\
    = \big(\theta_{\opt}(\epsilon,M) -
    \theta_{\opt}(\epsilon+\delta,M)\big)
    \hspace*{-3pt}\sum_{(i,j)\in\E_p\backslash(\overline\nS\cup\underline\nS)}m_{ij}^2.
  \end{multline}
  This equation is useful for computing the right one-sided derivative
  of $\epsilon\mapsto\eta(\epsilon,M)$, which equals to the derivative
  of this map when it exists,
  \begin{align*}
    \frac{d}{dt}&\eta(t,M)|_{t=\epsilon}=\lim_{\delta\to
      0^+}\frac{\eta(\epsilon+\delta,M)-\eta(\epsilon,M)}{\delta}
    \\
    &=\lim_{\delta\to 0^+}\frac{\theta_{\opt}(\epsilon+\delta,M) -
      \theta_{\opt}(\epsilon,M)}{\delta}
    \hspace*{-3pt}\sum_{(i,j)\in\E_p\backslash(\overline\nS\cup\underline\nS)}m_{ij}^2
    \\
    &=\frac{d}{dt}\theta_{\opt}(t,M)\big|_{t=\epsilon}\cdot
    \sum_{(i,j)\in\E_p\backslash(\overline\nS\cup\underline\nS)}m_{ij}^2
    \\
    &=\epsilon\sqrt{\frac{\sum_{(i,j)\in\E_p\backslash(\overline\nS\cup\underline\nS)}
        m_{ij}^2}{\epsilon^2-\sum_{(i,j)\in\overline\nS}
        \overline\Delta_{ij}^2-\sum_{(i,j)\in\underline\nS}\underline\Delta_{ij}^2}}
    =\frac{\epsilon}{\theta_{\opt}}.
  \end{align*}
\end{proof}

\bibliographystyle{IEEEtran}
\bibliography{bib/alias.bib,bib/Main.bib,bib/Main-add.bib,IEEEabrv.bib}

\vspace*{-6ex}

\begin{IEEEbiography}[{\includegraphics[width=1in,height=1.25in,clip,keepaspectratio]{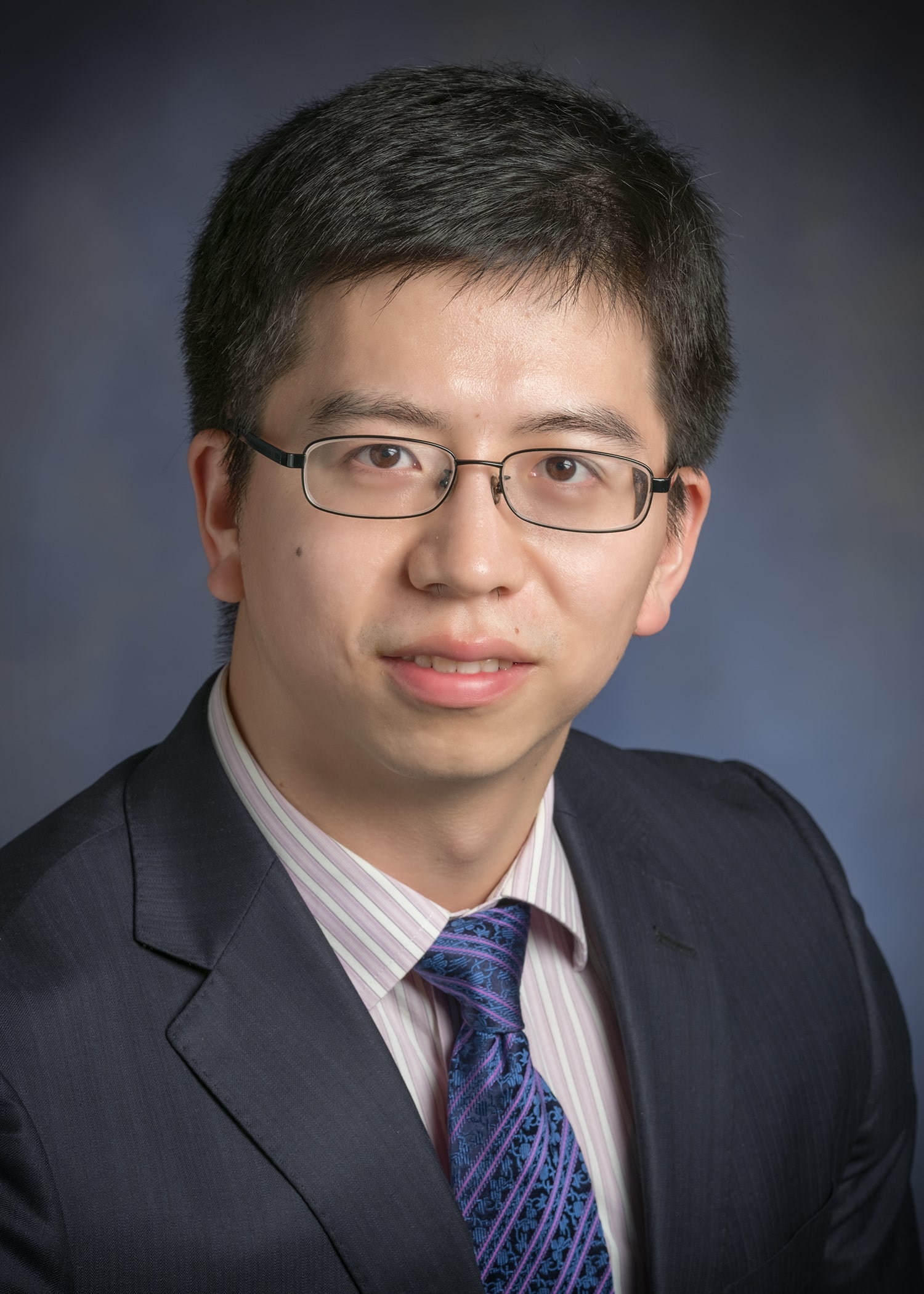}}]{Shenyu
    Liu} (S'16-M'20) received his B. Eng. degree in Mechanical
  Engineering and B.S. degree in Mathematics from the University of
  Singapore, Singapore, in 2014. He then received his M.S. degree in
  Mechanical Engineering from the University of Illinois,
  Urbana-Champaign in 2015, where he also received his Ph.D. degree in
  Electrical Engineering in 2020. He is currently a postdoctoral
  researcher in Department of Mechanical and Aerospace Engineering at
  University of California San Diego. His research interest includes
  matrix perturbation theory, Lyapunov methods, input-to-state
  stability theory, switched/hybrid systems and motion planning via
  geometric methods.
\end{IEEEbiography}

\vspace*{-6ex}

\begin{IEEEbiography}[{\includegraphics[width=1in,height=1.25in,clip,keepaspectratio]{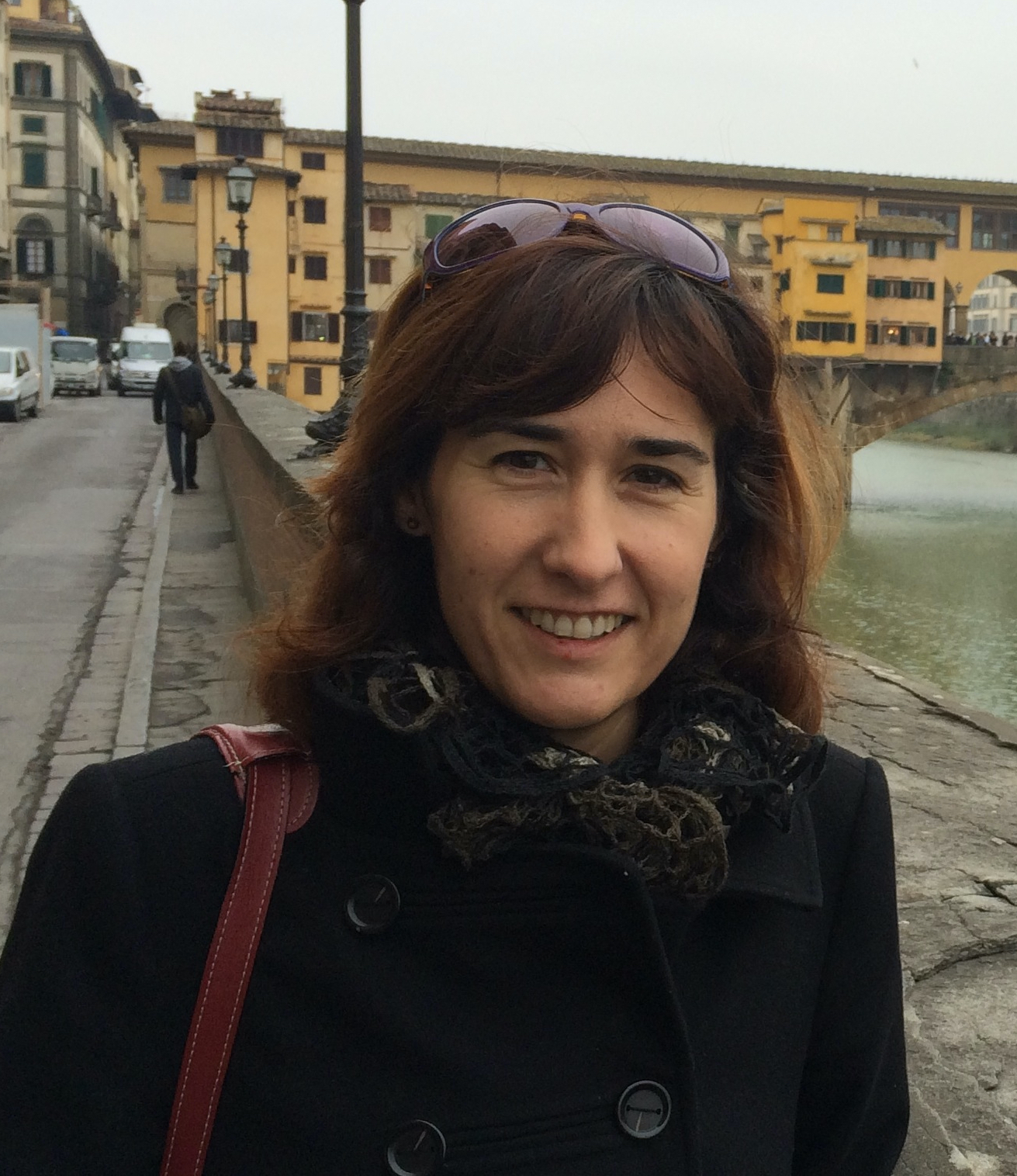}}]{Sonia
    Mart{\'\i}nez} (M'02-SM'07-F'18) is a Professor of Mechanical and
  Aerospace Engineering at the University of California, San Diego,
  CA, USA. She received the Ph.D. degree in Engineering Mathematics
  from the Universidad Carlos III de Madrid, Spain, in May 2002. She
  was a Visiting Assistant Professor of Applied Mathematics at the
  Technical University of Catalonia, Spain (2002-2003) and a
  Postdoctoral Fulbright Fellowship at the Coordinated Science
  Laboratory of the University of Illinois, Urbana-Champaign
  (2003-2004) and the Center for Control, Dynamical systems and
  Computation of the University of California, Santa Barbara
  (2004-2005).  Her research interests include the control of network
  systems, multi-agent systems, nonlinear control theory, and
  robotics.  She received the Best Student Paper award at the 2002
  IEEE Conference on Decision and Control for her work on the control
  of underactuated mechanical systems and was the recipient of a NSF
  CAREER Award in 2007. For the paper ``Motion coordination with
  Distributed Information,'' co-authored with Jorge Cort\'es and
  Francesco Bullo, she received the 2008 Control Systems Magazine
  Outstanding Paper Award. She is the Editor in Chief of the recently
  launched Open Journal of Control Systems.
\end{IEEEbiography}

\vspace*{-6ex}

\begin{IEEEbiography}[{\includegraphics[width=1in,height=1.25in,clip,keepaspectratio]{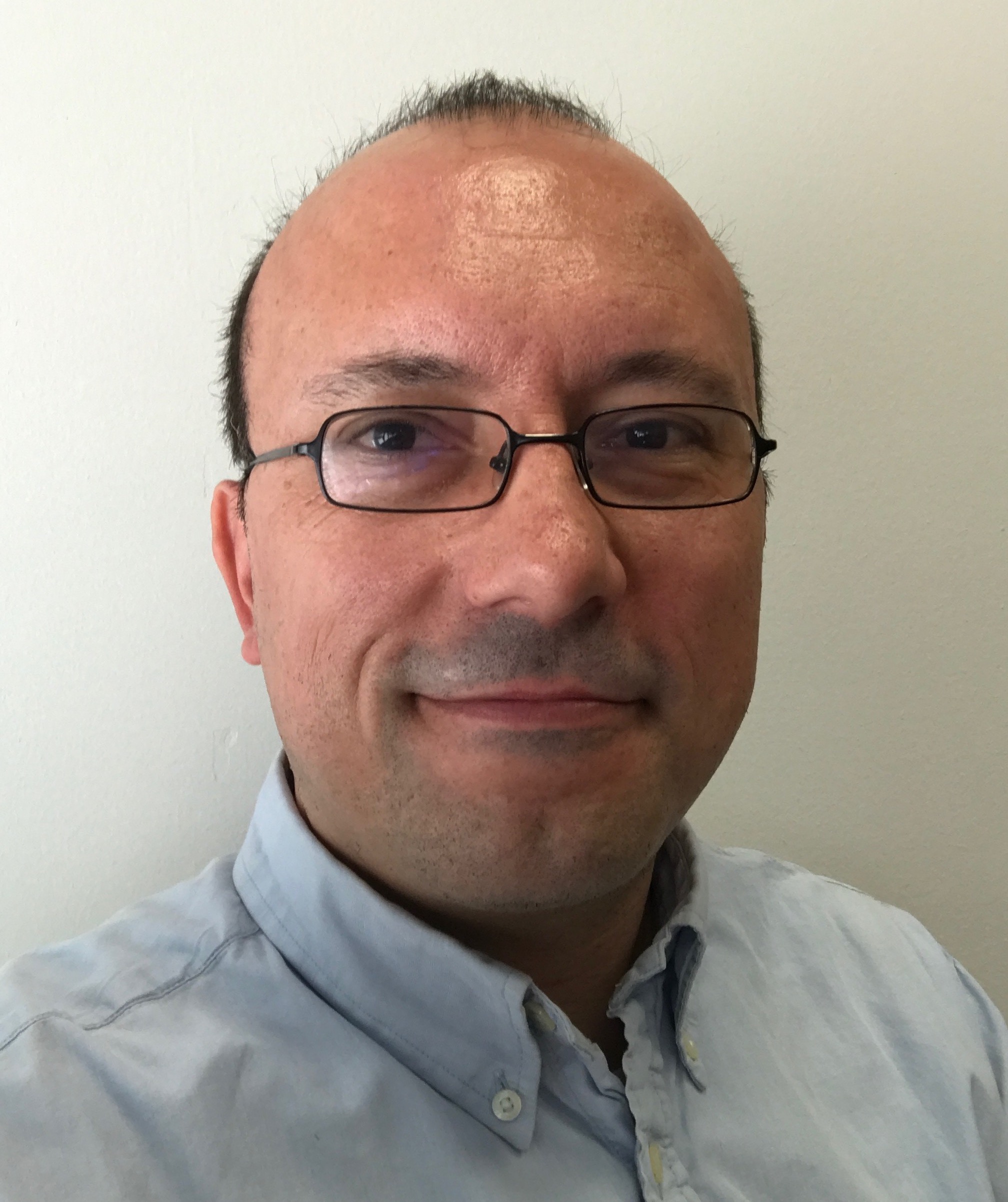}}]{Jorge
    Cort\'{e}s}
  (M'02, SM'06, F'14) received the Licenciatura degree in mathematics
  from Universidad de Zaragoza, Zaragoza, Spain, in 1997, and the
  Ph.D. degree in engineering mathematics from Universidad Carlos III
  de Madrid, Madrid, Spain, in 2001. He held postdoctoral positions
  with the University of Twente, Twente, The Netherlands, and the
  University of Illinois at Urbana-Champaign, Urbana, IL, USA. He was
  an Assistant Professor with the Department of Applied Mathematics
  and Statistics, University of California, Santa Cruz, CA, USA, from
  2004 to 2007. He is currently a Professor in the Department of
  Mechanical and Aerospace Engineering, University of California, San
  Diego, CA, USA.  He is the author of Geometric, Control and
  Numerical Aspects of Nonholonomic Systems (Springer-Verlag, 2002)
  and co-author (together with F. Bullo and S.  Mart{\'\i}nez) of
  Distributed Control of Robotic Networks (Princeton University Press,
  2009).  He is a Fellow of IEEE and SIAM.
  His current research interests include distributed control and
  optimization, network science, nonsmooth analysis, reasoning and
  decision making under uncertainty, network neuroscience, and
  multi-agent coordination in robotic, power, and transportation
  networks.
\end{IEEEbiography}

\end{document}

%% file: figure/CCF_csr1.tex
%
%
\definecolor{mycolor1}{rgb}{0.00000,0.44700,0.74100}%
\definecolor{mycolor2}{rgb}{0.85000,0.32500,0.09800}%
\begin{tikzpicture}

\begin{axis}[%
width=4.521in,
height=3.566in,
at={(0.758in,0.481in)},
scale only axis,
xmin=0,
xmax=40,
ymin=-5,
ymax=25,
axis x line*=bottom,
axis y line*=left,
axis background/.style={fill=white},
legend style={legend cell align=left, align=left, draw=white!15!black}
]
\addplot [color=mycolor1,line width=2pt]
  table[row sep=crcr]{%
0.4	-1.60007222475609\\
0.8	-1.45455223608217\\
1.2	-1.36099507471719\\
1.6	-1.29148452915376\\
2	-1.23604097786837\\
2.4	-1.18988188908079\\
2.8	-1.15032725490218\\
3.2	-1.11571804482418\\
3.6	-1.0535835830631\\
4	-0.991225938590774\\
4.4	-0.934154922800644\\
4.8	-0.881595534320789\\
5.2	-0.832934201778809\\
5.6	-0.787675998581939\\
6	-0.745415230592963\\
6.4	-0.705814938672672\\
6.8	-0.66859177637354\\
7.2	-0.633505201643481\\
7.6	-0.60034904100463\\
8	-0.568945174408982\\
8.4	-0.539138684414667\\
8.8	-0.510793779572808\\
9.2	-0.483790967806488\\
9.6	-0.255143106422553\\
10	-0.0684861087490899\\
10.4	0.118740526769415\\
10.8	0.306493011903213\\
11.2	0.494732371177321\\
11.6	0.68342375230352\\
12	0.872535870637569\\
12.4	1.06204053799864\\
12.8	1.251912268802\\
13.2	1.44212794568374\\
13.6	1.63266653885667\\
14	1.82350885800582\\
14.4	2.0146373494679\\
14.8	2.20603591397598\\
15.2	2.39768975158557\\
15.6	2.58958522584765\\
16	2.78170974603962\\
16.4	2.97405166332044\\
16.8	3.16660017952937\\
17.2	3.35934526696783\\
17.6	3.55227759714283\\
18	3.74538847768378\\
18.4	3.93866979633172\\
18.8	4.13211397039643\\
19.2	4.32571390230909\\
19.6	4.51946293931254\\
20	4.71335483725228\\
20.4	4.90738372808115\\
20.8	5.10154409048079\\
21.2	5.29583072323614\\
21.6	5.4902387213257\\
22	5.68476345367583\\
22.4	5.87940054363631\\
22.8	6.07414585079513\\
23.2	6.26899545421603\\
23.6	6.46394563747803\\
24	6.65899287481873\\
24.4	6.8541338183708\\
24.8	7.04936528662321\\
25.2	7.24468425344836\\
25.6	7.4400878384477\\
26	7.6355732978003\\
26.4	7.83113801570324\\
26.8	8.02677949678733\\
27.2	8.22249535879339\\
27.6	8.41828332590649\\
28	8.61414122264396\\
28.4	8.81006696802465\\
28.8	9.00605857028234\\
29.2	9.2021141218407\\
29.6	9.39823179474475\\
30	9.59440983627455\\
30.4	9.79064656495765\\
30.8	9.98694036676353\\
31.2	10.1832896916199\\
31.6	10.3796930501109\\
32	11.6299234195873\\
32.4	12.8966940063473\\
32.8	13.8029928249791\\
33.2	14.5781080533314\\
33.6	15.279863619083\\
34	15.9334695654607\\
34.4	16.5525726407576\\
34.8	17.1455260793297\\
35.2	17.7178708131152\\
35.6	18.273497820147\\
36	18.8152584319014\\
36.4	19.3453125252047\\
36.8	19.8653401812941\\
37.2	20.3766769697634\\
37.6	20.8804040128891\\
38	21.3774100057902\\
38.4	21.8684351590232\\
38.8	22.3541031005425\\
39.2	22.8349445295209\\
39.6	23.3114150805241\\
40	23.783909035953\\
};

\addplot[only marks, mark=*, mark options={fill=red}, mark size=3.000pt, draw=mycolor2] table[row sep=crcr]{%
x	y\\
1	-1.4037496320335\\
7.06305294082456	-0.645290620727375\\
14.4884330127315	2.05692951264711\\
10.1880542422128	0.0194676429912368\\
10.1464594119779	2.97421131612019e-06\\
};

\node[right, align=left]
at (axis cs:1,5) (l1) {Iter. 1};
\node[right, align=left]
at (axis cs:7.063,7) (l2) {Iter. 2};
\node[right, align=left]
at (axis cs:20.488,13) (l3) {Iter. 3};
\node[right, align=left]
at (axis cs:16.188,11) (l4) {Iter. 4};
\node[right, align=left]
at (axis cs:12.146,9) (l5) {Iter. 5};

\node[]
at (axis cs:1,-1.404) (1) {};
\node[]
at (axis cs:7.063,-0.645) (2) {};
\node[]
at (axis cs:14.488,2.057) (3) {};
\node[]
at (axis cs:10.188,0.02) (4) {};
\node[]
at (axis cs:10.146,0) (5) {};

\draw [->] (l1) -- (1);
\draw [->] (l2) -- (2);
\draw [->] (l3) -- (3);
\draw [->] (l4) -- (4);
\draw [->] (l5) -- (5);

\end{axis}

\end{tikzpicture}%

%% file: figure/CCF_cpa_5.tex
%
%
\begin{tikzpicture}

\begin{axis}[%
width=4.521in,
height=3.566in,
at={(0.758in,0.481in)},
scale only axis,
xmin=-14,
xmax=0,
ymin=-10,
ymax=10,
ytick={-10,-5,0,5,10},
axis background/.style={fill=white},
legend style={legend cell align=left, align=left, draw=white!15!black}
]
\addplot [color=black, draw=none, mark=asterisk, mark options={solid, black}]
  table[row sep=crcr]{%
-4.59534919089106	2.89948609393015\\
-4.59534919089106	-2.89948609393015\\
-1.26095401205483	1.55225102835262\\
-1.26095401205483	-1.55225102835262\\
-1.29308699436798	0\\
};

\addplot+[color=blue, draw=none, mark=*, mark size=1pt,mark options={solid, blue}]
  table[row sep=crcr]{%
-3.00000000000002	0.999999999999957\\
-3.00000000000002	-0.999999999999957\\
-2.9999999999999	0\\
-2.00000000000002	1.00000000000002\\
-2.00000000000002	-1.00000000000002\\
0	0\\
-8.78624476281581	0\\
-1.58416659375699	2.50446810479884\\
-1.58416659375699	-2.50446810479884\\
-1.23822951388988	0.642038366582111\\
-1.23822951388988	-0.642038366582111\\
-8.40544813979763	0\\
-1.65885874507553	2.41242570518213\\
-1.65885874507553	-2.41242570518213\\
-10.8570482482635	0\\
-1.11281844100735	2.38819008736651\\
-1.11281844100735	-2.38819008736651\\
-1.32511910427834	0.459242255502738\\
-1.32511910427834	-0.459242255502738\\
-3.36386904161998	3.49101809312252\\
-3.36386904161998	-3.49101809312252\\
-2.36148992767151	0\\
-1.34186479817847	0.958997963833137\\
-1.34186479817847	-0.958997963833137\\
-9.9630293558838	0\\
-1.32638153409094	2.46152581368334\\
-1.32638153409094	-2.46152581368334\\
-11.1976472708633	0\\
-1.25388001490392	2.32946454091697\\
-1.25388001490392	-2.32946454091697\\
-7.13860747147927	0\\
-1.82008861979585	2.39325797980465\\
-1.82008861979585	-2.39325797980465\\
-1.42238299899608	0.52641254983151\\
-1.42238299899608	-0.52641254983151\\
-4.07398311111797	3.53995548437983\\
-4.07398311111797	-3.53995548437983\\
-1.73838429225278	0\\
-1.2699566679777	1.15998674607675\\
-1.2699566679777	-1.15998674607675\\
-4.47557728110429	3.54694493348203\\
-4.47557728110429	-3.54694493348203\\
-1.22684992711121	1.54315686045827\\
-1.22684992711121	-1.54315686045827\\
-1.17324001092179	0\\
-1.36828565711742	2.20236124536823\\
-1.36828565711742	-2.20236124536823\\
-1.37340180083787	0.65065734052333\\
-1.37340180083787	-0.65065734052333\\
-3.5584824356404	4.89465603564942\\
-3.5584824356404	-4.89465603564942\\
-1.22171373567071	1.37378036864795\\
-1.22171373567071	-1.37378036864795\\
-3.99965690335887	4.78252553621075\\
-3.99965690335887	-4.78252553621075\\
-8.63718108180239	0\\
-1.35644281895612	1.91740085307298\\
-1.35644281895612	-1.91740085307298\\
-1.73738551741744	0.41905968304894\\
-1.73738551741744	-0.41905968304894\\
-4.99896822173377	0\\
-2.72114002861575	1.99131486044862\\
-2.72114002861575	-1.99131486044862\\
-4.22476537183725	1.32676341735595\\
-4.22476537183725	-1.32676341735595\\
-1.50322253933079	1.40340812757408\\
-1.50322253933079	-1.40340812757408\\
-1.86513510760031	0\\
-3.0205218738549	5.5174230006292\\
-3.0205218738549	-5.5174230006292\\
-1.13402746808167	1.15357293612593\\
-1.13402746808167	-1.15357293612593\\
-1.48246172105331	0\\
-9.18882755572944	0\\
-1.55289961547803	2.34915815424943\\
-1.55289961547803	-2.34915815424943\\
-3.53543794596282	5.44787511979926\\
-3.53543794596282	-5.44787511979926\\
-1.09502294785608	1.27378435223835\\
-1.09502294785608	-1.27378435223835\\
-9.38462162774632	0\\
-1.51662727676801	0.327367086952577\\
-1.51662727676801	-0.327367086952577\\
-3.20869640336636	0\\
-1.30300655561721	1.82955315357537\\
-1.30300655561721	-1.82955315357537\\
-1.30642936749918	0\\
-7.88988176268265	0\\
-1.70766839264074	2.0446053294854\\
-1.70766839264074	-2.0446053294854\\
-1.53198099831169	0.532700979810933\\
-1.53198099831169	-0.532700979810933\\
-1.41853265914116	2.10781906656687\\
-1.41853265914116	-2.10781906656687\\
-3.26998719292422	5.54404857680796\\
-3.26998719292422	-5.54404857680796\\
-1.7375228898599	2.17906814963743\\
-1.7375228898599	-2.17906814963743\\
-7.616112864133	0\\
-3.57710077016474	5.26165675928289\\
-3.57710077016474	-5.26165675928289\\
-11.6240850152633	0\\
-1.1273691858627	2.49606480462371\\
-1.1273691858627	-2.49606480462371\\
-1.16857638653499	0.570053612352107\\
-1.16857638653499	-0.570053612352107\\
-3.1403628096838	3.76249982395276\\
-3.1403628096838	-3.76249982395276\\
-2.1999227248453	0\\
-1.46300834065338	0.799775183188105\\
-1.46300834065338	-0.799775183188105\\
-2.08046593104512	0\\
-3.0453384155294	4.65478517872215\\
-3.0453384155294	-4.65478517872215\\
-6.65248010003216	0\\
-6.42156813474858	0\\
-2.03369078916619	2.51424295345814\\
-2.03369078916619	-2.51424295345814\\
-4.57981987139261	3.29826626702368\\
-4.57981987139261	-3.29826626702368\\
-3.55440877063422	5.11964370041916\\
-3.55440877063422	-5.11964370041916\\
-5.53715068145665	0\\
-3.63248598597281	0\\
-3.87514864880723	4.87782963679931\\
-3.87514864880723	-4.87782963679931\\
-1.46356577509938	1.8788064424125\\
-1.46356577509938	-1.8788064424125\\
-2.89240583691658	5.37534236127396\\
-2.89240583691658	-5.37534236127396\\
-3.2528059444131	5.43364292179047\\
-3.2528059444131	-5.43364292179047\\
-11.0724188305857	0\\
-1.22399188197373	2.18083292121111\\
-1.22399188197373	-2.18083292121111\\
-3.96531463931177	3.4314167636222\\
-3.96531463931177	-3.4314167636222\\
-12.7944290112278	0\\
-0.925501018607241	2.40732027038837\\
-0.925501018607241	-2.40732027038837\\
-5.8502873081338	0\\
-2.18134521390473	2.15630727270654\\
-2.18134521390473	-2.15630727270654\\
-3.74538501891555	3.97401896232441\\
-3.74538501891555	-3.97401896232441\\
-1.3384977896328	1.27933161830247\\
-1.3384977896328	-1.27933161830247\\
-2.01470973465862	2.16226531710194\\
-2.01470973465862	-2.16226531710194\\
-1.65626168413368	0.512524220818396\\
-1.65626168413368	-0.512524220818396\\
-2.70475912436751	6.23249862478574\\
-2.70475912436751	-6.23249862478574\\
-10.3705850929963	0\\
-4.0946075115655	3.3337444744589\\
-4.0946075115655	-3.3337444744589\\
-2.74369610799489	6.00079983817873\\
-2.74369610799489	-6.00079983817873\\
-1.32635540234377	0.79197592888855\\
-1.32635540234377	-0.79197592888855\\
-3.27639306751491	5.29762224316936\\
-3.27639306751491	-5.29762224316936\\
-13.3805328720976	0\\
-3.73793785766286	4.91467705334838\\
-3.73793785766286	-4.91467705334838\\
-4.54305783068965	2.66251908824773\\
-4.54305783068965	-2.66251908824773\\
-1.32510093526924	1.62865307029614\\
-1.32510093526924	-1.62865307029614\\
-3.70095033816051	4.57570378504645\\
-3.70095033816051	-4.57570378504645\\
-3.47339077223625	4.24426931714046\\
-3.47339077223625	-4.24426931714046\\
-3.66075923902509	5.4545742653818\\
-3.66075923902509	-5.4545742653818\\
-4.53659641051319	0\\
-2.90138593113153	1.79933093075275\\
-2.90138593113153	-1.79933093075275\\
-1.39915948481584	1.39696373320306\\
-1.39915948481584	-1.39696373320306\\
-7.99973985862661	0\\
-2.93237864257483	3.02950805838526\\
-2.93237864257483	-3.02950805838526\\
-3.9323233530702	0\\
-1.20503838890221	0.851663011921531\\
-1.20503838890221	-0.851663011921531\\
-5.71913444511238	0\\
-2.0405479992425	1.86511539911919\\
-2.0405479992425	-1.86511539911919\\
-1.79581743435879	0.306750021551951\\
-1.79581743435879	-0.306750021551951\\
-2.90468304042061	6.23861069448391\\
-2.90468304042061	-6.23861069448391\\
-3.40591682859963	2.52179024091969\\
-3.40591682859963	-2.52179024091969\\
-4.0030382689836	3.74996045465801\\
-4.0030382689836	-3.74996045465801\\
-3.50568314003054	5.59709213961625\\
-3.50568314003054	-5.59709213961625\\
-7.43694626002013	0\\
-1.62760858652332	0.617428587083638\\
-1.62760858652332	-0.617428587083638\\
-9.82159747874191	0\\
-12.4104312651263	0\\
-1.02872861684247	2.52718099684419\\
-1.02872861684247	-2.52718099684419\\
-1.53139463459633	1.9630507112951\\
-1.53139463459633	-1.9630507112951\\
-9.57972742220037	0\\
-7.3148232154803	0\\
-1.3908373108598	2.36784654546284\\
-1.3908373108598	-2.36784654546284\\
-4.34491612190152	3.34789837981236\\
-4.34491612190152	-3.34789837981236\\
-1.20366134864266	1.26399398755055\\
-1.20366134864266	-1.26399398755055\\
-8.99951082710679	0\\
-1.2571727558107	2.62502681861909\\
-1.2571727558107	-2.62502681861909\\
-2.20675877319042	2.01467815610886\\
-2.20675877319042	-2.01467815610886\\
-10.7269226619264	0\\
-3.31461200190433	4.33688868103283\\
-3.31461200190433	-4.33688868103283\\
-1.49076683532864	1.12219238131127\\
-1.49076683532864	-1.12219238131127\\
-1.73337112851814	2.64665066745502\\
-1.73337112851814	-2.64665066745502\\
-1.11393887447657	0.65426841417983\\
-1.11393887447657	-0.65426841417983\\
-4.20403214336011	0.876317379753793\\
-4.20403214336011	-0.876317379753793\\
-2.40762340966941	6.31947464942862\\
-2.40762340966941	-6.31947464942862\\
-11.8405812900048	0\\
-2.96502608646397	3.19865151436848\\
-2.96502608646397	-3.19865151436848\\
-3.74454517408327	0\\
-3.1576168274224	4.26133772060713\\
-3.1576168274224	-4.26133772060713\\
-1.4656256954276	1.01840463087869\\
-1.4656256954276	-1.01840463087869\\
-7.77963439751013	0\\
-1.75796095695913	0.173217127300854\\
-1.75796095695913	-0.173217127300854\\
-4.41530972622374	2.94235645257478\\
-4.41530972622374	-2.94235645257478\\
-10.5832506750551	0\\
-3.79063747123201	3.77867830326971\\
-3.79063747123201	-3.77867830326971\\
-3.02625297677454	3.89476436374194\\
-3.02625297677454	-3.89476436374194\\
-1.60996183419404	2.21070965170107\\
-1.60996183419404	-2.21070965170107\\
-2.2156051686428	2.82580277007351\\
-2.2156051686428	-2.82580277007351\\
-1.22068819264667	0.748728618297427\\
-1.22068819264667	-0.748728618297427\\
-2.80634795380066	5.85809190590073\\
-2.80634795380066	-5.85809190590073\\
-2.57166732862232	6.20726489304856\\
-2.57166732862232	-6.20726489304856\\
-1.1067841553268	2.61353660810314\\
-1.1067841553268	-2.61353660810314\\
-4.87794821358099	1.18746441760402\\
-4.87794821358099	-1.18746441760402\\
-4.33505146076581	3.84731028434317\\
-4.33505146076581	-3.84731028434317\\
-5.31385405588853	0\\
-2.07808552337718	2.31649424882197\\
-2.07808552337718	-2.31649424882197\\
-1.66277525959068	0.343818840023976\\
-1.66277525959068	-0.343818840023976\\
-13.0998603955261	0\\
-3.28831036836213	4.72796141650027\\
-3.28831036836213	-4.72796141650027\\
-1.59226960182548	0\\
-3.67250397155126	3.65315197845833\\
-3.67250397155126	-3.65315197845833\\
-1.23616417728665	1.00825646886644\\
-1.23616417728665	-1.00825646886644\\
-6.88215809029575	0\\
-3.444896265294	0\\
-13.2078077103436	0\\
-4.76632401375363	1.67634816828411\\
-4.76632401375363	-1.67634816828411\\
-1.39542461068843	1.55513438511807\\
-1.39542461068843	-1.55513438511807\\
-4.73061399395812	2.3021972443689\\
-4.73061399395812	-2.3021972443689\\
-3.4432055010966	5.11353565421645\\
-3.4432055010966	-5.11353565421645\\
-3.03450837155887	4.94306848436492\\
-3.03450837155887	-4.94306848436492\\
-3.10383433479511	2.37089796442375\\
-3.10383433479511	-2.37089796442375\\
-3.99959494104811	4.43530942241883\\
-3.99959494104811	-4.43530942241883\\
-5.95402617961494	0\\
-1.2458753907088	1.97790624614114\\
-1.2458753907088	-1.97790624614114\\
-1.52883136055141	0.671526598200505\\
-1.52883136055141	-0.671526598200505\\
-1.023346562917	2.24239954679507\\
-1.023346562917	-2.24239954679507\\
-1.29763336352979	0.5569789801463\\
-1.29763336352979	-0.5569789801463\\
-2.64289327844571	6.07248370297205\\
-2.64289327844571	-6.07248370297205\\
-1.61642076183291	1.67803629681062\\
-1.61642076183291	-1.67803629681062\\
-1.87027633717836	0.599045996078518\\
-1.87027633717836	-0.599045996078518\\
-4.75136058456022	1.80796510023912\\
-4.75136058456022	-1.80796510023912\\
-2.89976929212116	3.67251838803899\\
-2.89976929212116	-3.67251838803899\\
-2.90119032457732	2.83799072036629\\
-2.90119032457732	-2.83799072036629\\
-8.26886959179725	0\\
-4.11660627664511	4.76679704233966\\
-4.11660627664511	-4.76679704233966\\
-1.07852215168163	1.37628823863767\\
-1.07852215168163	-1.37628823863767\\
-3.24746267600398	3.72557693375764\\
-3.24746267600398	-3.72557693375764\\
-1.57068414190247	1.01764322368844\\
-1.57068414190247	-1.01764322368844\\
-12.078932733172	0\\
-4.47381813288888	4.05256378079153\\
-4.47381813288888	-4.05256378079153\\
-3.35270907622206	3.79995061879867\\
-3.35270907622206	-3.79995061879867\\
-4.30333151165106	0\\
-3.53980615972287	3.16796657524473\\
-3.53980615972287	-3.16796657524473\\
-3.55364302618724	3.61939733886443\\
-3.55364302618724	-3.61939733886443\\
-1.58692031791893	1.32033447070593\\
-1.58692031791893	-1.32033447070593\\
-7.00555675639452	0\\
-1.96824440031852	2.70757427847081\\
-1.96824440031852	-2.70757427847081\\
-3.97923866607168	3.2791434849171\\
-3.97923866607168	-3.2791434849171\\
-1.37847673668209	1.09349836013949\\
-1.37847673668209	-1.09349836013949\\
-1.8975182756401	2.20123832453099\\
-1.8975182756401	-2.20123832453099\\
-2.98688307116883	4.13177592875215\\
-2.98688307116883	-4.13177592875215\\
-2.49244570367808	0\\
-3.03682384645526	3.57949391808277\\
-3.03682384645526	-3.57949391808277\\
-3.47121299816779	2.12518006772511\\
-3.47121299816779	-2.12518006772511\\
-3.81724410679483	2.92680999482877\\
-3.81724410679483	-2.92680999482877\\
-4.41387395443922	2.2913808840931\\
-4.41387395443922	-2.2913808840931\\
-3.74453787600959	4.7748588664171\\
-3.74453787600959	-4.7748588664171\\
-1.82506815759651	0.777045669672989\\
-1.82506815759651	-0.777045669672989\\
-3.64542956861998	3.05734459057262\\
-3.64542956861998	-3.05734459057262\\
-1.60285990239753	0.172675299199163\\
-1.60285990239753	-0.172675299199163\\
-3.49284880657156	2.88459136905632\\
-3.49284880657156	-2.88459136905632\\
-3.80084979292969	1.7770157558873\\
-3.80084979292969	-1.7770157558873\\
-1.77652268014122	1.45724695486598\\
-1.77652268014122	-1.45724695486598\\
-4.34885106917888	4.19629547229736\\
-4.34885106917888	-4.19629547229736\\
-4.26125218582024	4.26142283559649\\
-4.26125218582024	-4.26142283559649\\
-1.14057359295133	1.46840949189176\\
-1.14057359295133	-1.46840949189176\\
-4.10760954592522	0\\
-1.52225828336016	2.13109381700582\\
-1.52225828336016	-2.13109381700582\\
-3.75997110797962	3.1093702441224\\
-3.75997110797962	-3.1093702441224\\
-1.3944375618889	2.66788219853223\\
-1.3944375618889	-2.66788219853223\\
-1.11989902624251	2.1939290444177\\
-1.11989902624251	-2.1939290444177\\
-3.61417619416609	5.03872529767036\\
-3.61417619416609	-5.03872529767036\\
-1.54668807854184	0.918892929973407\\
-1.54668807854184	-0.918892929973407\\
-2.16328330508494	2.60993187887553\\
-2.16328330508494	-2.60993187887553\\
-3.78885661893627	5.02279134200919\\
-3.78885661893627	-5.02279134200919\\
-2.38864726232602	0.905278286673296\\
-2.38864726232602	-0.905278286673296\\
-2.16879370127416	2.93084460307448\\
-2.16879370127416	-2.93084460307448\\
-10.1208398094221	0\\
-4.50779142520005	3.14724337133576\\
-4.50779142520005	-3.14724337133576\\
-1.45832021886773	0.414383533955214\\
-1.45832021886773	-0.414383533955214\\
-2.77284178946799	2.13416895181218\\
-2.77284178946799	-2.13416895181218\\
-6.12659975899054	0\\
-1.66445418958227	1.86941820364818\\
-1.66445418958227	-1.86941820364818\\
-4.67352736920872	3.38678425198956\\
-4.67352736920872	-3.38678425198956\\
-1.30078352871365	2.0997001997939\\
-1.30078352871365	-2.0997001997939\\
-4.64653591950403	1.9829021167722\\
-4.64653591950403	-1.9829021167722\\
-10.9618222651284	0\\
-2.74265439833512	0\\
-1.60845249787067	1.51553027124907\\
-1.60845249787067	-1.51553027124907\\
-1.99162786031253	0.673708306147346\\
-1.99162786031253	-0.673708306147346\\
-4.28741050032511	1.82065533672507\\
-4.28741050032511	-1.82065533672507\\
-1.97059150510352	0\\
-1.66269183622449	2.31103032787218\\
-1.66269183622449	-2.31103032787218\\
-3.20336476749409	5.70131194089347\\
-3.20336476749409	-5.70131194089347\\
-4.44982711873819	2.71592703723607\\
-4.44982711873819	-2.71592703723607\\
-4.59950891974487	1.20438012703621\\
-4.59950891974487	-1.20438012703621\\
-1.43749394426637	1.73642942930865\\
-1.43749394426637	-1.73642942930865\\
-6.2649850299849	0\\
-4.228931689065	3.83210276341072\\
-4.228931689065	-3.83210276341072\\
-4.83229589753926	0.627110997568328\\
-4.83229589753926	-0.627110997568328\\
-3.4312613884276	5.44421149257834\\
-3.4312613884276	-5.44421149257834\\
-3.63569058670691	4.32924333351874\\
-3.63569058670691	-4.32924333351874\\
-2.0612606157019	2.66530578865178\\
-2.0612606157019	-2.66530578865178\\
-4.1241047727326	3.65525641855508\\
-4.1241047727326	-3.65525641855508\\
-2.55250490205819	1.69478302781866\\
-2.55250490205819	-1.69478302781866\\
-2.88267517040764	0\\
-2.49531100525692	6.5459960477727\\
-2.49531100525692	-6.5459960477727\\
-6.77059693558963	0\\
-8.11644012392948	0\\
-3.10873865965212	5.26379478743326\\
-3.10873865965212	-5.26379478743326\\
-4.24122842977028	2.58158473331121\\
-4.24122842977028	-2.58158473331121\\
-4.80428956662215	1.08145532882112\\
-4.80428956662215	-1.08145532882112\\
-2.64929594972127	6.44467754460696\\
-2.64929594972127	-6.44467754460696\\
-3.14773266683916	5.84681773181212\\
-3.14773266683916	-5.84681773181212\\
-2.91209352868832	5.11490275677146\\
-2.91209352868832	-5.11490275677146\\
-1.87689166131854	2.03521853074058\\
-1.87689166131854	-2.03521853074058\\
-2.98546273904493	5.7492280316891\\
-2.98546273904493	-5.7492280316891\\
-11.385488625352	0\\
-3.93808656159202	4.0756094217089\\
-3.93808656159202	-4.0756094217089\\
-1.29361463895751	1.45862444466792\\
-1.29361463895751	-1.45862444466792\\
-4.01932243518075	3.92634502547257\\
-4.01932243518075	-3.92634502547257\\
-2.23781737675654	2.54149055739172\\
-2.23781737675654	-2.54149055739172\\
-3.50776146539773	4.33907643275454\\
-3.50776146539773	-4.33907643275454\\
-8.51270105046358	0\\
-10.4826724428354	0\\
-1.59070726464807	2.61184813996928\\
-1.59070726464807	-2.61184813996928\\
-4.19031519621508	4.52728932746652\\
-4.19031519621508	-4.52728932746652\\
-12.9446734409071	0\\
-2.35636632355388	2.17606427095214\\
-2.35636632355388	-2.17606427095214\\
-3.24120024523131	5.00226049471953\\
-3.24120024523131	-5.00226049471953\\
-4.36785257287107	3.46871380790474\\
-4.36785257287107	-3.46871380790474\\
-1.58551103941889	2.72112462216329\\
-1.58551103941889	-2.72112462216329\\
-3.3658890905037	2.85177877253603\\
-3.3658890905037	-2.85177877253603\\
-2.90660411608396	6.06724922066161\\
-2.90660411608396	-6.06724922066161\\
-3.55239160279582	3.32081753797117\\
-3.55239160279582	-3.32081753797117\\
-3.03713973548794	6.07664081598339\\
-3.03713973548794	-6.07664081598339\\
-0.922173307128959	2.29547933207441\\
-0.922173307128959	-2.29547933207441\\
-3.81911638706925	1.57487953585818\\
-3.81911638706925	-1.57487953585818\\
-2.77549078950202	5.71347977692938\\
-2.77549078950202	-5.71347977692938\\
-3.49906033201922	4.73365901741633\\
-3.49906033201922	-4.73365901741633\\
-3.68806683624091	4.16898426114537\\
-3.68806683624091	-4.16898426114537\\
-1.42549470605121	1.21911261181851\\
-1.42549470605121	-1.21911261181851\\
-1.42911656943807	2.53125162220505\\
-1.42911656943807	-2.53125162220505\\
-3.43273759094108	3.13847703579309\\
-3.43273759094108	-3.13847703579309\\
-4.22507691847455	3.52774339865088\\
-4.22507691847455	-3.52774339865088\\
-4.42001089069601	0\\
-2.53097291348604	2.75068311346696\\
-2.53097291348604	-2.75068311346696\\
-3.37746919384389	5.35400983660074\\
-3.37746919384389	-5.35400983660074\\
-11.4883890021678	0\\
-4.72622837820173	0\\
-2.45886717067123	3.0474099531887\\
-2.45886717067123	-3.0474099531887\\
-3.18433098555105	5.97136288489666\\
-3.18433098555105	-5.97136288489666\\
-3.85536175653329	3.50209749051504\\
-3.85536175653329	-3.50209749051504\\
-3.68607532725209	3.81615313032728\\
-3.68607532725209	-3.81615313032728\\
-1.94384694022416	1.75399354872512\\
-1.94384694022416	-1.75399354872512\\
-3.82689821146086	3.87748265147498\\
-3.82689821146086	-3.87748265147498\\
-10.2389371472495	0\\
-3.00885206001643	5.18495953797607\\
-3.00885206001643	-5.18495953797607\\
-1.91278313990302	1.57220755082753\\
-1.91278313990302	-1.57220755082753\\
-2.22839337162502	0.493967293060277\\
-2.22839337162502	-0.493967293060277\\
-4.90009176664539	1.40402855325865\\
-4.90009176664539	-1.40402855325865\\
-3.86823842538515	4.39293409874792\\
-3.86823842538515	-4.39293409874792\\
-2.98254225706062	4.7941708896024\\
-2.98254225706062	-4.7941708896024\\
-4.00883358654975	4.55904542604568\\
-4.00883358654975	-4.55904542604568\\
-13.5809606860453	0\\
-1.61057000309085	0.769083615120514\\
-1.61057000309085	-0.769083615120514\\
-12.3070675758179	0\\
-2.80533559758762	1.40206539111268\\
-2.80533559758762	-1.40206539111268\\
-1.77037728402447	2.28307598716468\\
-1.77037728402447	-2.28307598716468\\
-3.80035875142432	2.80183537878175\\
-3.80035875142432	-2.80183537878175\\
-3.45380624280749	3.76821622836083\\
-3.45380624280749	-3.76821622836083\\
-4.18678732943143	1.65346011278396\\
-4.18678732943143	-1.65346011278396\\
-3.35203017918215	4.89993237795111\\
-3.35203017918215	-4.89993237795111\\
-4.09865618311136	4.34692945727933\\
-4.09865618311136	-4.34692945727933\\
-3.4815143542755	4.54302579144088\\
-3.4815143542755	-4.54302579144088\\
-4.53141395764068	2.32701408130569\\
-4.53141395764068	-2.32701408130569\\
-1.96067783401796	0.208910464146542\\
-1.96067783401796	-0.208910464146542\\
-3.22081197145328	4.57402476138903\\
-3.22081197145328	-4.57402476138903\\
-4.60913512091337	3.65248192722924\\
-4.60913512091337	-3.65248192722924\\
-1.96266296662569	2.35754169400639\\
-1.96266296662569	-2.35754169400639\\
-1.71930338291112	2.76273790530078\\
-1.71930338291112	-2.76273790530078\\
-3.20435739976236	3.57806356463672\\
-3.20435739976236	-3.57806356463672\\
-3.98009427592488	4.32840655525268\\
-3.98009427592488	-4.32840655525268\\
-2.54906295981594	6.32332777687133\\
-2.54906295981594	-6.32332777687133\\
-4.04454443561492	2.32828121389138\\
-4.04454443561492	-2.32828121389138\\
-3.34906211469487	3.31099537381666\\
-3.34906211469487	-3.31099537381666\\
-1.75353132668026	1.21212310361631\\
-1.75353132668026	-1.21212310361631\\
-1.58470241086883	0.422732438872902\\
-1.58470241086883	-0.422732438872902\\
-12.6440094156264	0\\
-6.53261412591073	0\\
-3.4290018463697	4.9811603291729\\
-3.4290018463697	-4.9811603291729\\
-1.68815009031348	1.60671108591503\\
-1.68815009031348	-1.60671108591503\\
-3.67987882451876	2.10844200030388\\
-3.67987882451876	-2.10844200030388\\
-4.14383219603037	3.91625349866205\\
-4.14383219603037	-3.91625349866205\\
-1.83605497404057	2.5888807970735\\
-1.83605497404057	-2.5888807970735\\
-4.01273450254306	1.80216112497401\\
-4.01273450254306	-1.80216112497401\\
-2.90214986642016	5.01089641758487\\
-2.90214986642016	-5.01089641758487\\
-3.27246888045255	4.47738967866493\\
-3.27246888045255	-4.47738967866493\\
-2.87204795961636	3.55931394173788\\
-2.87204795961636	-3.55931394173788\\
-4.61182813558349	1.30552088550414\\
-4.61182813558349	-1.30552088550414\\
-2.29837449709602	3.00360523433145\\
-2.29837449709602	-3.00360523433145\\
-11.9657790193045	0\\
-3.1396197759671	5.06486316059242\\
-3.1396197759671	-5.06486316059242\\
-2.01634347249718	0.548624947005758\\
-2.01634347249718	-0.548624947005758\\
-3.05492285385279	4.41327014868482\\
-3.05492285385279	-4.41327014868482\\
-2.8833118665051	5.27466967520789\\
-2.8833118665051	-5.27466967520789\\
-4.13650323689643	4.17037912407534\\
-4.13650323689643	-4.17037912407534\\
-4.29613392228634	4.36767618946941\\
-4.29613392228634	-4.36767618946941\\
-4.51471454417413	2.98077591621445\\
-4.51471454417413	-2.98077591621445\\
-4.76220509536798	2.92564769926053\\
-4.76220509536798	-2.92564769926053\\
-3.97622880330569	2.76799639924534\\
-3.97622880330569	-2.76799639924534\\
-2.4981399229558	1.58002749841846\\
-2.4981399229558	-1.58002749841846\\
-2.7268878358848	5.61933614877393\\
-2.7268878358848	-5.61933614877393\\
-2.91800584578074	3.45175052308776\\
-2.91800584578074	-3.45175052308776\\
-2.66997754708595	6.33770931081365\\
-2.66997754708595	-6.33770931081365\\
-3.86252736619932	4.70022920402693\\
-3.86252736619932	-4.70022920402693\\
-4.78703118891768	0.240084500241956\\
-4.78703118891768	-0.240084500241956\\
-3.44938125725173	4.43530920347856\\
-3.44938125725173	-4.43530920347856\\
-4.05311565834397	2.48860483518358\\
-4.05311565834397	-2.48860483518358\\
-1.89642182761668	2.49960972434424\\
-1.89642182761668	-2.49960972434424\\
-3.65258524132622	4.04325655390563\\
-3.65258524132622	-4.04325655390563\\
-3.34473484203499	0\\
-5.21222069623663	0\\
-2.52354284542046	6.43902927617419\\
-2.52354284542046	-6.43902927617419\\
-4.5394363674676	3.86202671797307\\
-4.5394363674676	-3.86202671797307\\
-3.11488004974726	3.43850668892369\\
-3.11488004974726	-3.43850668892369\\
-2.28973810514553	2.43193737457403\\
-2.28973810514553	-2.43193737457403\\
-3.65870438080735	3.48284352372485\\
-3.65870438080735	-3.48284352372485\\
-4.41135943441631	1.97467255285159\\
-4.41135943441631	-1.97467255285159\\
-2.62160251173151	0\\
-2.8126002089626	1.24687071986028\\
-2.8126002089626	-1.24687071986028\\
-3.31610200028073	5.69331659097899\\
-3.31610200028073	-5.69331659097899\\
-2.77673053737006	5.52119126855807\\
-2.77673053737006	-5.52119126855807\\
-2.54725575648624	1.97850110386207\\
-2.54725575648624	-1.97850110386207\\
-2.66311766064221	2.64432967624454\\
-2.66311766064221	-2.64432967624454\\
-1.89203615984017	0.467079797947391\\
-1.89203615984017	-0.467079797947391\\
-3.10260105987512	5.63858734604008\\
-3.10260105987512	-5.63858734604008\\
-2.67390495715762	2.76952191700653\\
-2.67390495715762	-2.76952191700653\\
-2.03624801702579	2.80651954686753\\
-2.03624801702579	-2.80651954686753\\
-3.4490427491027	5.72344027404186\\
-3.4490427491027	-5.72344027404186\\
-3.06916086076919	5.38247911922305\\
-3.06916086076919	-5.38247911922305\\
-3.08087163479963	-5.37559335323635\\
-3.59926203316173	4.58393043335917\\
-3.59926203316173	-4.58393043335917\\
-1.51856252536222	2.25319251417666\\
-1.51856252536222	-2.25319251417666\\
-2.26058798962691	2.71978867021503\\
-2.26058798962691	-2.71978867021503\\
-3.19610783662432	2.50344212766712\\
-3.19610783662432	-2.50344212766712\\
-2.14437709700699	1.09449909578183\\
-2.14437709700699	-1.09449909578183\\
-2.60115004854802	2.94779323796178\\
-2.60115004854802	-2.94779323796178\\
-4.33318717170265	0.782775887514959\\
-4.33318717170265	-0.782775887514959\\
-2.42845846892296	2.67023457310736\\
-2.42845846892296	-2.67023457310736\\
-2.8084790622877	6.14253380594987\\
-2.8084790622877	-6.14253380594987\\
-3.89289304947139	5.00861637015726\\
-3.89289304947139	-5.00861637015726\\
-2.9164807595823	2.43778070647286\\
-2.9164807595823	-2.43778070647286\\
-3.57400389335335	4.47662591050017\\
-3.57400389335335	-4.47662591050017\\
-3.20974846007624	4.13906045408526\\
-3.20974846007624	-4.13906045408526\\
-1.01988958034187	2.3440612518875\\
-1.01988958034187	-2.3440612518875\\
-3.84926776672047	1.24517259078886\\
-3.84926776672047	-1.24517259078886\\
-4.07909819798835	2.99455454143162\\
-4.07909819798835	-2.99455454143162\\
-3.2248694840501	4.80767440060623\\
-3.2248694840501	-4.80767440060623\\
-3.46748594945292	3.46211751434412\\
-3.46748594945292	-3.46211751434412\\
-3.80225126017768	5.18851616642609\\
-3.80225126017768	-5.18851616642609\\
-3.4195289713044	4.82025009930245\\
-3.4195289713044	-4.82025009930245\\
-3.70555883358958	5.31492681350139\\
-3.70555883358958	-5.31492681350139\\
-2.76509752839458	6.40886846968235\\
-2.76509752839458	-6.40886846968235\\
-9.72147642860548	0\\
-3.79391267070238	4.47266589851884\\
-3.79391267070238	-4.47266589851884\\
-3.85211891737006	1.44241060301193\\
-3.85211891737006	-1.44241060301193\\
-1.67298779697275	1.39627809157919\\
-1.67298779697275	-1.39627809157919\\
-1.49802677331379	1.58134185763234\\
-1.49802677331379	-1.58134185763234\\
-2.61766936194727	3.08302754212387\\
-2.61766936194727	-3.08302754212387\\
-2.85273744266106	5.63845026330226\\
-2.85273744266106	-5.63845026330226\\
-3.82512165552694	3.60561560959597\\
-3.82512165552694	-3.60561560959597\\
-4.17650283653655	3.02535629434306\\
-4.17650283653655	-3.02535629434306\\
-4.27423671173064	3.24447277668587\\
-4.27423671173064	-3.24447277668587\\
-3.41760325791018	2.69502038632271\\
-3.41760325791018	-2.69502038632271\\
-3.2071010446323	2.85990858871534\\
-3.2071010446323	-2.85990858871534\\
-3.11911539937128	3.65989929894881\\
-3.11911539937128	-3.65989929894881\\
-3.1661006833561	3.2577282134645\\
-3.1661006833561	-3.2577282134645\\
-4.32367215653736	2.87229212077672\\
-4.32367215653736	-2.87229212077672\\
-4.5569499341128	3.41024684754384\\
-4.5569499341128	-3.41024684754384\\
-3.27986131376281	3.11586898467896\\
-3.27986131376281	-3.11586898467896\\
-3.31133302899203	5.87619010129032\\
-3.31133302899203	-5.87619010129032\\
-3.1452404279098	5.51526960446096\\
-3.1452404279098	-5.51526960446096\\
-2.15044495701693	0.937185408862434\\
-2.15044495701693	-0.937185408862434\\
-3.31840563786379	3.92725887810477\\
-3.31840563786379	-3.92725887810477\\
-1.16143833534988	2.10274545450331\\
-1.16143833534988	-2.10274545450331\\
-4.89364915089534	1.51794181709347\\
-4.89364915089534	-1.51794181709347\\
-1.82786377391022	1.72207067356384\\
-1.82786377391022	-1.72207067356384\\
-4.09520454310972	4.46832584576361\\
-4.09520454310972	-4.46832584576361\\
-2.79494866006787	1.67591310251386\\
-2.79494866006787	-1.67591310251386\\
-3.13844632916562	4.06719859633008\\
-3.13844632916562	-4.06719859633008\\
-2.70554359147723	5.83402194578185\\
-2.70554359147723	-5.83402194578185\\
-5.04179848858048	1.21613289067331\\
-5.04179848858048	-1.21613289067331\\
-3.57709940808161	2.82638333275018\\
-3.57709940808161	-2.82638333275018\\
-3.94896134752653	4.21716861828511\\
-3.94896134752653	-4.21716861828511\\
-2.75963805845308	1.13891210675394\\
-2.75963805845308	-1.13891210675394\\
-1.67552293556921	1.11863832417384\\
-1.67552293556921	-1.11863832417384\\
-2.55585929500942	1.81090645598317\\
-2.55585929500942	-1.81090645598317\\
-3.20516590740396	3.38185291042792\\
-3.20516590740396	-3.38185291042792\\
-4.61714455145041	3.1413620297855\\
};

\end{axis}

\end{tikzpicture}%

%% file: figure/CCF_cpa_10.tex
%
%
\begin{tikzpicture}

\begin{axis}[%
width=4.521in,
height=3.566in,
at={(0.758in,0.481in)},
scale only axis,
xmin=-20,
xmax=0,
ymin=-10,
ymax=10,
xtick={-20,-15,-10,-5,0},
ytick={-10,-5,0,5,10},
axis background/.style={fill=white},
legend style={legend cell align=left, align=left, draw=white!15!black}
]
\addplot [color=black, draw=none, mark=asterisk, mark options={solid, black}]
  table[row sep=crcr]{%
-9.73666370266562	0\\
-1.5189844492995	2.25688748935236\\
-1.5189844492995	-2.25688748935236\\
-1.28978278090554	0.561205645153733\\
-1.28978278090554	-0.561205645153733\\
};

\addplot+[color=blue, draw=none, mark=*, mark size=1pt,mark options={solid, blue}]
  table[row sep=crcr]{%
-3.00000000000002	0.999999999999957\\
-3.00000000000002	-0.999999999999957\\
-2.9999999999999	0\\
-2.00000000000002	1.00000000000002\\
-2.00000000000002	-1.00000000000002\\
0	0\\
-14.3727249368743	0\\
-0.895308006699595	2.46190392053095\\
-0.895308006699595	-2.46190392053095\\
-1.02663164518954	0.711100376602991\\
-1.02663164518954	-0.711100376602991\\
-18.013706705844	0\\
-0.611474945118964	2.2126103372819\\
-0.611474945118964	-2.2126103372819\\
-1.12009554083502	0.601148211649915\\
-1.12009554083502	-0.601148211649915\\
-12.4782775392497	0\\
-0.842007936839689	2.27426422002525\\
-0.842007936839689	-2.27426422002525\\
-1.40477613267454	0.335471938008631\\
-1.40477613267454	-0.335471938008631\\
-3.01617479725018	3.36433306530965\\
-3.01617479725018	-3.36433306530965\\
-1.99965218435422	0\\
-1.67556389887511	0.951762798074744\\
-1.67556389887511	-0.951762798074744\\
-15.0508153315501	0\\
-0.676169820572601	2.55512961235664\\
-0.676169820572601	-2.55512961235664\\
-1.08566739458424	0.481666795474951\\
-1.08566739458424	-0.481666795474951\\
-2.40067789000289	5.60246825128531\\
-2.40067789000289	-5.60246825128531\\
-1.57172193004949	0\\
-1.19691546644129	1.08374860182364\\
-1.19691546644129	-1.08374860182364\\
-13.6942834791959	0\\
-0.796595382300658	2.06528154843391\\
-0.796595382300658	-2.06528154843391\\
-8.78912552435428	0\\
-1.56783162859263	2.34005491903943\\
-1.56783162859263	-2.34005491903943\\
-1.24364893826446	0.84047687596446\\
-1.24364893826446	-0.84047687596446\\
-2.40891268724502	6.17966998518187\\
-2.40891268724502	-6.17966998518187\\
-1.7450413074142	0\\
-0.967046187977416	1.02500746302197\\
-0.967046187977416	-1.02500746302197\\
-10.8873788816056	0\\
-1.07808364421133	2.49302637444727\\
-1.07808364421133	-2.49302637444727\\
-1.31063022854451	0.237741855711192\\
-1.31063022854451	-0.237741855711192\\
-8.65014945089941	0\\
-1.25489065127165	2.28962047261812\\
-1.25489065127165	-2.28962047261812\\
-1.58589618185437	0.332697163614529\\
-1.58589618185437	-0.332697163614529\\
-1.98678302074027	6.66241369390492\\
-1.98678302074027	-6.66241369390492\\
-1.4099076754522	0\\
-0.965767613702744	1.15082478733216\\
-0.965767613702744	-1.15082478733216\\
-5.4115756910543	0\\
-2.58855511641007	3.00540987002929\\
-2.58855511641007	-3.00540987002929\\
-1.03432857915562	0.861179341927887\\
-1.03432857915562	-0.861179341927887\\
-4.24552189316566	1.05613580325602\\
-4.24552189316566	-1.05613580325602\\
-1.57100781791298	2.03687802680819\\
-1.57100781791298	-2.03687802680819\\
-1.17524163365056	0\\
-10.6800662914244	0\\
-1.39475569557182	2.52251405680962\\
-1.39475569557182	-2.52251405680962\\
-4.85775719373584	2.60185707733355\\
-4.85775719373584	-2.60185707733355\\
-1.33539852890382	1.53120640605067\\
-1.33539852890382	-1.53120640605067\\
-5.75490152081015	0\\
-1.8401223391156	3.21131039511441\\
-1.8401223391156	-3.21131039511441\\
-1.22854876227147	0.631708034574451\\
-1.22854876227147	-0.631708034574451\\
-7.60934156715525	0\\
-1.85445309772419	3.01006450275057\\
-1.85445309772419	-3.01006450275057\\
-1.50409595733382	2.91007522766636\\
-1.50409595733382	-2.91007522766636\\
-2.56958780346674	6.17124481075353\\
-2.56958780346674	-6.17124481075353\\
-1.09424844202867	1.34705983523988\\
-1.09424844202867	-1.34705983523988\\
-19.4955094663219	0\\
-0.486319652838699	2.33691977757616\\
-0.486319652838699	-2.33691977757616\\
-2.25966130193118	5.61544225891315\\
-2.25966130193118	-5.61544225891315\\
-1.8904730967907	0\\
-1.17035183388569	0.931979683735712\\
-1.17035183388569	-0.931979683735712\\
-1.81368177694457	7.3601395257498\\
-1.81368177694457	-7.3601395257498\\
-0.854962723982921	1.25034895356022\\
-0.854962723982921	-1.25034895356022\\
-1.89899620475454	7.46155540794183\\
-1.89899620475454	-7.46155540794183\\
-14.0584572591523	0\\
-18.6962362345763	0\\
-3.72433672004245	5.47537852788021\\
-3.72433672004245	-5.47537852788021\\
-1.290794490223	0\\
-2.66143029203193	6.69609216641993\\
-2.66143029203193	-6.69609216641993\\
-0.953133556863122	1.34344623446195\\
-0.953133556863122	-1.34344623446195\\
-2.73906986124725	6.29084269294824\\
-2.73906986124725	-6.29084269294824\\
-1.02411511070952	0\\
-1.16768853650397	7.34566658187866\\
-1.16768853650397	-7.34566658187866\\
-3.60285293439651	4.35125609422447\\
-3.60285293439651	-4.35125609422447\\
-1.38624278470147	1.31989647611004\\
-1.38624278470147	-1.31989647611004\\
-3.50369947881003	5.60806378291086\\
-3.50369947881003	-5.60806378291086\\
-18.9284882091117	0\\
-1.51736862005376	3.25011435373427\\
-1.51736862005376	-3.25011435373427\\
-3.05793834079977	5.69104545851988\\
-3.05793834079977	-5.69104545851988\\
-6.48871801613405	0\\
-1.71350432698588	3.37505641467956\\
-1.71350432698588	-3.37505641467956\\
-11.3207378404065	0\\
-1.31937720663009	0.416988382429992\\
-1.31937720663009	-0.416988382429992\\
-1.51696462707442	3.01475283654938\\
-1.51696462707442	-3.01475283654938\\
-8.26781618412031	0\\
-1.7206624130103	2.64138775119328\\
-1.7206624130103	-2.64138775119328\\
-2.34395461050213	4.20507572063294\\
-2.34395461050213	-4.20507572063294\\
-3.16718614994042	0\\
-7.8387833751294	0\\
-3.95732074880017	0\\
-0.987960875184022	1.71611900201377\\
-0.987960875184022	-1.71611900201377\\
-2.67395555595435	5.79584938779675\\
-2.67395555595435	-5.79584938779675\\
-1.18473659050544	2.12303994941706\\
-1.18473659050544	-2.12303994941706\\
-1.42363257642265	0.691383129791022\\
-1.42363257642265	-0.691383129791022\\
-2.98359237418929	6.70663019665619\\
-2.98359237418929	-6.70663019665619\\
-3.2294938031056	4.88167471667654\\
-3.2294938031056	-4.88167471667654\\
-1.08553708964512	1.10454961295165\\
-1.08553708964512	-1.10454961295165\\
-1.42154808129341	2.63026449518221\\
-1.42154808129341	-2.63026449518221\\
-2.27412296055804	7.17456689372749\\
-2.27412296055804	-7.17456689372749\\
-14.6013013529883	0\\
-1.22155735657312	0.389398646114148\\
-1.22155735657312	-0.389398646114148\\
-13.2855311382976	0\\
-2.71194076239804	6.16944147070209\\
-2.71194076239804	-6.16944147070209\\
-5.3154845469528	2.81129371802345\\
-5.3154845469528	-2.81129371802345\\
-1.06305273920841	1.64093692815627\\
-1.06305273920841	-1.64093692815627\\
-10.4436136981219	0\\
-7.40889979508026	0\\
-2.14547995766809	2.47156889792267\\
-2.14547995766809	-2.47156889792267\\
-1.28881619991992	7.28036586438508\\
-1.28881619991992	-7.28036586438508\\
-6.85926576481174	0\\
-1.73713075463416	2.96411270465192\\
-1.73713075463416	-2.96411270465192\\
-16.1787681875985	0\\
-2.86945185036858	5.63235381275687\\
-2.86945185036858	-5.63235381275687\\
-1.07356509701176	0.993234158876594\\
-1.07356509701176	-0.993234158876594\\
-0.177660464454607	7.99931279008585\\
-0.177660464454607	-7.99931279008585\\
-4.41714318572232	3.51051108223129\\
-4.41714318572232	-3.51051108223129\\
-1.25397487499177	1.1811577533424\\
-1.25397487499177	-1.1811577533424\\
-0.741460537931313	7.97802288790938\\
-0.741460537931313	-7.97802288790938\\
-9.93734325829106	0\\
-1.12744728675007	2.80133374683846\\
-1.12744728675007	-2.80133374683846\\
-1.80511967375001	1.67272233842068\\
-1.80511967375001	-1.67272233842068\\
-1.56134560027924	1.05973865794231\\
-1.56134560027924	-1.05973865794231\\
-4.38379167491111	3.32522927317787\\
-4.38379167491111	-3.32522927317787\\
-1.22056628427694	1.34980727161986\\
-1.22056628427694	-1.34980727161986\\
-1.9260932485139	7.12510717390376\\
-1.9260932485139	-7.12510717390376\\
-1.39589094257767	7.65495146004647\\
-1.39589094257767	-7.65495146004647\\
-17.3547671978457	0\\
-3.76734970494356	4.89040059804943\\
-3.76734970494356	-4.89040059804943\\
-10.5476953004074	0\\
-1.38583422503277	0.157015990065576\\
-1.38583422503277	-0.157015990065576\\
-15.2275212658558	0\\
-0.984374267309128	0.572368631753607\\
-0.984374267309128	-0.572368631753607\\
-3.04742139656695	5.87155811348103\\
-3.04742139656695	-5.87155811348103\\
-1.96922340240381	7.58155209130672\\
-1.96922340240381	-7.58155209130672\\
-3.59793946579384	5.08531509918203\\
-3.59793946579384	-5.08531509918203\\
-4.46788758332979	3.86654533897981\\
-4.46788758332979	-3.86654533897981\\
-1.16998253884831	1.63743899172241\\
-1.16998253884831	-1.63743899172241\\
-7.02731045486435	0\\
-1.56800556450299	3.10753211172284\\
-1.56800556450299	-3.10753211172284\\
-1.19803101119365	0.535435267560702\\
-1.19803101119365	-0.535435267560702\\
-2.96535787504719	4.77429263251658\\
-2.96535787504719	-4.77429263251658\\
-15.4998341555173	0\\
-0.980310201981451	2.19560266184136\\
-0.980310201981451	-2.19560266184136\\
-1.45407068490506	0.454378929649379\\
-1.45407068490506	-0.454378929649379\\
-1.41287681671359	6.92655329034191\\
-1.41287681671359	-6.92655329034191\\
-2.63113024739655	4.32471076530061\\
-2.63113024739655	-4.32471076530061\\
-1.58309943566776	7.04196643210499\\
-1.58309943566776	-7.04196643210499\\
-1.67463996035982	2.51624855309527\\
-1.67463996035982	-2.51624855309527\\
-17.4993363698528	0\\
-0.545014437504495	2.43511699404211\\
-0.545014437504495	-2.43511699404211\\
-15.3696799372623	0\\
-3.05103616028813	6.82335964639456\\
-3.05103616028813	-6.82335964639456\\
-1.69051296544432	6.53039945301422\\
-1.69051296544432	-6.53039945301422\\
-16.5555721236357	0\\
-0.701121842704054	2.1615817234628\\
-0.701121842704054	-2.1615817234628\\
-2.17110416799166	6.09419161874388\\
-2.17110416799166	-6.09419161874388\\
-1.32874220019558	1.99918653026168\\
-1.32874220019558	-1.99918653026168\\
-2.50658975276479	0\\
-0.718454365125692	2.44605706431972\\
-0.718454365125692	-2.44605706431972\\
-8.02985524318403	0\\
-1.72935120537887	7.23792837273165\\
-1.72935120537887	-7.23792837273165\\
-3.37623566132493	6.18667038367754\\
-3.37623566132493	-6.18667038367754\\
-3.08508451750611	3.18734249549589\\
-3.08508451750611	-3.18734249549589\\
-2.7106146340408	0\\
-1.4200777466833	0.879221939905587\\
-1.4200777466833	-0.879221939905587\\
-2.922017243255	5.25600274259189\\
-2.922017243255	-5.25600274259189\\
-17.0440010122473	0\\
-17.7061638595995	0\\
-19.6009606173711	0\\
-3.95364411848201	4.42050490276685\\
-3.95364411848201	-4.42050490276685\\
-2.49994463002667	5.13260765263761\\
-2.49994463002667	-5.13260765263761\\
-17.1474248891507	0\\
-11.5713873461239	0\\
-0.934714398155048	2.63778443405251\\
-0.934714398155048	-2.63778443405251\\
-3.322239975139	4.00369656555073\\
-3.322239975139	-4.00369656555073\\
-15.777863779014	0\\
-12.5790105115813	0\\
-6.74022100557756	0\\
-1.4357347429775	1.86772585814182\\
-1.4357347429775	-1.86772585814182\\
-2.2217444972231	0\\
-3.63680456056019	4.70941655694352\\
-3.63680456056019	-4.70941655694352\\
-1.30758226259031	1.89661678816541\\
-1.30758226259031	-1.89661678816541\\
-2.07014396735294	5.89141438457891\\
-2.07014396735294	-5.89141438457891\\
-1.40883336462995	3.19657615715739\\
-1.40883336462995	-3.19657615715739\\
-3.82393561528847	5.64482791690517\\
-3.82393561528847	-5.64482791690517\\
-1.02317109817149	1.26660280510395\\
-1.02317109817149	-1.26660280510395\\
-4.3404538317943	2.7101300925629\\
-4.3404538317943	-2.7101300925629\\
-19.2772446109519	0\\
-2.21953008359413	6.7031970173609\\
-2.21953008359413	-6.7031970173609\\
-12.22784447708	0\\
-0.848702816878645	2.69665467278393\\
-0.848702816878645	-2.69665467278393\\
-2.84985607279196	6.8355839212974\\
-2.84985607279196	-6.8355839212974\\
-3.49364749893812	6.42543418774077\\
-3.49364749893812	-6.42543418774077\\
-0.852640577808306	1.3703952952175\\
-0.852640577808306	-1.3703952952175\\
-9.43602699321014	0\\
-1.16579437059734	2.23994608781225\\
-1.16579437059734	-2.23994608781225\\
-11.8085372237088	0\\
-1.29779304240707	2.17013177007306\\
-1.29779304240707	-2.17013177007306\\
-0.517691386376989	7.71029091481269\\
-0.517691386376989	-7.71029091481269\\
-4.54611166929205	0\\
-2.52788803806045	2.35445480030481\\
-2.52788803806045	-2.35445480030481\\
-2.85559679466806	4.03107923049161\\
-2.85559679466806	-4.03107923049161\\
-3.62150162155398	5.72180064842132\\
-3.62150162155398	-5.72180064842132\\
-0.969532537062225	1.45525834120375\\
-0.969532537062225	-1.45525834120375\\
-4.0670425547251	5.39934935955386\\
-4.0670425547251	-5.39934935955386\\
-0.946064776127592	2.7524334624585\\
-0.946064776127592	-2.7524334624585\\
-3.85022555004615	6.05652162190617\\
-3.85022555004615	-6.05652162190617\\
-4.67997501323731	2.93474512545359\\
-4.67997501323731	-2.93474512545359\\
-1.215913905072	1.75128534453891\\
-1.215913905072	-1.75128534453891\\
-3.34871600260071	5.93781272753484\\
-3.34871600260071	-5.93781272753484\\
-14.2165391619428	0\\
-4.5371821444475	4.56064131998889\\
-4.5371821444475	-4.56064131998889\\
-1.04022935910651	1.52621495892203\\
-1.04022935910651	-1.52621495892203\\
-4.31058171854305	3.86280317899155\\
-4.31058171854305	-3.86280317899155\\
-2.90871170949233	6.53726601326394\\
-2.90871170949233	-6.53726601326394\\
-13.4956380707544	0\\
-3.4925906073171	4.82092212622079\\
-3.4925906073171	-4.82092212622079\\
-4.0543732510084	3.20426681259901\\
-4.0543732510084	-3.20426681259901\\
-0.993065050414756	2.31732237103354\\
-0.993065050414756	-2.31732237103354\\
-9.82127073995448	0\\
-3.20401193736194	3.53779894171911\\
-3.20401193736194	-3.53779894171911\\
-2.86545956680559	0\\
-2.4212484734237	6.49435542623827\\
-2.4212484734237	-6.49435542623827\\
-1.75638934521559	1.5558486477508\\
-1.75638934521559	-1.5558486477508\\
-1.88196570106362	0.40229382074354\\
-1.88196570106362	-0.40229382074354\\
-3.47514855723877	5.48140674348495\\
-3.47514855723877	-5.48140674348495\\
-1.12295644082198	1.4475810450814\\
-1.12295644082198	-1.4475810450814\\
-0.230264687800642	7.83971285855996\\
-0.230264687800642	-7.83971285855996\\
-2.66127424931481	4.99275821951171\\
-2.66127424931481	-4.99275821951171\\
-3.08562284378583	6.60626967208185\\
-3.08562284378583	-6.60626967208185\\
-1.45151555465389	2.31333842524917\\
-1.45151555465389	-2.31333842524917\\
-1.53637215353823	7.31323957553765\\
-1.53637215353823	-7.31323957553765\\
-0.960255338346826	2.0007095977241\\
-0.960255338346826	-2.0007095977241\\
-3.48078360367276	6.1453999889395\\
-3.48078360367276	-6.1453999889395\\
-1.13394228683211	0.756031809423562\\
-1.13394228683211	-0.756031809423562\\
-12.8640672420487	0\\
-2.01881218503116	7.00676741968092\\
-2.01881218503116	-7.00676741968092\\
-2.00740701954554	-6.92177191511643\\
-2.50528422299639	5.86835073473759\\
-2.50528422299639	-5.86835073473759\\
-1.2943396407556	1.05124663658313\\
-1.2943396407556	-1.05124663658313\\
-2.924099298815	6.36004884578686\\
-2.924099298815	-6.36004884578686\\
-0.786092156980182	2.52071391023298\\
-0.786092156980182	-2.52071391023298\\
-16.2908491676887	0\\
-5.10686333677907	0\\
-2.31227632661488	1.89351052299428\\
-2.31227632661488	-1.89351052299428\\
-1.78190849830797	0.221058993756976\\
-1.78190849830797	-0.221058993756976\\
-2.07123318016886	7.18054831566471\\
-2.07123318016886	-7.18054831566471\\
-0.928839602236631	7.4055674482215\\
-0.928839602236631	-7.4055674482215\\
-1.3063991821346	2.76661602631297\\
-1.3063991821346	-2.76661602631297\\
-1.39223581694594	7.4868854416641\\
-1.39223581694594	-7.4868854416641\\
-16.0436809811884	0\\
-3.48963320710723	0\\
-1.75964344067023	2.30762238369549\\
-1.75964344067023	-2.30762238369549\\
-1.55873972091813	0.587343517371159\\
-1.55873972091813	-0.587343517371159\\
-17.8217737219159	0\\
-1.13656366029517	6.95456101969224\\
-1.13656366029517	-6.95456101969224\\
-0.655889189140853	2.3627447962732\\
-0.655889189140853	-2.3627447962732\\
-9.32103606336153	0\\
-18.3168186980103	0\\
-3.89806293079061	4.71841923298175\\
-3.89806293079061	-4.71841923298175\\
-3.19589805873738	4.0975355100726\\
-3.19589805873738	-4.0975355100726\\
-1.50600831246667	1.24470513914739\\
-1.50600831246667	-1.24470513914739\\
-8.98965567617972	0\\
-1.74060865501198	2.40622629437697\\
-1.74060865501198	-2.40622629437697\\
-1.58091300028119	1.93302582672577\\
-1.58091300028119	-1.93302582672577\\
-1.60383517591211	-1.92611534883585\\
-3.36058138619459	5.13343380225682\\
-3.36058138619459	-5.13343380225682\\
-13.1388558459325	0\\
-4.23447044023668	0\\
-2.26887084186044	2.75473205097719\\
-2.26887084186044	-2.75473205097719\\
-1.12560527269544	1.22702360691072\\
-1.12560527269544	-1.22702360691072\\
-7.1795615789371	0\\
-1.56664722312523	3.35012434206581\\
-1.56664722312523	-3.35012434206581\\
-1.72911155082458	7.52700774464896\\
-1.72911155082458	-7.52700774464896\\
-0.683894438186553	7.57290677280211\\
-0.683894438186553	-7.57290677280211\\
-1.32141072616558	0.701773306805274\\
-1.32141072616558	-0.701773306805274\\
-2.2424778258195	3.51974565160495\\
-2.2424778258195	-3.51974565160495\\
-0.66321295756578	7.80729847483109\\
-0.66321295756578	-7.80729847483109\\
-3.33418275889259	5.54411797549596\\
-3.33418275889259	-5.54411797549596\\
-3.34904833787759	0\\
-0.746135221258604	2.65297989799922\\
-0.746135221258604	-2.65297989799922\\
-1.4665819731528	0.255852402256444\\
-1.4665819731528	-0.255852402256444\\
-2.12792965931796	6.80941134750453\\
-2.12792965931796	-6.80941134750453\\
-2.10830974158976	2.23953372259066\\
-2.10830974158976	-2.23953372259066\\
-1.38024412152187	2.39988744502047\\
-1.38024412152187	-2.39988744502047\\
-2.65554529930834	7.01796006412213\\
-2.65554529930834	-7.01796006412213\\
-3.16385684536722	6.04577002403122\\
-3.16385684536722	-6.04577002403122\\
-1.1739179971007	2.6323027049313\\
-1.1739179971007	-2.6323027049313\\
-4.1034354650939	4.38918227655827\\
-4.1034354650939	-4.38918227655827\\
-5.98898834605311	0\\
-11.1082163750801	0\\
-0.992823500937969	2.90375661434815\\
-0.992823500937969	-2.90375661434815\\
-3.05223194472219	6.17142470699921\\
-3.05223194472219	-6.17142470699921\\
-11.934733147268	0\\
-1.16091264455374	2.39450629623853\\
-1.16091264455374	-2.39450629623853\\
-4.90297350351858	4.42577897038196\\
-4.90297350351858	-4.42577897038196\\
-13.9160269461549	0\\
-12.3682762659116	0\\
-2.15112096324739	1.89501494584084\\
-2.15112096324739	-1.89501494584084\\
-3.5932944997444	2.77262552731945\\
-3.5932944997444	-2.77262552731945\\
-3.221865687247	6.38230633221203\\
-3.221865687247	-6.38230633221203\\
-6.10156814110182	0\\
-11.706245139848	0\\
-2.15335771397941	6.61585756824454\\
-2.15335771397941	-6.61585756824454\\
-14.5010795664705	0\\
-4.92177157629653	2.7205701352939\\
-4.92177157629653	-2.7205701352939\\
-3.19939984066245	5.18193629464413\\
-3.19939984066245	-5.18193629464413\\
-1.36734344883661	1.2179065603214\\
-1.36734344883661	-1.2179065603214\\
-2.19054843890823	3.75529709287356\\
-2.19054843890823	-3.75529709287356\\
-9.18673996456649	0\\
-1.51208251636389	6.78276919131517\\
-1.51208251636389	-6.78276919131517\\
-2.70493704456346	3.62597366425452\\
-2.70493704456346	-3.62597366425452\\
-4.34999340814453	0\\
-4.31335692110543	1.7471170815598\\
-4.31335692110543	-1.7471170815598\\
-1.52027995176682	1.42734480666926\\
-1.52027995176682	-1.42734480666926\\
-3.41274088617702	5.04651499954213\\
-3.41274088617702	-5.04651499954213\\
-9.59215202006314	0\\
-1.36435817173687	2.86419523363662\\
-1.36435817173687	-2.86419523363662\\
-1.16601203513773	2.91929871443984\\
-1.16601203513773	-2.91929871443984\\
-3.18973402611002	2.89893767833544\\
-3.18973402611002	-2.89893767833544\\
-4.19663931215609	5.60267823432181\\
-4.19663931215609	-5.60267823432181\\
-1.26110988716808	2.53814503094117\\
-1.26110988716808	-2.53814503094117\\
-2.4030800330904	4.30879343779097\\
-2.4030800330904	-4.30879343779097\\
-0.583459996005484	2.59942211301865\\
-0.583459996005484	-2.59942211301865\\
-10.0597998008413	0\\
-6.27477493453462	0\\
-2.44730409301109	2.55986929430124\\
-2.44730409301109	-2.55986929430124\\
-1.63312713452477	7.6270160291485\\
-1.63312713452477	-7.6270160291485\\
-1.41772470728138	2.73069820196757\\
-1.41772470728138	-2.73069820196757\\
-1.50986264705801	7.11180698384286\\
-1.50986264705801	-7.11180698384286\\
-2.51125164867909	6.83578518996741\\
-2.51125164867909	-6.83578518996741\\
-3.49188155678757	3.91274828236325\\
-3.49188155678757	-3.91274828236325\\
-2.16091323418845	6.43674740341743\\
-2.16091323418845	-6.43674740341743\\
-1.58515180237855	2.82098906989832\\
-1.58515180237855	-2.82098906989832\\
-1.89066904588546	2.17248840135239\\
-1.89066904588546	-2.17248840135239\\
-5.72288768720629	2.19042776816184\\
-5.72288768720629	-2.19042776816184\\
-3.53353301269201	3.72650659637724\\
-3.53353301269201	-3.72650659637724\\
-10.1924127278494	0\\
-1.13933681650612	1.84442860053888\\
-1.13933681650612	-1.84442860053888\\
-3.70622743229413	5.65319138386414\\
-3.70622743229413	-5.65319138386414\\
-2.4897407274265	6.27611603296637\\
-2.4897407274265	-6.27611603296637\\
-2.33989698327475	0\\
-1.99214366136418	2.9289272376398\\
-1.99214366136418	-2.9289272376398\\
-16.7699607655967	0\\
-2.03053897998148	6.39671572761822\\
-2.03053897998148	-6.39671572761822\\
-1.61671334429836	0.437344961774474\\
-1.61671334429836	-0.437344961774474\\
-15.8969991245301	0\\
-2.88313375982939	6.2155824254946\\
-2.88313375982939	-6.2155824254946\\
-1.24226526320359	7.57674226504418\\
-1.24226526320359	-7.57674226504418\\
-15.674212123696	0\\
-2.30534179720229	5.27812521698826\\
-2.30534179720229	-5.27812521698826\\
-3.99725203669461	4.7413832971942\\
-3.99725203669461	-4.7413832971942\\
-1.96836812060778	2.57657291415884\\
-1.96836812060778	-2.57657291415884\\
-1.50513656809544	0.147120790339176\\
-1.50513656809544	-0.147120790339176\\
-2.69117539230873	5.58639486931792\\
-2.69117539230873	-5.58639486931792\\
-2.22048209645374	7.02334817960164\\
-2.22048209645374	-7.02334817960164\\
-5.2133326319736	3.84053604595616\\
-5.2133326319736	-3.84053604595616\\
-4.16890742195979	3.32508442203832\\
-4.16890742195979	-3.32508442203832\\
-1.52572725754007	2.67287677494012\\
-1.52572725754007	-2.67287677494012\\
-3.80607754293939	2.69215806332346\\
-3.80607754293939	-2.69215806332346\\
-2.85670710776702	5.92545005871049\\
-2.85670710776702	-5.92545005871049\\
-2.566638603458	5.41828739380895\\
-2.566638603458	-5.41828739380895\\
-2.42209626137906	6.93287490937958\\
-2.42209626137906	-6.93287490937958\\
-3.82377468062143	5.8883724232527\\
-3.82377468062143	-5.8883724232527\\
-2.86677122732411	4.37272164690174\\
-2.86677122732411	-4.37272164690174\\
-4.92540914881233	3.77830186759406\\
-4.92540914881233	-3.77830186759406\\
-5.26185507351072	0\\
-0.346564798192464	7.73535969527392\\
-0.346564798192464	-7.73535969527392\\
-2.22585597947478	5.98188339346498\\
-2.22585597947478	-5.98188339346498\\
-20.1115422866157	0\\
-3.55194837935421	5.99942703780039\\
-3.55194837935421	-5.99942703780039\\
-8.4318179807757	0\\
-1.70030930822055	7.05662641698497\\
-1.70030930822055	-7.05662641698497\\
-2.39147592367169	6.36317033970749\\
-2.39147592367169	-6.36317033970749\\
-1.62026896487118	6.70505818860968\\
-1.62026896487118	-6.70505818860968\\
-18.4404240284353	0\\
-4.03350951392634	4.30592355917113\\
-4.03350951392634	-4.30592355917113\\
-2.60988896028963	6.50045202268576\\
-2.60988896028963	-6.50045202268576\\
-1.4870560035151	0.974644254080176\\
-1.4870560035151	-0.974644254080176\\
-1.76103050357265	6.88893905390265\\
-1.76103050357265	-6.88893905390265\\
-2.46995205752702	4.94808383491872\\
-2.46995205752702	-4.94808383491872\\
-2.61579864203397	5.92190358943996\\
-2.61579864203397	-5.92190358943996\\
-3.62298285301192	5.48133058270964\\
-3.62298285301192	-5.48133058270964\\
-3.58325363832698	3.45464075980564\\
-3.58325363832698	-3.45464075980564\\
-1.30173681464632	0.545907412758043\\
-1.30173681464632	-0.545907412758043\\
-1.44434927954619	0.569324584646026\\
-1.44434927954619	-0.569324584646026\\
-3.78250623670198	5.24291623500101\\
-3.78250623670198	-5.24291623500101\\
-3.13589254101671	4.52817672777164\\
-3.13589254101671	-4.52817672777164\\
-1.19730649555723	7.44889440887608\\
-1.19730649555723	-7.44889440887608\\
-13.0283246952522	0\\
-1.46478779827505	7.74577308000624\\
-1.46478779827505	-7.74577308000624\\
-1.42654923200964	1.4651623349982\\
-1.42654923200964	-1.4651623349982\\
-2.08182891821976	0.618675181988124\\
-2.08182891821976	-0.618675181988124\\
-3.82764074188227	4.22726176855642\\
-3.82764074188227	-4.22726176855642\\
-2.1364539871479	3.09505446411267\\
-2.1364539871479	-3.09505446411267\\
-4.0430029748493	4.18830039047798\\
-4.0430029748493	-4.18830039047798\\
-14.8436243954711	0\\
-4.32708348942723	4.08129330948747\\
-4.32708348942723	-4.08129330948747\\
-1.63940582196349	6.42236490232461\\
-1.63940582196349	-6.42236490232461\\
-5.6377052263574	2.59764171206408\\
-5.6377052263574	-2.59764171206408\\
-0.955310339808131	1.62134451931333\\
-0.955310339808131	-1.62134451931333\\
-1.29680110030804	1.4285296487225\\
-1.29680110030804	-1.4285296487225\\
-2.5316291073637	5.60281338750068\\
-2.5316291073637	-5.60281338750068\\
-3.2501016899595	4.58508979343383\\
-3.2501016899595	-4.58508979343383\\
-1.59412271420687	7.47426738340703\\
-1.59412271420687	-7.47426738340703\\
-11.4340482636544	0\\
-3.76110637007313	4.5234218750318\\
-3.76110637007313	-4.5234218750318\\
-2.84180790112169	3.87915930186129\\
-2.84180790112169	-3.87915930186129\\
-0.543714808884951	7.9561960799823\\
-0.543714808884951	-7.9561960799823\\
-1.20400453151968	7.78772345069247\\
-1.20400453151968	-7.78772345069247\\
-2.35881092372336	7.10599638387977\\
-2.35881092372336	-7.10599638387977\\
-7.30198859694916	0\\
-1.05533868936475	2.60565565110986\\
-1.05533868936475	-2.60565565110986\\
-3.22069565866331	3.07549151207802\\
-3.22069565866331	-3.07549151207802\\
-2.79152658931009	5.79719950744454\\
-2.79152658931009	-5.79719950744454\\
-3.1782957586248	6.23933856398322\\
-3.1782957586248	-6.23933856398322\\
-1.25862417269009	3.10910196527065\\
-1.25862417269009	-3.10910196527065\\
-3.48195115308004	5.86301225760402\\
-3.48195115308004	-5.86301225760402\\
-2.4186697058257	5.99052494812351\\
-2.4186697058257	-5.99052494812351\\
-1.4919868749294	2.12038600871052\\
-1.4919868749294	-2.12038600871052\\
-2.05832426747521	6.25884905227698\\
-2.05832426747521	-6.25884905227698\\
-1.30467447914354	7.07498698245588\\
-1.30467447914354	-7.07498698245588\\
-3.70937692259099	3.56142312328076\\
-3.70937692259099	-3.56142312328076\\
-3.00119339909171	5.99991161231629\\
-3.00119339909171	-5.99991161231629\\
-1.9656340268729	6.49565478967935\\
-1.9656340268729	-6.49565478967935\\
-4.7895482545752	0\\
-2.52582248426542	3.45874305705317\\
-2.52582248426542	-3.45874305705317\\
-1.70261396488266	3.06129978667164\\
-1.70261396488266	-3.06129978667164\\
-2.80319529911323	4.49661737884476\\
-2.80319529911323	-4.49661737884476\\
-2.95424665975041	4.50445525091167\\
-2.95424665975041	-4.50445525091167\\
-3.28511830665721	4.6886291606011\\
-3.28511830665721	-4.6886291606011\\
-1.84637221008429	1.85817996429664\\
-1.84637221008429	-1.85817996429664\\
-1.27805265154459	0.947080192830607\\
-1.27805265154459	-0.947080192830607\\
-2.7220259400032	6.54019492947126\\
-2.7220259400032	-6.54019492947126\\
-2.37846708737009	1.78575621057095\\
-2.37846708737009	-1.78575621057095\\
-1.75820828697298	0.344185104903452\\
-1.75820828697298	-0.344185104903452\\
-3.48503849308132	5.14707829618542\\
-3.48503849308132	-5.14707829618542\\
-1.13973086346309	7.18519709386116\\
-1.13973086346309	-7.18519709386116\\
-1.8658085753886	6.77901018494596\\
-1.8658085753886	-6.77901018494596\\
-4.34337585116946	4.21585955769135\\
-4.34337585116946	-4.21585955769135\\
-1.95653983771461	6.87832674069805\\
-2.44790191087587	5.49916506683626\\
-2.44790191087587	-5.49916506683626\\
-7.73339530882944	0\\
-1.49674672723395	2.49787853247021\\
-1.49674672723395	-2.49787853247021\\
-2.41327358392542	6.68999231721631\\
-2.41327358392542	-6.68999231721631\\
-2.00326477595214	6.07314439465741\\
-2.00326477595214	-6.07314439465741\\
-2.10734410914366	0\\
-5.5399525338885	0\\
-2.06665534955057	2.84110831670154\\
-2.06665534955057	-2.84110831670154\\
-0.834644305768128	2.80046392577831\\
-0.834644305768128	-2.80046392577831\\
-1.39233403680292	7.14833627523901\\
-1.39233403680292	-7.14833627523901\\
-0.808248179792991	7.76236035399789\\
-0.808248179792991	-7.76236035399789\\
-4.15716805973583	4.06360699953646\\
-4.15716805973583	-4.06360699953646\\
-1.53559746176466	0.879814453230458\\
-1.53559746176466	-0.879814453230458\\
-3.35375747036684	4.49707350217951\\
-3.35375747036684	-4.49707350217951\\
-4.82571001532852	1.37931494921842\\
-4.82571001532852	-1.37931494921842\\
-2.57307815954534	2.51303666103618\\
-2.57307815954534	-2.51303666103618\\
-3.81019761042796	0\\
-4.20898577960589	4.64819769121973\\
-4.20898577960589	-4.64819769121973\\
-2.07216553882998	2.62212746385415\\
-2.07216553882998	-2.62212746385415\\
-0.363485881258263	8.05328668607656\\
-0.363485881258263	-8.05328668607656\\
-2.70039290466777	5.23962009352291\\
-2.70039290466777	-5.23962009352291\\
-0.898169937778211	8.03789032049303\\
-0.898169937778211	-8.03789032049303\\
-3.41789615218398	3.28126247276506\\
-3.41789615218398	-3.28126247276506\\
-3.09713887124733	4.20680340539572\\
-3.09713887124733	-4.20680340539572\\
-0.871499370364119	7.87165008661172\\
-0.871499370364119	-7.87165008661172\\
-3.97720137922203	5.70072022361499\\
-3.97720137922203	-5.70072022361499\\
-4.4911622841431	4.85792506687056\\
-4.4911622841431	-4.85792506687056\\
-2.6532681616051	4.50432134090596\\
-2.6532681616051	-4.50432134090596\\
-5.20369069927607	3.26813830354373\\
-5.20369069927607	-3.26813830354373\\
-10.3418954536184	0\\
-5.3079732720395	3.08929531723132\\
-5.3079732720395	-3.08929531723132\\
-4.14390766542658	3.8152151825738\\
-4.14390766542658	-3.8152151825738\\
-2.60882621694286	0\\
-4.86614649593724	1.16129706596603\\
-4.86614649593724	-1.16129706596603\\
-1.46168092020423	1.63275725813959\\
-1.46168092020423	-1.63275725813959\\
-4.58280852567909	4.06030265359587\\
-4.58280852567909	-4.06030265359587\\
-2.36319849935846	4.72084712787449\\
-2.36319849935846	-4.72084712787449\\
-2.74998813563415	5.0787442980018\\
-2.74998813563415	-5.0787442980018\\
-4.06127044063847	5.54484949056005\\
-4.06127044063847	-5.54484949056005\\
-3.18814839688507	5.75830542944041\\
-3.18814839688507	-5.75830542944041\\
-0.909065417742446	7.29551410428196\\
-0.909065417742446	-7.29551410428196\\
-3.06303497958624	4.93892684846979\\
-3.06303497958624	-4.93892684846979\\
-5.12245473920457	4.05181387154469\\
-5.12245473920457	-4.05181387154469\\
-5.32735466288575	2.66769958851957\\
-5.32735466288575	-2.66769958851957\\
-4.90179565869857	3.38783831253958\\
-4.90179565869857	-3.38783831253958\\
-3.86083331955855	4.598054041159\\
-3.86083331955855	-4.598054041159\\
-3.03654886811822	2.19605086855538\\
-3.03654886811822	-2.19605086855538\\
-4.99901416712282	3.55357215042901\\
-4.99901416712282	-3.55357215042901\\
-16.6593583430593	0\\
-1.06994187615215	1.97694027979675\\
-1.06994187615215	-1.97694027979675\\
-4.0726423847728	0\\
-0.875836046381155	2.17481237675909\\
-0.875836046381155	-2.17481237675909\\
-3.03927578744813	5.29905290342937\\
-3.03927578744813	-5.29905290342937\\
-3.73811086521201	5.37613957122683\\
-3.73811086521201	-5.37613957122683\\
-3.49223590952316	4.97178008431856\\
-3.49223590952316	-4.97178008431856\\
-1.43317597275864	2.01722118195782\\
-1.43317597275864	-2.01722118195782\\
-3.20036218088768	5.60552824413285\\
-3.20036218088768	-5.60552824413285\\
-4.97232424545215	2.18119759247374\\
-4.97232424545215	-2.18119759247374\\
-3.66091639391615	3.66181549952088\\
-3.66091639391615	-3.66181549952088\\
-2.88767901240456	3.29769533077025\\
-2.88767901240456	-3.29769533077025\\
-1.87519028159079	0.910142448413985\\
-1.87519028159079	-0.910142448413985\\
-5.17212716696269	1.28763934983336\\
-5.17212716696269	-1.28763934983336\\
-0.46084820511261	7.54063870129903\\
-0.46084820511261	-7.54063870129903\\
-3.74679701431898	1.87802729228933\\
-3.74679701431898	-1.87802729228933\\
-3.58917232322125	6.19605344037253\\
-3.58917232322125	-6.19605344037253\\
-1.05316967636199	7.45163910571566\\
-1.05316967636199	-7.45163910571566\\
-2.80570926492541	3.73007778314734\\
-2.80570926492541	-3.73007778314734\\
-1.83901321999486	0.659830591896091\\
-1.83901321999486	-0.659830591896091\\
-2.67676479425456	4.08115429555432\\
-2.67676479425456	-4.08115429555432\\
-1.92085443699587	2.8495606460431\\
-1.92085443699587	-2.8495606460431\\
-4.95188177480269	4.20538001056497\\
-4.95188177480269	-4.20538001056497\\
-1.8638893949941	-6.91104292552919\\
-3.39638172344273	3.86773788663443\\
-3.39638172344273	-3.86773788663443\\
-1.6313808335657	1.49695392405932\\
-1.6313808335657	-1.49695392405932\\
-2.52065112436285	3.32704082871426\\
-2.52065112436285	-3.32704082871426\\
-18.8190779377021	0\\
-0.998609075472461	8.00713832617337\\
-0.998609075472461	-8.00713832617337\\
-2.25830670938022	2.1677709061669\\
-2.25830670938022	-2.1677709061669\\
-16.435924049142	0\\
-1.98545696470961	7.35029917174141\\
-1.98545696470961	-7.35029917174141\\
-1.61567028524008	7.20102768650359\\
-1.61567028524008	-7.20102768650359\\
-2.82103544475297	6.05576973212462\\
-2.82103544475297	-6.05576973212462\\
-4.22090155345399	2.8289860162862\\
-4.22090155345399	-2.8289860162862\\
-1.12239671487658	7.56173043453729\\
-1.12239671487658	-7.56173043453729\\
-1.59163460086725	0.213081321283107\\
-1.59163460086725	-0.213081321283107\\
-4.10755487329604	4.54511027608707\\
-4.10755487329604	-4.54511027608707\\
-18.2163292063297	0\\
-12.7561690296951	0\\
-12.7604600692279	0\\
-3.36931485736673	4.90148542678225\\
-3.36931485736673	-4.90148542678225\\
-4.45430032276961	2.15957456507423\\
-4.45430032276961	-2.15957456507423\\
-1.54985510940993	1.74842173474731\\
-1.54985510940993	-1.74842173474731\\
-1.36060901108444	2.98128813253668\\
-1.36060901108444	-2.98128813253668\\
-2.78538268469363	5.42763829260848\\
-2.78538268469363	-5.42763829260848\\
-3.08010077026928	3.56110242979205\\
-3.08010077026928	-3.56110242979205\\
-13.7950984828659	0\\
-1.74924503620683	7.64763547725951\\
-1.74924503620683	-7.64763547725951\\
-0.98093234295393	1.90249413915176\\
-0.98093234295393	-1.90249413915176\\
-0.79734484689995	7.52854007900877\\
-0.79734484689995	-7.52854007900877\\
-1.05745643347791	7.34449170229814\\
-1.05745643347791	-7.34449170229814\\
-0.627205700887933	8.04073660594218\\
-0.627205700887933	-8.04073660594218\\
-2.6892723958399	2.29018010522903\\
-2.6892723958399	-2.29018010522903\\
-2.32637095767432	5.74307100755484\\
-2.32637095767432	-5.74307100755484\\
-1.31194360236222	6.90289879068204\\
-1.31194360236222	-6.90289879068204\\
-2.35432140000897	6.82878100010598\\
-2.35432140000897	-6.82878100010598\\
-4.73428291759109	3.06417325156054\\
-4.73428291759109	-3.06417325156054\\
-2.87077084848825	4.16276953292637\\
-2.87077084848825	-4.16276953292637\\
-3.19827275907732	4.77064981478217\\
-3.19827275907732	-4.77064981478217\\
-3.11614007390854	4.65746426000547\\
-3.11614007390854	-4.65746426000547\\
-2.72832126030846	6.83901147852633\\
-2.72832126030846	-6.83901147852633\\
-1.19951045938147	1.52841389890785\\
-1.19951045938147	-1.52841389890785\\
-2.42666852190618	4.45055074313995\\
-2.42666852190618	-4.45055074313995\\
-1.82125834561043	6.55937858251382\\
-1.82125834561043	-6.55937858251382\\
-2.81775711644772	3.00353809865714\\
-2.81775711644772	-3.00353809865714\\
-1.44993349474226	1.08999557178484\\
-1.44993349474226	-1.08999557178484\\
-4.63646260192746	3.6725527166922\\
-4.63646260192746	-3.6725527166922\\
-3.28918107356849	6.49737152919579\\
-3.28918107356849	-6.49737152919579\\
-1.83439825780709	7.53956304692391\\
-1.83439825780709	-7.53956304692391\\
-2.37272903720739	2.03782894788834\\
-2.37272903720739	-2.03782894788834\\
-2.22767520403002	2.60928717155415\\
-2.22767520403002	-2.60928717155415\\
-6.63591202220335	0\\
-3.33620429490666	4.28774784402431\\
-3.33620429490666	-4.28774784402431\\
-4.81986032153749	4.25829461750263\\
-4.81986032153749	-4.25829461750263\\
-2.30406091418255	6.51891309594819\\
-2.30406091418255	-6.51891309594819\\
-0.980323494395534	7.14938168191466\\
-0.980323494395534	-7.14938168191466\\
-4.87894975227663	3.2052288121905\\
-4.87894975227663	-3.2052288121905\\
-4.16566784094938	1.41246137658932\\
-4.16566784094938	-1.41246137658932\\
-1.64470408110026	1.85003342474429\\
-3.53561135415624	5.2565489416466\\
-3.53561135415624	-5.2565489416466\\
-2.24182508799443	2.28991436365949\\
-2.24182508799443	-2.28991436365949\\
-3.23069324491899	3.87232935275771\\
-3.23069324491899	-3.87232935275771\\
-1.7532517957018	1.33019281292639\\
-1.7532517957018	-1.33019281292639\\
-3.99042359214584	4.91847441873918\\
-3.99042359214584	-4.91847441873918\\
-4.81558188863284	4.07799563210628\\
-4.81558188863284	-4.07799563210628\\
-12.102491890285	0\\
-3.72727322559479	5.94510688052186\\
-3.72727322559479	-5.94510688052186\\
-2.03374385314547	5.68558835359936\\
-2.03374385314547	-5.68558835359936\\
-1.82709918906149	6.67071949191568\\
-1.82709918906149	-6.67071949191568\\
-1.9878167093798	3.25670442014415\\
-1.9878167093798	-3.25670442014415\\
-3.27388241617003	3.71867521880329\\
-3.27388241617003	-3.71867521880329\\
-1.1747918299113	1.94681868306379\\
-1.1747918299113	-1.94681868306379\\
-0.963489503052607	2.53836687537532\\
-0.963489503052607	-2.53836687537532\\
-3.73808278553969	6.23957410604304\\
-3.73808278553969	-6.23957410604304\\
-2.94619232032535	5.10247009248661\\
-2.94619232032535	-5.10247009248661\\
-5.78698622975131	1.37418593897691\\
-5.78698622975131	-1.37418593897691\\
-2.32922409087804	3.13213832128986\\
-2.32922409087804	-3.13213832128986\\
-4.4945023032726	3.96581972698231\\
-4.4945023032726	-3.96581972698231\\
-3.22327123845955	3.35578815413425\\
-3.22327123845955	-3.35578815413425\\
-9.71060331237381	0\\
-2.52830340779126	3.16550330036615\\
-2.52830340779126	-3.16550330036615\\
-4.42952395693517	2.85627800984285\\
-4.42952395693517	-2.85627800984285\\
-2.09687771042064	7.472870397043\\
-2.09687771042064	-7.472870397043\\
-5.53355650671628	1.38578074814285\\
-5.53355650671628	-1.38578074814285\\
-1.29874768839922	6.73067104990601\\
-1.29874768839922	-6.73067104990601\\
-1.19123748891667	7.67990894580595\\
-1.19123748891667	-7.67990894580595\\
-1.78780092839445	2.498470322617\\
-1.78780092839445	-2.498470322617\\
-0.408459669708435	7.64797819157897\\
};

\end{axis}

\end{tikzpicture}%

%% file: figure/CCF_csr2.tex
%
%
\definecolor{mycolor1}{rgb}{0.00000,0.44700,0.74100}%
\definecolor{mycolor2}{rgb}{0.85000,0.32500,0.09800}%
\definecolor{mycolor3}{rgb}{0.92900,0.69400,0.12500}%
\begin{tikzpicture}

\begin{axis}[%
width=4.521in,
height=3.566in,
at={(0.758in,0.481in)},
scale only axis,
xmin=0,
xmax=40,
ymin=-2,
ymax=0.5,
xtick={0,10,20,30,40},
ytick={-2,-1,5,-1,-0,5,0,0.5},
axis x line*=bottom,
axis y line*=left,
axis background/.style={fill=white},
legend style={legend cell align=left, align=left, draw=white!15!black}
]
\addplot [color=mycolor1, line width = 2pt]
  table[row sep=crcr]{%
0.4	-1.78757394251049\\
0.8	-1.66117300147183\\
1.2	-1.57194911968303\\
1.6	-1.50191255730589\\
2	-1.44381362719659\\
2.4	-1.39394856394339\\
2.8	-1.35014589821526\\
3.2	-1.31098433409552\\
3.6	-1.27557917725932\\
4	-1.24324391196309\\
4.4	-1.21342267582135\\
4.8	-1.18573090097723\\
5.2	-1.15994981887984\\
5.6	-1.1357322122914\\
6	-1.11296009884795\\
6.4	-1.09144671850487\\
6.8	-1.07106200507452\\
7.2	-1.0516837566781\\
7.6	-1.03321982150599\\
8	-1.01558710079657\\
8.4	-0.99871359267527\\
8.8	-0.982536473655808\\
9.2	-0.966997184589799\\
9.6	-0.952053896608568\\
10	-0.937659841067547\\
10.4	-0.923776390833147\\
10.8	-0.910383369448115\\
11.2	-0.897406572941307\\
11.6	-0.884860959873541\\
12	-0.872709373572809\\
12.4	-0.860923657627748\\
12.8	-0.849484362435623\\
13.2	-0.838369689146265\\
13.6	-0.827570502519022\\
14	-0.817062323112067\\
14.4	-0.806832644040838\\
14.8	-0.796865223852943\\
15.2	-0.787151883002913\\
15.6	-0.777688950497893\\
16	-0.768435507623964\\
16.4	-0.7594169334127\\
16.8	-0.750598650947147\\
17.2	-0.74197998509483\\
17.6	-0.733544709505934\\
18	-0.725308108226812\\
18.4	-0.71723958876042\\
18.8	-0.70934229505464\\
19.2	-0.701602369417868\\
19.6	-0.694013823684903\\
20	-0.686582609912006\\
20.4	-0.679296047335109\\
20.8	-0.672148896264663\\
21.2	-0.66514286879377\\
21.6	-0.658253297110248\\
22	-0.651502226152511\\
22.4	-0.644865611338309\\
22.8	-0.638346137222306\\
23.2	-0.631942000886932\\
23.6	-0.625645420543367\\
24	-0.619455022407916\\
24.4	-0.613367504935724\\
24.8	-0.607377774431948\\
25.2	-0.601488653574404\\
25.6	-0.595691436117727\\
26	-0.589977743047229\\
26.4	-0.584365831034272\\
26.8	-0.578834162914938\\
27.2	-0.573385991272437\\
27.6	-0.56801360949209\\
28	-0.562725741755443\\
28.4	-0.557514741671521\\
28.8	-0.55238384574109\\
29.2	-0.547320474695793\\
29.6	-0.5423280738194\\
30	-0.537399669076519\\
30.4	-0.532543853548626\\
30.8	-0.527753760639812\\
31.2	-0.523031785966799\\
31.6	0.110366623921405\\
32	0.115166888347182\\
32.4	0.119880304053584\\
32.8	0.124499715695013\\
33.2	0.129031544134861\\
33.6	0.13347757820651\\
34	0.137841057906153\\
34.4	0.142123516318978\\
34.8	0.146327503795068\\
35.2	0.15045505893419\\
35.6	0.15450857328122\\
36	0.158489829153597\\
36.4	0.162398035334728\\
36.8	0.166240649190103\\
37.2	0.170019535546912\\
37.6	0.173730739982241\\
38	0.177378775788865\\
38.4	0.180965275796061\\
38.8	0.184491820192744\\
39.2	0.187959930061112\\
39.6	0.191367469876928\\
40	0.194723105485505\\
};

\addplot[only marks, mark=*, mark options={fill=red}, mark size=3.000pt, draw=mycolor2] table[row sep=crcr]{%
x	y\\
1	-1.61345427760348\\
8.27388563512289	-1.00397428002172\\
32.1188188994419	0.116578774171014\\
22.2574002228109	-0.0370718910606369\\
24.0621301353908	-0.00203887009047481\\
24.1733033615672	-1.19930515497108e-05\\
};

\addplot [color=mycolor3, line width = 2pt]
  table[row sep=crcr]{%
0.4	-1.63206211014044\\
0.8	-1.50271617431412\\
1.2	-1.42588780041284\\
1.6	-1.37107123504165\\
2	-1.32845255937657\\
2.4	-1.29040767579865\\
2.8	-1.20557553385461\\
3.2	-1.13019500515516\\
3.6	-1.06243632317354\\
4	-1.00096941388357\\
4.4	-0.944779479564179\\
4.8	-0.893134214929781\\
5.2	-0.845377819464133\\
5.6	-0.801029487435658\\
6	-0.759684267241802\\
6.4	-0.721005902502867\\
6.8	-0.684708645010838\\
7.2	-0.650554857813494\\
7.6	-0.618337632919275\\
8	-0.587878925727987\\
8.4	-0.55902401494395\\
8.8	-0.531637776130619\\
9.2	-0.505599747094973\\
9.6	-0.480804149588649\\
10	-0.457158307807254\\
10.4	-0.434576544145677\\
10.8	-0.412984579540817\\
11.2	-0.392313383106873\\
11.6	-0.372502108429931\\
12	-0.353494854118443\\
12.4	-0.335240851403388\\
12.8	-0.317693769252428\\
13.2	-0.300811203012581\\
13.6	-0.284554475369121\\
14	-0.268887749920825\\
14.4	-0.253778067639555\\
14.8	-0.239195273735656\\
15.2	-0.225111097853397\\
15.6	-0.211499598360878\\
16	-0.198336549571954\\
16.4	-0.185599470498211\\
16.8	-0.173267584048802\\
17.2	-0.161321075129125\\
17.6	-0.149737108885916\\
18	-0.13851293276811\\
18.4	-0.12761819523043\\
18.8	-0.117036838088099\\
19.2	-0.106772301025527\\
19.6	-0.0967938372833723\\
20	-0.0870962386626565\\
20.4	-0.0776648443068286\\
20.8	-0.0684900574592475\\
21.2	-0.0595614178385994\\
21.6	-0.0508690393763068\\
22	-0.0424035699087809\\
22.4	-0.0341561622012239\\
22.8	-0.026118410656077\\
23.2	-0.0182823662988269\\
23.6	-0.0106404968300376\\
24	-0.00318557677156411\\
24.4	0.004089177273841\\
24.8	0.0111902404980072\\
25.2	0.0181238337057639\\
25.6	0.0248970378311135\\
26	0.0315129920345852\\
26.4	0.0379782964972853\\
26.8	0.044301594539942\\
27.2	0.0504759938167553\\
27.6	0.0565190957079323\\
28	0.0624306437077397\\
28.4	0.0682149264074721\\
28.8	0.0738759913089304\\
29.2	0.0794222648613054\\
29.6	0.084848443266686\\
30	0.0901626169786607\\
30.4	0.095368277139027\\
30.8	0.100468801903669\\
31.2	0.105467336771159\\
31.6	0.110366821286474\\
32	0.115170228582928\\
32.4	0.119880304602222\\
32.8	0.124499943804553\\
33.2	0.129031635322523\\
33.6	0.13347787930893\\
34	0.137840752967766\\
34.4	0.142123242567616\\
34.8	0.146327270081592\\
35.2	0.150454999053217\\
35.6	0.154508512026083\\
36	0.1584898179115\\
36.4	0.162400853852016\\
36.8	0.166243486014856\\
37.2	0.170019532873116\\
37.6	0.17373073148302\\
38	0.177378767787943\\
38.4	0.180965268274006\\
38.8	0.184491812064927\\
39.2	0.187956286927431\\
39.6	0.191367463013626\\
40	0.194723101305383\\
};

\node[right, align=left]
at (axis cs:1,-1.613) (l1) {Iter. 1};
\node[right, align=left]
at (axis cs:8.274,-1.004) (l2) {Iter. 2};
\node[right, align=left]
at (axis cs:32.119,0.117) (l3) {Iter. 3};
\node[right, align=left]
at (axis cs:12.257,0) (l4) {Iter. 4};
\node[right, align=left]
at (axis cs:20.062,0.2) (l5) {Iter. 5};
\node[right, align=left]
at (axis cs:25,0.3) (l6) {Iter. 6};

 \node[]
at (axis cs:22.257,-0.037) (4) {};
\node[]
at (axis cs:24.062,-0.002) (5) {};
\node[]
at (axis cs:24.173,0.0) (6) {};

\draw [->] (l4) -- (4);
\draw [->] (l5) -- (5);
\draw [->] (l6) -- (6);

\end{axis}

\end{tikzpicture}%

%% file: figure/CCF_np_cpa.tex
%
%
\begin{tikzpicture}

\begin{axis}[%
width=4.521in,
height=3.566in,
at={(0.758in,0.481in)},
scale only axis,
xmin=-40,
xmax=0,
ymin=-6,
ymax=6,
xtick={-40,-30,-20,-10,0},
axis background/.style={fill=white},
legend style={legend cell align=left, align=left, draw=white!15!black}
]
\addplot [color=black, draw=none, mark=asterisk, mark options={solid, black}]
  table[row sep=crcr]{%
-31.7669039983759	0\\
-0.285377708643937	1.78174227582248\\
-0.285377708643937	-1.78174227582248\\
-1.10690563184051	0.474292147712071\\
-1.10690563184051	-0.474292147712071\\
};

\addplot+[color=blue, draw=none, mark=*, mark size=1pt,mark options={solid, blue}]
  table[row sep=crcr]{%
-3.00000000000002	0.999999999999957\\
-3.00000000000002	-0.999999999999957\\
-2.9999999999999	0\\
-2.00000000000002	1.00000000000002\\
-2.00000000000002	-1.00000000000002\\
0	0\\
-23.9559019385747	0\\
-0.313006794163903	2.21331851909705\\
-0.313006794163903	-2.21331851909705\\
-0.968069990209975	0.689448179884485\\
-0.968069990209975	-0.689448179884485\\
-35.3660015596173	0\\
-0.00361106897834228	2.00330329737755\\
-0.00361106897834228	-2.00330329737755\\
-0.900038151213045	0.496763851944215\\
-0.900038151213045	-0.496763851944215\\
-5.38435857308885	5.99420518258453\\
-5.38435857308885	-5.99420518258453\\
-0.651544648276718	1.43690731295013\\
-0.651544648276718	-1.43690731295013\\
-0.928193557268862	0\\
-11.6768319945707	0\\
-1.19489858802836	2.70187230110884\\
-1.19489858802836	-2.70187230110884\\
-1.0981670295541	0.515612003177866\\
-1.0981670295541	-0.515612003177866\\
-5.09084282655966	4.84841160704807\\
-5.09084282655966	-4.84841160704807\\
-0.942455112784733	1.53730033792785\\
-0.942455112784733	-1.53730033792785\\
-26.4157801395334	0\\
-0.161953902421561	2.16698268924069\\
-0.161953902421561	-2.16698268924069\\
-4.50062045096305	1.74356300558913\\
-4.50062045096305	-1.74356300558913\\
-1.52796110949911	2.12007271846561\\
-1.52796110949911	-2.12007271846561\\
-5.34764238948238	5.74276802141172\\
-5.34764238948238	-5.74276802141172\\
-25.7035366913256	0\\
-0.314314230360635	2.10821413773146\\
-0.314314230360635	-2.10821413773146\\
-5.22600365345769	5.0684323878365\\
-5.22600365345769	-5.0684323878365\\
-0.836802725707364	1.59264617733651\\
-0.836802725707364	-1.59264617733651\\
-5.95765553984491	3.67551800804508\\
-5.95765553984491	-3.67551800804508\\
-1.14756145596	0\\
-24.1328644650633	0\\
-0.371237699097274	1.90021164816312\\
-0.371237699097274	-1.90021164816312\\
-1.2789541727959	0.149587440552473\\
-1.2789541727959	-0.149587440552473\\
-6.50004090222358	0\\
-2.36972615698099	3.37391667781388\\
-2.36972615698099	-3.37391667781388\\
-0.880253391907221	0.795235945806763\\
-0.880253391907221	-0.795235945806763\\
-0.871149442919793	-0.758429386456849\\
-4.93160149155927	0\\
-3.12255288708176	3.26865409803844\\
-3.12255288708176	-3.26865409803844\\
-30.0186908392473	0\\
-14.5088298906279	0\\
-0.673031312894794	1.71126805532161\\
-0.673031312894794	-1.71126805532161\\
-30.8095712410966	0\\
-0.186332244816822	1.91464332091412\\
-0.186332244816822	-1.91464332091412\\
-30.1378201482459	0\\
-4.33855501158888	0\\
-3.15254495422618	1.9178070592239\\
-3.15254495422618	-1.9178070592239\\
-1.17817753997938	1.24907481041421\\
-1.17817753997938	-1.24907481041421\\
-8.2511712419551	0\\
-1.26453351078804	2.3849256930757\\
-1.26453351078804	-2.3849256930757\\
-2.59675721774039	0\\
-12.4555276022702	0\\
-1.43936382807673	1.59316615268077\\
-1.43936382807673	-1.59316615268077\\
-1.31332612798012	0.986522401002325\\
-1.31332612798012	-0.986522401002325\\
-11.3912783790852	0\\
-0.943736807458752	2.45660403082273\\
-0.943736807458752	-2.45660403082273\\
-1.91601435402641	0\\
-33.6503589093149	0\\
-0.0887490937215512	1.94962667579813\\
-0.0887490937215512	-1.94962667579813\\
-17.130684572876	0\\
-0.583502610661496	2.16188672232192\\
-0.583502610661496	-2.16188672232192\\
-1.78952735722153	0\\
-32.0694258581946	0\\
-5.72341807635772	0\\
-2.64474261524219	2.79593520038487\\
-2.64474261524219	-2.79593520038487\\
-0.993548346578943	1.02142322839548\\
-0.993548346578943	-1.02142322839548\\
-5.70413337774346	5.01042602212061\\
-5.70413337774346	-5.01042602212061\\
-22.307618040269	0\\
-5.28911636356177	5.42516858114871\\
-5.28911636356177	-5.42516858114871\\
-0.737775946094025	1.52867782924451\\
-0.737775946094025	-1.52867782924451\\
-5.19503987113383	5.62339465038412\\
-5.19503987113383	-5.62339465038412\\
-0.763585917320732	1.35134356947435\\
-0.763585917320732	-1.35134356947435\\
-5.58868865153251	0\\
-2.26709621488193	2.68786086853271\\
-2.26709621488193	-2.68786086853271\\
-1.43855945935182	0.318281226146441\\
-1.43855945935182	-0.318281226146441\\
-2.28989248474779	2.80569736163101\\
-2.28989248474779	-2.80569736163101\\
-1.33203792970616	0.461046669337805\\
-1.33203792970616	-0.461046669337805\\
-6.26016465827332	0\\
-2.42893084109754	3.18257557921371\\
-2.42893084109754	-3.18257557921371\\
-33.7647849060267	0\\
-11.2890729109132	0\\
-0.892476607989083	2.05342349789046\\
-0.892476607989083	-2.05342349789046\\
-2.73346987085311	0\\
-1.04705116894895	0\\
-13.4120517706277	0\\
-0.885971455644575	1.75675653953367\\
-0.885971455644575	-1.75675653953367\\
-1.63254755308049	0.780122374264603\\
-1.63254755308049	-0.780122374264603\\
-4.6232603759132	3.97516599030932\\
-4.6232603759132	-3.97516599030932\\
-1.27175108391967	1.31044466512088\\
-1.27175108391967	-1.31044466512088\\
-1.10046370059075	2.27569751558079\\
-1.10046370059075	-2.27569751558079\\
-1.23373545473448	0.891109799424286\\
-1.23373545473448	-0.891109799424286\\
-35.1147833626225	0\\
-0.891507747630722	2.29886576832917\\
-0.891507747630722	-2.29886576832917\\
-1.19482094770584	0.581916468630537\\
-1.19482094770584	-0.581916468630537\\
-9.22859984468644	0\\
-1.33370611190933	0\\
-34.2959799507998	0\\
-10.8169951226441	0\\
-1.51529703197557	2.92611662499782\\
-1.51529703197557	-2.92611662499782\\
-3.75135860635214	2.84525176068604\\
-3.75135860635214	-2.84525176068604\\
-3.3586737690947	0\\
-4.97413569466879	4.54094815389915\\
-4.97413569466879	-4.54094815389915\\
-1.45828605299515	0\\
-27.3899302114263	0\\
-0.875507354105399	0.638796723854897\\
-0.875507354105399	-0.638796723854897\\
-4.80883880730756	4.39264590830171\\
-4.80883880730756	-4.39264590830171\\
-1.06142181317303	1.34814204224654\\
-1.06142181317303	-1.34814204224654\\
-0.378834523247864	2.3366892357623\\
-0.378834523247864	-2.3366892357623\\
-0.983657706593182	0.565493621688478\\
-0.983657706593182	-0.565493621688478\\
-15.056753491958	0\\
-0.692489518568814	2.15808351740534\\
-0.692489518568814	-2.15808351740534\\
-2.04293074698054	0\\
-28.4526928668578	0\\
-0.187226001249717	2.03477854604573\\
-0.187226001249717	-2.03477854604573\\
-1.07223736461522	0.336043982278035\\
-1.07223736461522	-0.336043982278035\\
-4.27436487827	1.6809072111102\\
-4.27436487827	-1.6809072111102\\
-1.72650577804896	2.03520356739886\\
-1.72650577804896	-2.03520356739886\\
-5.01701007856168	5.0867614953295\\
-5.01701007856168	-5.0867614953295\\
-0.928809355378694	1.33737528175011\\
-0.928809355378694	-1.33737528175011\\
-14.6721191995755	0\\
-0.953316873421632	2.15638894208832\\
-0.953316873421632	-2.15638894208832\\
-4.29645336165384	2.98755068154911\\
-4.29645336165384	-2.98755068154911\\
-2.36549525407227	0\\
-3.70892710550857	1.37216172656706\\
-3.70892710550857	-1.37216172656706\\
-1.47401569009694	1.41961758374131\\
-1.47401569009694	-1.41961758374131\\
-5.75923674439056	2.03357857792009\\
-5.75923674439056	-2.03357857792009\\
-16.983087516229	0\\
-0.831850807485875	2.59127151864433\\
-0.831850807485875	-2.59127151864433\\
-2.31866247897115	3.14073546817197\\
-2.31866247897115	-3.14073546817197\\
-1.07937087819536	0.649162544787365\\
-1.07937087819536	-0.649162544787365\\
-34.9607550411137	0\\
-4.82411833391867	3.70243919461812\\
-4.82411833391867	-3.70243919461812\\
-3.43268109272232	2.38000897545712\\
-3.43268109272232	-2.38000897545712\\
-4.01436818233151	0\\
-1.06013481611193	1.16331446933392\\
-1.06013481611193	-1.16331446933392\\
-3.55408673458828	3.32902948241327\\
-3.55408673458828	-3.32902948241327\\
-5.37595903328393	0\\
-2.2246544327331	2.54999977522676\\
-2.2246544327331	-2.54999977522676\\
-18.9883291059102	0\\
-0.515007196497497	1.83955282691794\\
-0.515007196497497	-1.83955282691794\\
-0.642727881552479	2.43328118486119\\
-0.642727881552479	-2.43328118486119\\
-7.63028786888453	0\\
-2.10726496217861	3.23134983377769\\
-2.10726496217861	-3.23134983377769\\
-4.77580413012287	4.70257661106603\\
-4.77580413012287	-4.70257661106603\\
-1.56818934000492	0\\
-0.940101199874678	1.11600580317221\\
-0.940101199874678	-1.11600580317221\\
-1.16929808664613	0.364583693931625\\
-1.16929808664613	-0.364583693931625\\
-6.0810252965602	0\\
-2.52013813332461	3.05541147377127\\
-2.52013813332461	-3.05541147377127\\
-4.90289987866967	3.78086888097691\\
-4.90289987866967	-3.78086888097691\\
-1.03292380421853	1.6572967390214\\
-1.03292380421853	-1.6572967390214\\
-1.24974458920173	0.279256905268502\\
-1.24974458920173	-0.279256905268502\\
-23.2575521270899	0\\
-1.54300562453733	1.92435454283178\\
-1.54300562453733	-1.92435454283178\\
-32.3741600366176	0\\
-24.6897705670089	0\\
-23.1561796850843	0\\
-8.85018502805068	0\\
-1.58947183826599	2.5260393162822\\
-1.58947183826599	-2.5260393162822\\
-1.24019187594436	0.763893735746358\\
-1.24019187594436	-0.763893735746358\\
-26.7167883530551	0\\
-4.58805240891831	4.24157586624169\\
-4.58805240891831	-4.24157586624169\\
-23.8146150586946	0\\
-32.2499674173407	0\\
-4.80681065432626	4.02258837040953\\
-4.80681065432626	-4.02258837040953\\
-1.18292082185913	1.54198809642959\\
-1.18292082185913	-1.54198809642959\\
-0.501836805548872	1.94227811734411\\
-0.501836805548872	-1.94227811734411\\
-4.60788676122042	3.81140691695553\\
-4.60788676122042	-3.81140691695553\\
-1.36568760229675	2.35976701816954\\
-1.36568760229675	-2.35976701816954\\
-33.5320815150584	0\\
-5.83510834061048	0\\
-20.2983588651661	0\\
-29.8306725274733	0\\
-0.976114165692237	0.391888759911252\\
-0.976114165692237	-0.391888759911252\\
-14.8708240383038	0\\
-5.01026983419511	5.23292293263542\\
-5.01026983419511	-5.23292293263542\\
-0.849737381434554	1.22910440300954\\
-0.849737381434554	-1.22910440300954\\
-5.07814605425472	4.66630816358126\\
-5.07814605425472	-4.66630816358126\\
-17.319199606368	0\\
-19.5033608536921	0\\
-0.90650238039495	1.86524247987004\\
-0.90650238039495	-1.86524247987004\\
-1.27885084758709	0.636674986732782\\
-1.27885084758709	-0.636674986732782\\
-6.36159549036658	0\\
-2.32610207401906	3.2634477362609\\
-2.32610207401906	-3.2634477362609\\
-14.1046025072396	0\\
-1.06887370304893	2.69197143737699\\
-1.06887370304893	-2.69197143737699\\
-2.81433767938439	3.30057000763041\\
-2.81433767938439	-3.30057000763041\\
-13.7628206291386	0\\
-2.58664918410072	2.93083942405965\\
-2.58664918410072	-2.93083942405965\\
-4.72182294412589	3.56750938245345\\
-4.72182294412589	-3.56750938245345\\
-1.68827948143579	0\\
-11.050870067481	0\\
-3.48169425166015	0\\
-0.875703343175341	1.46283441492426\\
-0.875703343175341	-1.46283441492426\\
-0.421931768829673	2.19082956836109\\
-0.421931768829673	-2.19082956836109\\
-18.2772792347107	0\\
-0.559775529884123	1.72322663375519\\
-0.559775529884123	-1.72322663375519\\
-2.21550214673108	0\\
-4.84206354645215	3.90430364685929\\
-4.84206354645215	-3.90430364685929\\
-18.4702259319212	0\\
-0.786197602077553	1.94648483859169\\
-0.786197602077553	-1.94648483859169\\
-19.7138373387528	0\\
-0.895407528695259	0.90403153528087\\
-0.895407528695259	-0.90403153528087\\
-33.2548456877698	0\\
-0.0871969963894497	2.07776676592162\\
-0.0871969963894497	-2.07776676592162\\
-5.23485930526875	5.78034446335557\\
-5.23485930526875	-5.78034446335557\\
-28.6786076474638	0\\
-6.04435591563011	5.07685508098996\\
-6.04435591563011	-5.07685508098996\\
-25.5626152646672	0\\
-10.3003812832794	0\\
-33.0189620634882	0\\
-1.08268641112423	1.45100371050328\\
-1.08268641112423	-1.45100371050328\\
-4.23037896841203	1.24991250130667\\
-4.23037896841203	-1.24991250130667\\
-4.85093280057877	4.1392768166926\\
-4.85093280057877	-4.1392768166926\\
-4.99675726062401	4.3983507447955\\
-4.99675726062401	-4.3983507447955\\
-2.94051391717813	2.19805711665895\\
-2.94051391717813	-2.19805711665895\\
-1.88766506559149	2.01049667163864\\
-1.88766506559149	-2.01049667163864\\
-1.32126122586884	1.13642402260132\\
-1.32126122586884	-1.13642402260132\\
-28.2341170844765	0\\
-30.6370351170855	0\\
-24.5460458946713	0\\
-2.68780687350285	2.55181930482468\\
-2.68780687350285	-2.55181930482468\\
-1.10345647853994	1.05102268328419\\
-1.10345647853994	-1.05102268328419\\
-4.1809482518742	2.94488894639247\\
-4.1809482518742	-2.94488894639247\\
-25.3361212531455	0\\
-0.449144853820629	2.07501870664053\\
-0.449144853820629	-2.07501870664053\\
-1.42196136754464	2.60023906398879\\
-1.42196136754464	-2.60023906398879\\
-33.3638519629413	0\\
-10.4853917483246	0\\
-25.1981444067269	0\\
-31.7618457700584	0\\
-3.61273978052005	2.53750284327189\\
-3.61273978052005	-2.53750284327189\\
-14.2229109821418	0\\
-15.8510837876711	0\\
-0.718592120310572	2.03097028894559\\
-0.718592120310572	-2.03097028894559\\
-4.24382112089478	2.56162783368424\\
-4.24382112089478	-2.56162783368424\\
-0.656257405839523	2.26854620592522\\
-0.656257405839523	-2.26854620592522\\
-5.20179103601442	4.94646410571224\\
-5.20179103601442	-4.94646410571224\\
-15.6000141679959	0\\
-11.8725521851711	0\\
-4.14942238715302	0\\
-7.97591623609794	0\\
-2.17484942004666	2.26813649855124\\
-2.17484942004666	-2.26813649855124\\
-31.0651238041424	0\\
-4.14838999524555	0.948711147000053\\
-4.14838999524555	-0.948711147000053\\
-1.84082219434135	2.17244595742313\\
-1.84082219434135	-2.17244595742313\\
-4.52369806439364	2.39566964987486\\
-4.52369806439364	-2.39566964987486\\
-7.23312821688712	1.54927567948248\\
-7.23312821688712	-1.54927567948248\\
-5.28546181561296	5.55280820421856\\
-5.28546181561296	-5.55280820421856\\
-18.8872163508575	0\\
-4.63728340203963	3.35267406375259\\
-4.63728340203963	-3.35267406375259\\
-16.4152055231934	0\\
-31.4981856445714	0\\
-21.3196004150618	0\\
-0.52689711517925	2.39015694603022\\
-0.52689711517925	-2.39015694603022\\
-34.6414279384539	0\\
-21.0858422842628	0\\
-3.95194688714	1.44114660404532\\
-3.95194688714	-1.44114660404532\\
-1.54648635165249	1.56168903901609\\
-1.54648635165249	-1.56168903901609\\
-5.29229560169492	5.90653625248452\\
-5.29229560169492	-5.90653625248452\\
-34.0777757484388	0\\
-23.6338611869579	0\\
-21.5127058843139	0\\
-29.6598451128683	0\\
-4.50646532104751	3.05298566477798\\
-4.50646532104751	-3.05298566477798\\
-16.632240226763	0\\
-27.8024753751701	0\\
-4.19010647915629	3.77085820793642\\
-4.19010647915629	-3.77085820793642\\
-0.749094026113161	2.37360359179802\\
-0.749094026113161	-2.37360359179802\\
-2.10904386660503	1.23620796937965\\
-2.10904386660503	-1.23620796937965\\
-4.53608443447722	2.92209460007507\\
-4.53608443447722	-2.92209460007507\\
-7.26008198048149	0\\
-2.36013504773223	2.11442112238839\\
-2.36013504773223	-2.11442112238839\\
-16.8390204100994	0\\
-27.0458241196882	0\\
-12.2841528294321	0\\
-22.1002035197341	0\\
-4.96400694566554	3.86990602879443\\
-4.96400694566554	-3.86990602879443\\
-1.0909069222541	1.75021666551848\\
-1.0909069222541	-1.75021666551848\\
-24.2730277440096	0\\
-9.90557286154696	0\\
-0.61858992379083	1.53766355027173\\
-0.61858992379083	-1.53766355027173\\
-30.4253087768265	0\\
-0.748876911406451	2.49765444541365\\
-0.748876911406451	-2.49765444541365\\
-12.6698026493275	0\\
-0.838694967102021	2.42122903080463\\
-0.838694967102021	-2.42122903080463\\
-1.36788752201062	0.245425356367847\\
-1.36788752201062	-0.245425356367847\\
-0.61796319687395	1.93909397117264\\
-0.61796319687395	-1.93909397117264\\
-2.31013254019723	2.99283614356295\\
-2.31013254019723	-2.99283614356295\\
-3.93404151009715	3.53144788558634\\
-3.93404151009715	-3.53144788558634\\
-3.23692116982335	0\\
-25.4560689336026	0\\
-12.1024362837945	0\\
-35.224717770888	0\\
-5.95202095066863	0\\
-31.9247646122514	0\\
-24.8789967013332	0\\
-6.48635338440222	2.75516133134038\\
-6.48635338440222	-2.75516133134038\\
-34.4959862899512	0\\
-0.930846592182406	1.65330051720335\\
-0.930846592182406	-1.65330051720335\\
-2.49286344221319	2.84639433302947\\
-2.49286344221319	-2.84639433302947\\
-1.08934468082691	0.908899417192576\\
-1.08934468082691	-0.908899417192576\\
-18.6218841720851	0\\
-7.73473543194494	0\\
-17.5886329271654	0\\
-1.10851381873072	0.7491996475206\\
-1.10851381873072	-0.7491996475206\\
-9.65191866252739	0\\
-1.71031505975941	2.91304859706595\\
-1.71031505975941	-2.91304859706595\\
-4.09716467135938	1.51066386609062\\
-4.09716467135938	-1.51066386609062\\
-1.56450803821493	1.69053947641208\\
-1.56450803821493	-1.69053947641208\\
-2.78014073051278	2.49666850071192\\
-2.78014073051278	-2.49666850071192\\
-4.76059766350448	3.03319545706504\\
-4.76059766350448	-3.03319545706504\\
-1.28353723806763	1.87506538329334\\
-1.28353723806763	-1.87506538329334\\
-18.077105003799	0\\
-29.5393290875841	0\\
-32.5052893687518	0\\
-33.9395517804825	0\\
-29.1339467891055	0\\
-17.9014224996994	0\\
-20.8268406244075	0\\
-20.1572811530038	0\\
-4.19825388401532	2.35387090299731\\
-4.19825388401532	-2.35387090299731\\
-1.19392959775845	1.40018526723834\\
-1.19392959775845	-1.40018526723834\\
-5.05144827381247	0\\
-2.83564421072608	2.95258858812053\\
-2.83564421072608	-2.95258858812053\\
-2.48733236158944	2.18304626799488\\
-2.48733236158944	-2.18304626799488\\
-1.49920048238943	0.896530894054585\\
-1.49920048238943	-0.896530894054585\\
-12.8993546041654	0\\
-3.2223289586599	1.69554378980737\\
-3.2223289586599	-1.69554378980737\\
-1.33377609775647	0.351961596932439\\
-1.33377609775647	-0.351961596932439\\
-5.08218784503678	5.34170278073726\\
-5.08218784503678	-5.34170278073726\\
-1.08008726063687	1.90399693402097\\
-1.08008726063687	-1.90399693402097\\
-4.40575352436279	1.58256754406251\\
-4.40575352436279	-1.58256754406251\\
-2.69889773899808	3.32320223210023\\
-2.69889773899808	-3.32320223210023\\
-4.18836257395318	1.56108255936086\\
-4.18836257395318	-1.56108255936086\\
-4.66827834719255	0\\
-2.58023466937639	1.87072878474186\\
-2.58023466937639	-1.87072878474186\\
-1.58562615702733	1.03373421935856\\
-1.58562615702733	-1.03373421935856\\
-7.06357207170173	1.71554651721974\\
-7.06357207170173	-1.71554651721974\\
-34.8186503653225	0\\
-1.85557093864388	3.08005146831607\\
-1.85557093864388	-3.08005146831607\\
-3.84966365824065	0\\
-16.1465880606529	0\\
-22.7784163036012	0\\
-1.99073238667087	2.94044371161552\\
-1.99073238667087	-2.94044371161552\\
-21.7600835771444	0\\
-20.5724813984514	0\\
-1.18988817908675	0.100351675434121\\
-1.18988817908675	-0.100351675434121\\
-18.7344926245429	0\\
-0.551370987236318	2.05090876237254\\
-0.551370987236318	-2.05090876237254\\
-15.9830986422595	0\\
-2.84595532640161	2.36532048820892\\
-2.84595532640161	-2.36532048820892\\
-1.52473315698606	0.502570687387437\\
-1.52473315698606	-0.502570687387437\\
-23.033515506632	0\\
-5.05992726989294	4.98806598702342\\
-5.05992726989294	-4.98806598702342\\
-25.8417344704539	0\\
-26.9298136310535	0\\
-5.26441990575557	5.31031719869033\\
-5.26441990575557	-5.31031719869033\\
-1.81792570691615	2.59626183399852\\
-1.81792570691615	-2.59626183399852\\
-31.6136558106646	0\\
-6.40453723660793	4.2386852322243\\
-6.40453723660793	-4.2386852322243\\
-9.52053739213921	0\\
-9.47354241259439	0\\
-1.5779649474383	2.84525322540093\\
-1.5779649474383	-2.84525322540093\\
-6.67507373582559	3.80840355920135\\
-6.67507373582559	-3.80840355920135\\
-26.0287531848283	0\\
-6.52688199555218	1.79147254211522\\
-6.52688199555218	-1.79147254211522\\
-5.54047569623247	4.29390800622329\\
-5.54047569623247	-4.29390800622329\\
-19.2475556850733	0\\
-20.4284838868606	0\\
-4.81352493452688	3.26891294819841\\
-4.81352493452688	-3.26891294819841\\
-5.2033324199977	0\\
-2.17966687908627	2.45342903059238\\
-2.17966687908627	-2.45342903059238\\
-1.45953560350143	0.139261407460604\\
-1.45953560350143	-0.139261407460604\\
-2.73813130587714	2.85059761866999\\
-2.73813130587714	-2.85059761866999\\
-8.67116696453946	0\\
-2.4850869757633	0\\
-1.73618283834331	2.1449174803213\\
-1.73618283834331	-2.1449174803213\\
-6.56829137038528	4.05870721861553\\
-6.56829137038528	-4.05870721861553\\
-19.8254965169214	0\\
-12.7765631720221	0\\
-4.84223809926599	3.46039803166513\\
-4.84223809926599	-3.46039803166513\\
-6.29110727584997	3.08204160731373\\
-6.29110727584997	-3.08204160731373\\
-1.13888870636117	2.44498519103373\\
-1.13888870636117	-2.44498519103373\\
-15.3472449811903	0\\
-0.972071770458409	2.65773367183204\\
-0.972071770458409	-2.65773367183204\\
-3.17590355473344	3.14856022630311\\
-3.17590355473344	-3.14856022630311\\
-2.0904057363339	2.32931643975046\\
-2.0904057363339	-2.32931643975046\\
-5.4327351752459	4.86014013136508\\
-5.4327351752459	-4.86014013136508\\
-29.4076885660797	0\\
-26.3131632259838	0\\
-3.2758231659695	2.5413793766978\\
-3.2758231659695	-2.5413793766978\\
-17.7577117487139	0\\
-10.6830077334525	0\\
-1.69104771357326	2.56786554488181\\
-1.69104771357326	-2.56786554488181\\
-4.05782309978464	3.63572024693335\\
-4.05782309978464	-3.63572024693335\\
-2.21838203965839	3.1935746206118\\
-2.21838203965839	-3.1935746206118\\
-4.62652405535117	2.45127003626642\\
-4.62652405535117	-2.45127003626642\\
-1.4100411519455	1.97870548542858\\
-1.4100411519455	-1.97870548542858\\
-0.709881096118262	1.81664838330037\\
-0.709881096118262	-1.81664838330037\\
-1.66557863915186	0.511209118817555\\
-1.66557863915186	-0.511209118817555\\
-32.7396891298392	0\\
-13.1035739346934	0\\
-2.86295528687636	3.09773667882101\\
-2.86295528687636	-3.09773667882101\\
-4.3172452585368	1.43509047989829\\
-4.3172452585368	-1.43509047989829\\
-20.9377047809271	0\\
-6.56465014475012	0.810870844701593\\
-6.56465014475012	-0.810870844701593\\
-22.6136205658689	0\\
-22.6038308487715	0\\
-4.78608872193761	3.13812292859389\\
-4.78608872193761	-3.13812292859389\\
-4.65601786704342	4.46516775882182\\
-4.65601786704342	-4.46516775882182\\
-31.3079982084302	0\\
-4.71450685533612	4.57804650696782\\
-4.71450685533612	-4.57804650696782\\
-5.72828892731593	3.71162766611564\\
-5.72828892731593	-3.71162766611564\\
-1.12529873558801	0.24159618391523\\
-1.12529873558801	-0.24159618391523\\
-27.2727106554603	0\\
-4.02294207149882	0.337493590084557\\
-4.02294207149882	-0.337493590084557\\
-1.94996593140191	2.21995046515174\\
-1.94996593140191	-2.21995046515174\\
-28.5620039940421	0\\
-11.9982631324161	0\\
-4.40517077351834	1.6916910100596\\
-4.40517077351834	-1.6916910100596\\
-4.51636962773941	1.92669005631968\\
-4.51636962773941	-1.92669005631968\\
-4.18714860345114	1.06535780690769\\
-4.18714860345114	-1.06535780690769\\
-1.70874162250901	0.62301916838066\\
-1.70874162250901	-0.62301916838066\\
-13.9549562752516	0\\
-4.44072876643614	0\\
-33.127835975953	0\\
-2.4154664712828	2.3068794118528\\
-2.4154664712828	-2.3068794118528\\
-6.71352202877205	0\\
-2.07873403110468	1.42315670669743\\
-2.07873403110468	-1.42315670669743\\
-2.3509849307164	0.946053623637067\\
-2.3509849307164	-0.946053623637067\\
-4.10267767584308	1.70555594099079\\
-4.10267767584308	-1.70555594099079\\
-3.0106620658364	2.05148934303936\\
-3.0106620658364	-2.05148934303936\\
-6.55228921408238	2.12858530021151\\
-6.55228921408238	-2.12858530021151\\
-8.55158824192642	0\\
-6.53055990176416	3.28992817484788\\
-6.53055990176416	-3.28992817484788\\
-0.994768524674187	0.790969526443877\\
-0.994768524674187	-0.790969526443877\\
-1.65034350605421	1.53754671491767\\
-1.65034350605421	-1.53754671491767\\
-1.84864973053812	0.911731845636942\\
-1.84864973053812	-0.911731845636942\\
-28.0279906220902	0\\
-26.5270672176486	0\\
-4.43597567534631	2.22897185992858\\
-4.43597567534631	-2.22897185992858\\
-1.41143892393508	1.78144548272745\\
-1.41143892393508	-1.78144548272745\\
-4.81906488651795	0\\
-5.53872674515505	2.89163552407894\\
-5.53872674515505	-2.89163552407894\\
-19.1140614198837	0\\
-1.56040846211852	1.79653058607535\\
-1.56040846211852	-1.79653058607535\\
-4.86885749835364	4.90342919074475\\
-4.86885749835364	-4.90342919074475\\
-5.2322813783062	5.17979102613361\\
-5.2322813783062	-5.17979102613361\\
-1.25943715749558	1.62833117810776\\
-1.25943715749558	-1.62833117810776\\
-3.34723687139898	1.97423585816531\\
-3.34723687139898	-1.97423585816531\\
-2.50472804771789	1.17900770819549\\
-2.50472804771789	-1.17900770819549\\
-5.25815199918789	2.52352370551653\\
-5.25815199918789	-2.52352370551653\\
-5.12664270878601	5.5152645467306\\
-5.12664270878601	-5.5152645467306\\
-4.67745558538276	2.69464649260625\\
-4.67745558538276	-2.69464649260625\\
-1.44035341080461	0.594716990389983\\
-1.44035341080461	-0.594716990389983\\
-6.83638155036041	0\\
-1.77845681980571	2.41955661687573\\
-1.77845681980571	-2.41955661687573\\
-1.59274310607217	0.147950016603923\\
-1.59274310607217	-0.147950016603923\\
-4.88076457204693	4.76510609954448\\
-4.88076457204693	-4.76510609954448\\
-1.4228700738986	0.715635617039732\\
-1.4228700738986	-0.715635617039732\\
-4.04560691649353	3.03496776854961\\
-4.04560691649353	-3.03496776854961\\
-4.31104125292241	2.67393262658795\\
-4.31104125292241	-2.67393262658795\\
-4.21592148628393	3.13035880242459\\
-4.21592148628393	-3.13035880242459\\
-0.282891270728795	1.98662280346257\\
-0.282891270728795	-1.98662280346257\\
-21.9715858952216	0\\
-2.69430781026738	2.68296909249568\\
-2.69430781026738	-2.68296909249568\\
-5.12196877172244	5.24144040703801\\
};

\end{axis}

\end{tikzpicture}%

%% file: figure/1_agent_ps.tex
%
%
\begin{tikzpicture}

\begin{axis}[%
width=4.521in,
height=3.566in,
at={(0.758in,0.481in)},
scale only axis,
xmin=-0.4,
xmax=0,
ymin=-2.5,
ymax=2.5,
xtick={-0.4,-0.2,0},
ytick={-2.5,0,2.5},
axis background/.style={fill=white},
legend style={legend cell align=left, align=left, draw=white!15!black}
]
\addplot [color=black, draw=none, mark=asterisk, mark options={solid, black}]
  table[row sep=crcr]{%
-0.0306850936689468	1.8445472375376\\
-0.0306850936689468	-1.8445472375376\\
-0.1	1.90211303259031\\
-0.1	-1.90211303259031\\
-0.1	0\\
-0.00828843526465679	0\\
-0.1	1.17557050458495\\
-0.1	-1.17557050458495\\
-0.00978900719670908	1.14298125995345\\
-0.00978900719670908	-1.14298125995345\\
};

\addplot+[color=blue, draw=none, mark=*, mark size=1pt,mark options={solid, blue}]
  table[row sep=crcr]{%
-0.1	0\\
-0.1	1.17557050458495\\
-0.1	-1.17557050458495\\
-0.1	1.90211303259031\\
-0.1	-1.90211303259031\\
-0.1	1.90211303259031\\
-0.1	-1.90211303259031\\
-0.0999999999999999	1.17557050458495\\
-0.0999999999999999	-1.17557050458495\\
-0.1	0\\
0	0\\
-0.122937551519874	1.95048865225197\\
-0.122937551519874	-1.95048865225197\\
-0.11838630409114	0\\
-0.118499106453145	1.2029459778814\\
-0.118499106453145	-1.2029459778814\\
-0.185989029826253	1.9628945917221\\
-0.185989029826253	-1.9628945917221\\
-0.164293521993296	1.21530685088756\\
-0.164293521993296	-1.21530685088756\\
-0.163711708858955	0\\
-0.0353190516437994	1.82778359684824\\
-0.0353190516437994	-1.82778359684824\\
-0.00711852836529847	1.1280129799572\\
-0.00711852836529847	-1.1280129799572\\
-0.150417711961804	2.01630726793759\\
-0.150417711961804	-2.01630726793759\\
-0.131454388120136	1.2330232397513\\
-0.131454388120136	-1.2330232397513\\
-0.131556935920875	0\\
-0.112991966567579	1.78505960175808\\
-0.112991966567579	-1.78505960175808\\
-0.125371343098858	1.07606605208168\\
-0.125371343098858	-1.07606605208168\\
-0.185489805309562	1.94012366136739\\
-0.185489805309562	-1.94012366136739\\
-0.170023889981333	1.20479883626604\\
-0.170023889981333	-1.20479883626604\\
-0.0325036548519514	1.93727882361899\\
-0.0325036548519514	-1.93727882361899\\
-0.0439583483547744	0\\
-0.0434267156145072	1.20053073602539\\
-0.0434267156145072	-1.20053073602539\\
-0.15427006456279	1.84744699783289\\
-0.15427006456279	-1.84744699783289\\
-0.170383773766133	1.14197602816796\\
-0.170383773766133	-1.14197602816796\\
-0.168584717363852	1.83955872580931\\
-0.168584717363852	-1.83955872580931\\
-0.1918008582462	1.13884217557823\\
-0.1918008582462	-1.13884217557823\\
-0.193695601162466	0\\
-0.143711677657195	1.80591936099183\\
-0.143711677657195	-1.80591936099183\\
-0.172992580207173	1.1027566642581\\
-0.172992580207173	-1.1027566642581\\
-0.178174407154428	0\\
-0.0147555420339395	1.86700628631101\\
-0.0147555420339395	-1.86700628631101\\
-0.00280442498309172	1.16333380297319\\
-0.00280442498309172	-1.16333380297319\\
-0.0700169694214109	1.88690536870552\\
-0.0700169694214109	-1.88690536870552\\
-0.0677872006290551	0\\
-0.0678984818955318	1.16710477106353\\
-0.0678984818955318	-1.16710477106353\\
-0.0781771055732858	1.7884168434108\\
-0.0781771055732858	-1.7884168434108\\
-0.0585050380997539	1.08064477131292\\
-0.0585050380997539	-1.08064477131292\\
-0.0728521313474001	1.07859391068632\\
-0.0728521313474001	-1.07859391068632\\
-0.111688156057576	1.94863578731232\\
-0.111688156057576	-1.94863578731232\\
-0.0928811946625575	2.02696165094136\\
-0.0928811946625575	-2.02696165094136\\
-0.0956890303741359	1.23524550628024\\
-0.0956890303741359	-1.23524550628024\\
-0.122219202879665	1.85771913811655\\
-0.122219202879665	-1.85771913811655\\
-0.127692344415567	1.14575953974199\\
-0.127692344415567	-1.14575953974199\\
-0.154106830294108	1.99484287240651\\
-0.154106830294108	-1.99484287240651\\
-0.0759269360815342	1.84509816497408\\
-0.0759269360815342	-1.84509816497408\\
-0.0678648692787363	1.13590860552846\\
-0.0678648692787363	-1.13590860552846\\
-0.125562775641474	1.93567933952304\\
-0.125562775641474	-1.93567933952304\\
-0.134271480870375	1.81358070063203\\
-0.134271480870375	-1.81358070063203\\
-0.154869045521949	1.1088645913657\\
-0.154869045521949	-1.1088645913657\\
-0.173731639294298	1.99176712775403\\
-0.173731639294298	-1.99176712775403\\
-0.149776583050729	1.22570452075553\\
-0.149776583050729	-1.22570452075553\\
-0.149621189115578	0\\
-0.00811270794780045	1.91676023757488\\
-0.00811270794780045	-1.91676023757488\\
-0.0173491762692898	1.19434502093238\\
-0.0173491762692898	-1.19434502093238\\
-0.018401434901811	0\\
-0.0392776627395243	2.00796985434407\\
-0.0392776627395243	-2.00796985434407\\
-0.0611199495315679	1.23069712234875\\
-0.0611199495315679	-1.23069712234875\\
-0.0818871707963267	2.02966014116947\\
-0.0818871707963267	-2.02966014116947\\
-0.0890646695085173	0\\
-0.174663595255801	1.88415683896654\\
-0.174663595255801	-1.88415683896654\\
-0.179009509606712	1.1719707986688\\
-0.179009509606712	-1.1719707986688\\
-0.0563176761316869	1.95094778912278\\
-0.0563176761316869	-1.95094778912278\\
-0.0650278013722123	1.20455857296012\\
-0.0650278013722123	-1.20455857296012\\
-0.092349686348935	1.83898953746483\\
-0.092349686348935	-1.83898953746483\\
-0.0893892677153744	1.12986371327731\\
-0.0893892677153744	-1.12986371327731\\
-0.166487460276542	1.09412221691403\\
-0.166487460276542	-1.09412221691403\\
-0.122815910497779	1.81484811721375\\
-0.122815910497779	-1.81484811721375\\
-0.136437374997714	1.10876262507019\\
-0.136437374997714	-1.10876262507019\\
-0.149413323679659	1.82967934408931\\
-0.149413323679659	-1.82967934408931\\
-0.171088351310382	1.12629703719189\\
-0.171088351310382	-1.12629703719189\\
-0.0217102357760463	1.98904745637404\\
-0.0217102357760463	-1.98904745637404\\
-0.0467370949923183	1.22516957805295\\
-0.0467370949923183	-1.22516957805295\\
-0.051361420633252	-1.22436930717329\\
-0.0456220420182619	1.98406630533722\\
-0.0456220420182619	-1.98406630533722\\
-0.0619119016076602	1.2205251894218\\
-0.0619119016076602	-1.2205251894218\\
-0.0599400606470815	1.86714438026411\\
-0.0599400606470815	-1.86714438026411\\
-0.0528157046164682	1.1545666319942\\
-0.0528157046164682	-1.1545666319942\\
-0.0910473494425372	1.93659083940618\\
-0.0910473494425372	-1.93659083940618\\
-0.0923321543573419	1.19536835823764\\
-0.0923321543573419	-1.19536835823764\\
-0.107258555186213	1.9850950928064\\
-0.107258555186213	-1.9850950928064\\
-0.1051147438231	1.21852553348439\\
-0.1051147438231	-1.21852553348439\\
-0.174370602652607	1.86027504821315\\
-0.174370602652607	-1.86027504821315\\
-0.188646859598473	1.15593646152192\\
-0.188646859598473	-1.15593646152192\\
-0.133735100506271	2.01823695872796\\
-0.133735100506271	-2.01823695872796\\
-0.120992215192531	1.23274965683642\\
-0.120992215192531	-1.23274965683642\\
-0.0875940593874647	1.79939244529917\\
-0.0875940593874647	-1.79939244529917\\
-0.078040124381868	1.09217047158439\\
-0.078040124381868	-1.09217047158439\\
-0.0523826158381861	1.92911378403481\\
-0.0523826158381861	-1.92911378403481\\
-0.0582774771149625	1.19381876240642\\
-0.0582774771149625	-1.19381876240642\\
-0.110528559631058	1.18066503000329\\
-0.110528559631058	-1.18066503000329\\
-0.110941481263581	1.9102660859205\\
-0.110941481263581	-1.9102660859205\\
-0.0145205900577254	1.93882982734985\\
-0.0145205900577254	-1.93882982734985\\
-0.0304163752949258	0\\
-0.0296109124978963	1.20417484532741\\
-0.0296109124978963	-1.20417484532741\\
-0.0290423079836232	1.86208020578012\\
-0.0290423079836232	-1.86208020578012\\
-0.0159529935764115	1.15641710408193\\
-0.0159529935764115	-1.15641710408193\\
-0.141416871860489	1.89208828906556\\
-0.141416871860489	-1.89208828906556\\
-0.143057359970359	1.17159806443507\\
-0.143057359970359	-1.17159806443507\\
-0.133181277213756	1.8994247643764\\
-0.133181277213756	-1.8994247643764\\
-0.133418557024203	1.17536330491464\\
-0.133418557024203	-1.17536330491464\\
-0.00950843903324761	1.87948212869298\\
-0.00950843903324761	-1.87948212869298\\
-0.159971300630426	1.85910585706808\\
-0.159971300630426	-1.85910585706808\\
-0.172772476859819	1.15193661642842\\
-0.172772476859819	-1.15193661642842\\
-0.185295439413936	1.86846198519741\\
-0.185295439413936	-1.86846198519741\\
-0.196545654103987	1.16431988148059\\
-0.196545654103987	-1.16431988148059\\
-0.0979035169351111	1.79953524630875\\
-0.0979035169351111	-1.79953524630875\\
-0.096287382711439	1.09200698060692\\
-0.096287382711439	-1.09200698060692\\
-0.189410187096048	1.88184076424487\\
-0.189410187096048	-1.88184076424487\\
-0.0412206549067286	1.8679313391958\\
-0.0412206549067286	-1.8679313391958\\
-0.0317970016140664	1.15803610131413\\
-0.0317970016140664	-1.15803610131413\\
-0.116021821674521	1.80531055584219\\
-0.116021821674521	-1.80531055584219\\
-0.127251751343416	1.09863166160632\\
-0.127251751343416	-1.09863166160632\\
-0.0869646229104069	2.00346060322901\\
-0.0869646229104069	-2.00346060322901\\
-0.0867551710264222	1.8573330526223\\
-0.0867551710264222	-1.8573330526223\\
-0.0834248648472081	1.1449311971142\\
-0.0834248648472081	-1.1449311971142\\
-0.0235987719213118	1.84731817216132\\
-0.0235987719213118	-1.84731817216132\\
-0.00270396309586657	1.14693241478957\\
-0.00270396309586657	-1.14693241478957\\
-0.0305554263008474	1.96623736455376\\
-0.0305554263008474	-1.96623736455376\\
-0.0482648199425554	1.21444037473106\\
-0.0482648199425554	-1.21444037473106\\
-0.195437762392089	1.90648709898103\\
-0.195437762392089	-1.90648709898103\\
-0.18952189881517	1.18967072481178\\
-0.18952189881517	-1.18967072481178\\
-0.0120475868913031	1.95442406322225\\
-0.0120475868913031	-1.95442406322225\\
-0.1085392705642	2.00686407275342\\
-0.1085392705642	-2.00686407275342\\
-0.0829263972994055	2.01815773276026\\
-0.0829263972994055	-2.01815773276026\\
-0.0233308749322011	1.94718112289433\\
-0.0233308749322011	-1.94718112289433\\
-0.0973288415411048	1.78379463651597\\
-0.0973288415411048	-1.78379463651597\\
-0.0947276771102331	1.07422254955584\\
-0.0947276771102331	-1.07422254955584\\
-0.0750902051131568	1.82758607321615\\
-0.0750902051131568	-1.82758607321615\\
-0.0632210025499866	1.12091864593188\\
-0.0632210025499866	-1.12091864593188\\
-0.115083214221101	1.09319725425105\\
-0.115083214221101	-1.09319725425105\\
-0.165392244891482	1.91446504988819\\
-0.165392244891482	-1.91446504988819\\
-0.160547707626582	1.18818596725673\\
-0.160547707626582	-1.18818596725673\\
-0.0234209019866062	1.87287035497012\\
-0.0234209019866062	-1.87287035497012\\
-0.0553371023239922	0\\
-0.0625090814680224	2.01101415013313\\
-0.0625090814680224	-2.01101415013313\\
-0.0760867690043038	1.23018581855254\\
-0.0760867690043038	-1.23018581855254\\
-0.0984615476705941	2.01865279727905\\
-0.0984615476705941	-2.01865279727905\\
-0.112512847308137	2.03287698143796\\
-0.112512847308137	-2.03287698143796\\
-0.107422880767868	1.23747490613664\\
-0.107422880767868	-1.23747490613664\\
-0.124322410735208	2.00679666546327\\
-0.124322410735208	-2.00679666546327\\
-0.166774004319383	1.96228232820479\\
-0.166774004319383	-1.96228232820479\\
-0.150589080793312	1.21233513440856\\
-0.150589080793312	-1.21233513440856\\
-0.0485825397046292	1.91441662924145\\
-0.0485825397046292	-1.91441662924145\\
-0.163147669960754	1.99035461008164\\
-0.163147669960754	-1.99035461008164\\
-0.0497628659006895	1.89384947955506\\
-0.0497628659006895	-1.89384947955506\\
-0.0484265746364501	1.17381310642688\\
-0.0484265746364501	-1.17381310642688\\
-0.144661426661148	1.85263986255806\\
-0.144661426661148	-1.85263986255806\\
-0.156661076607138	1.14438270420107\\
-0.156661076607138	-1.14438270420107\\
-0.163719991775317	1.82736097751024\\
-0.163719991775317	-1.82736097751024\\
-0.191809243710024	1.127404807643\\
-0.191809243710024	-1.127404807643\\
-0.0897059329473694	1.17935807132661\\
-0.0897059329473694	-1.17935807132661\\
-0.0894091034611195	1.90809357981069\\
-0.0894091034611195	-1.90809357981069\\
-0.192827621214776	1.17885447535759\\
-0.192827621214776	-1.17885447535759\\
-0.0379903340996632	1.89521252134657\\
-0.0379903340996632	-1.89521252134657\\
-0.037175063503445	1.17643776718244\\
-0.037175063503445	-1.17643776718244\\
-0.125794970007283	1.9232181360927\\
-0.125794970007283	-1.9232181360927\\
-0.123355693512432	1.18877550656375\\
-0.123355693512432	-1.18877550656375\\
-0.0709359445551443	1.80242059879014\\
-0.0709359445551443	-1.80242059879014\\
-0.0498883234021694	1.09685955449551\\
-0.0498883234021694	-1.09685955449551\\
-0.121199083803611	1.79577518820606\\
-0.121199083803611	-1.79577518820606\\
-0.138339118086561	1.08887023456661\\
-0.138339118086561	-1.08887023456661\\
-0.167050665615054	1.92995432018075\\
-0.167050665615054	-1.92995432018075\\
-0.155472758982438	1.81589544333301\\
-0.155472758982438	-1.81589544333301\\
-0.186207231037145	1.11501240941886\\
-0.186207231037145	-1.11501240941886\\
-0.143926269460876	1.95121353301928\\
-0.143926269460876	-1.95121353301928\\
-0.135126378058109	1.2047092509334\\
-0.135126378058109	-1.2047092509334\\
-0.110232009801248	1.81451725528502\\
-0.110232009801248	-1.81451725528502\\
-0.116431177445061	1.10759235446278\\
-0.116431177445061	-1.10759235446278\\
-0.199190006215834	1.15303153875275\\
-0.199190006215834	-1.15303153875275\\
-0.0897755682761784	1.95838393734948\\
-0.0897755682761784	-1.95838393734948\\
-0.0920023358126835	1.20638504661345\\
-0.0920023358126835	-1.20638504661345\\
-0.035596377753999	1.85099950908894\\
-0.035596377753999	-1.85099950908894\\
-0.159139736563441	1.8852585840499\\
-0.159139736563441	-1.8852585840499\\
-0.162944102488018	1.16972771288288\\
-0.162944102488018	-1.16972771288288\\
-0.0783682922404391	1.8750473671089\\
-0.0783682922404391	-1.8750473671089\\
-0.0753700491354926	1.15841833766389\\
-0.0753700491354926	-1.15841833766389\\
-0.0506660436704283	2.01377543644952\\
-0.0506660436704283	-2.01377543644952\\
-0.0892069452514007	1.89780723743072\\
-0.0892069452514007	-1.89780723743072\\
-0.00836568605766858	1.17722234482927\\
-0.00836568605766858	-1.17722234482927\\
-0.136257009989281	2.00279029370229\\
-0.136257009989281	-2.00279029370229\\
-0.145116615026835	1.88224919137806\\
-0.145116615026835	-1.88224919137806\\
-0.0481046878360061	1.88116526995785\\
-0.0481046878360061	-1.88116526995785\\
-0.195593094537415	1.91997458050789\\
-0.195593094537415	-1.91997458050789\\
-0.122857686558077	2.03075840684289\\
-0.122857686558077	-2.03075840684289\\
-0.147008338335165	1.86707195540377\\
-0.147008338335165	-1.86707195540377\\
-0.15519006345156	1.15547059030978\\
-0.15519006345156	-1.15547059030978\\
-0.020698754466626	1.9617678894998\\
-0.020698754466626	-1.9617678894998\\
-0.178778641674185	1.9552555840464\\
-0.178778641674185	-1.9552555840464\\
-0.0843506626803506	1.81921648044347\\
-0.0843506626803506	-1.81921648044347\\
-0.0756211689673647	1.11241504329724\\
-0.0756211689673647	-1.11241504329724\\
-0.0686152800562046	1.91335024078841\\
-0.0686152800562046	-1.91335024078841\\
-0.0703437882440776	1.18355117607353\\
-0.0703437882440776	-1.18355117607353\\
-0.0539230724701617	1.9978444655172\\
-0.0539230724701617	-1.9978444655172\\
-0.0648030021580633	1.96301945055705\\
-0.0648030021580633	-1.96301945055705\\
-0.0213814335977094	1.91762708515506\\
-0.0213814335977094	-1.91762708515506\\
-0.0288322094512398	1.19216756691624\\
-0.0288322094512398	-1.19216756691624\\
-0.0777997205141074	1.9813349402337\\
-0.0777997205141074	-1.9813349402337\\
-0.0841550269962749	1.2172658043966\\
-0.0841550269962749	-1.2172658043966\\
-0.194929966639035	1.94666166287455\\
-0.194929966639035	-1.94666166287455\\
-0.0380065598819983	1.21596856682758\\
-0.0380065598819983	-1.21596856682758\\
-0.00704024637614808	1.89210599204519\\
-0.00704024637614808	-1.89210599204519\\
-0.0766649020270817	1.93541374021795\\
-0.0766649020270817	-1.93541374021795\\
-0.0799523262592597	1.19525314601045\\
-0.0799523262592597	-1.19525314601045\\
-0.162717388247871	1.87067767539618\\
-0.162717388247871	-1.87067767539618\\
-0.137673438376411	1.22097219933303\\
-0.137673438376411	-1.22097219933303\\
-0.076137544131855	1.94879677502895\\
-0.076137544131855	-1.94879677502895\\
-0.112371721302939	1.99733818779539\\
-0.112371721302939	-1.99733818779539\\
-0.0724750383877167	1.89803720081372\\
-0.0724750383877167	-1.89803720081372\\
-0.0453901759111563	1.96248225175071\\
-0.0453901759111563	-1.96248225175071\\
-0.047847273461014	1.84650836061839\\
-0.047847273461014	-1.84650836061839\\
-0.031912042725642	1.14083823389283\\
-0.031912042725642	-1.14083823389283\\
-0.112999307301886	1.88756227029273\\
-0.112999307301886	-1.88756227029273\\
-0.113927711766811	1.16649949560837\\
-0.113927711766811	-1.16649949560837\\
-0.0984952349625966	1.82152507215034\\
-0.0984952349625966	-1.82152507215034\\
-0.0976843527083597	1.11411755665404\\
-0.0976843527083597	-1.11411755665404\\
-0.12899935813023	1.11618397349774\\
-0.12899935813023	-1.11618397349774\\
-0.0633039107139129	1.09121764920691\\
-0.0633039107139129	-1.09121764920691\\
-0.148949216035174	1.90594797620539\\
-0.148949216035174	-1.90594797620539\\
-0.147520883538049	1.18098542767126\\
-0.147520883538049	-1.18098542767126\\
-0.0896903458987437	1.22560019953663\\
-0.0896903458987437	-1.22560019953663\\
-0.0933694213591456	1.80889224343137\\
-0.0933694213591456	-1.80889224343137\\
-0.0889593509194892	1.10185903273462\\
-0.0889593509194892	-1.10185903273462\\
-0.0302184220512981	1.99912615108297\\
-0.0302184220512981	-1.99912615108297\\
-0.0541696413336512	1.80879020514672\\
-0.0541696413336512	-1.80879020514672\\
-0.0249723206000592	1.10605889609647\\
-0.0249723206000592	-1.10605889609647\\
-0.181780560342059	1.20008448499119\\
-0.181780560342059	-1.20008448499119\\
-0.181410012542209	1.92085953522625\\
-0.181410012542209	-1.92085953522625\\
-0.172522808858113	1.19439717600052\\
-0.172522808858113	-1.19439717600052\\
-0.0788748877068225	0\\
-0.102893797740884	1.20345386316619\\
-0.102893797740884	-1.20345386316619\\
-0.169589065088188	1.89908493447415\\
-0.169589065088188	-1.89908493447415\\
-0.168898099593133	1.18008310872918\\
-0.168898099593133	-1.18008310872918\\
-0.109432786113012	1.82646978977919\\
-0.109432786113012	-1.82646978977919\\
-0.114077709035495	1.11887391455328\\
-0.114077709035495	-1.11887391455328\\
-0.145552350121845	1.83916509214482\\
-0.145552350121845	-1.83916509214482\\
-0.162226407539194	1.13370996439333\\
-0.162226407539194	-1.13370996439333\\
-0.0804216929908962	1.1722471425518\\
-0.0804216929908962	-1.1722471425518\\
-0.025695667593047	1.89012790549677\\
-0.025695667593047	-1.89012790549677\\
-0.0235529438823087	1.17560149518432\\
-0.0235529438823087	-1.17560149518432\\
-0.120316876371828	2.02091423544862\\
-0.120316876371828	-2.02091423544862\\
-0.074059914569378	1.21217399726612\\
-0.074059914569378	-1.21217399726612\\
-0.153034423394688	1.19606500443007\\
-0.153034423394688	-1.19606500443007\\
-0.0571300782938571	1.18210978931691\\
-0.0571300782938571	-1.18210978931691\\
-0.101670646835635	1.08284572310321\\
-0.101670646835635	-1.08284572310321\\
-0.114313290980338	1.86537840119229\\
-0.114313290980338	-1.86537840119229\\
-0.117158628574221	1.15104934923188\\
-0.117158628574221	-1.15104934923188\\
-0.161142906017554	2.00676905786487\\
-0.161142906017554	-2.00676905786487\\
-0.0300195736747013	1.92694019563479\\
-0.0300195736747013	-1.92694019563479\\
-0.0604883224785994	1.98633987494058\\
-0.0604883224785994	-1.98633987494058\\
-0.0770075348615554	1.91909944731833\\
-0.0770075348615554	-1.91909944731833\\
-0.0592945284741466	1.83553602480134\\
-0.0592945284741466	-1.83553602480134\\
-0.0430650722792331	1.12987565699091\\
-0.0430650722792331	-1.12987565699091\\
-0.00375389721973945	1.90506965090487\\
-0.00375389721973945	-1.90506965090487\\
-0.105633787364639	1.79315200220793\\
-0.105633787364639	-1.79315200220793\\
-0.131999682436039	1.8665038361661\\
-0.131999682436039	-1.8665038361661\\
-0.137947708125864	1.153182104851\\
-0.137947708125864	-1.153182104851\\
-0.0153661518130182	1.14377025325043\\
-0.0153661518130182	-1.14377025325043\\
-0.0417881072172377	1.14461704310919\\
-0.0417881072172377	-1.14461704310919\\
-0.0809498733274354	1.96685697553036\\
-0.0809498733274354	-1.96685697553036\\
-0.0631089127274603	1.81436121685063\\
-0.0631089127274603	-1.81436121685063\\
-0.0413137031371845	1.11000855550909\\
-0.0413137031371845	-1.11000855550909\\
-0.0117325092430672	1.96702011197655\\
-0.0117325092430672	-1.96702011197655\\
-0.106080642166882	1.8520698470382\\
-0.106080642166882	-1.8520698470382\\
-0.107837623040304	1.14058049289982\\
-0.107837623040304	-1.14058049289982\\
-0.00351718350628766	1.9352998959968\\
-0.00351718350628766	-1.9352998959968\\
-0.156550811737832	1.93364576025202\\
-0.156550811737832	-1.93364576025202\\
-0.152173812944672	1.08540730207955\\
-0.152173812944672	-1.08540730207955\\
-0.135191911556471	1.79681181195044\\
-0.135191911556471	-1.79681181195044\\
-0.0409402050343531	1.0915657294956\\
-0.0409402050343531	-1.0915657294956\\
-0.122816876782201	1.96351724025364\\
-0.122816876782201	-1.96351724025364\\
-0.0849220357041893	1.88873978477237\\
-0.0849220357041893	-1.88873978477237\\
-0.0944691859022065	1.98235117702845\\
-0.0944691859022065	-1.98235117702845\\
-0.144398470369634	1.9279533869731\\
-0.144398470369634	-1.9279533869731\\
-0.139148306262478	1.1928647703634\\
-0.139148306262478	-1.1928647703634\\
-0.181087059999799	1.14574947916736\\
-0.181087059999799	-1.14574947916736\\
-0.0995083228228546	1.99614496141359\\
-0.0995083228228546	-1.99614496141359\\
-0.123871555990041	1.08742235526445\\
-0.123871555990041	-1.08742235526445\\
-0.00911171339392175	1.18830801666434\\
-0.00911171339392175	-1.18830801666434\\
-0.15370494577072	1.95483179402306\\
-0.15370494577072	-1.95483179402306\\
-0.10206490961528	1.97023130348182\\
-0.10206490961528	-1.97023130348182\\
-0.0361524276190953	1.83983642955036\\
-0.0361524276190953	-1.83983642955036\\
-0.132206969998352	1.83197177995292\\
-0.132206969998352	-1.83197177995292\\
-0.146194885864977	1.125623107206\\
-0.146194885864977	-1.125623107206\\
-0.0162998266404491	1.97865137047311\\
-0.0162998266404491	-1.97865137047311\\
-0.123475455095862	1.91033864220018\\
-0.123475455095862	-1.91033864220018\\
-0.0149022818101835	1.1124706107731\\
-0.0149022818101835	-1.1124706107731\\
-0.182858826277138	1.97514402133828\\
-0.182858826277138	-1.97514402133828\\
-0.155436272856391	1.09869264988222\\
-0.155436272856391	-1.09869264988222\\
-0.153774804373391	1.92027484926417\\
-0.153774804373391	-1.92027484926417\\
-0.0478338756048641	1.94113069912385\\
-0.0478338756048641	-1.94113069912385\\
-0.169585599729549	1.97734273542367\\
-0.169585599729549	-1.97734273542367\\
-0.138256359843001	1.98980612152703\\
-0.138256359843001	-1.98980612152703\\
-0.126357915279597	1.22170172576804\\
-0.126357915279597	-1.22170172576804\\
-0.13332021379277	1.12706905103367\\
-0.13332021379277	-1.12706905103367\\
-0.0694605281912311	2.02807866173352\\
-0.0694605281912311	-2.02807866173352\\
-0.193717711274449	1.89266628428452\\
-0.193717711274449	-1.89266628428452\\
-0.109198030666645	2.0197508005706\\
-0.109198030666645	-2.0197508005706\\
-0.141083660587484	1.94044079330902\\
-0.141083660587484	-1.94044079330902\\
-0.042230286916667	1.99428884080191\\
-0.042230286916667	-1.99428884080191\\
-0.0331880354635265	1.90943894765982\\
-0.0331880354635265	-1.90943894765982\\
-0.0228932849094873	1.90204542901324\\
-0.0228932849094873	-1.90204542901324\\
-0.057989410797869	1.91894563299802\\
-0.057989410797869	-1.91894563299802\\
-0.158753609435933	1.90266146990263\\
-0.158753609435933	-1.90266146990263\\
-0.0421915504705513	1.81867197529748\\
-0.0421915504705513	-1.81867197529748\\
-0.0898260686836011	1.16690459316395\\
-0.0898260686836011	-1.16690459316395\\
-0.179451842201271	1.89546664860934\\
-0.179451842201271	-1.89546664860934\\
-0.0601341155731927	1.8001712733664\\
-0.0601341155731927	-1.8001712733664\\
-0.0307126140593732	1.09608253421692\\
-0.0307126140593732	-1.09608253421692\\
-0.174256884736992	1.11367051807597\\
-0.174256884736992	-1.11367051807597\\
-0.0421885346363172	1.94973168903704\\
-0.0421885346363172	-1.94973168903704\\
-0.0537470021751563	1.20540093861008\\
-0.0537470021751563	-1.20540093861008\\
-0.115626037197525	1.83753940088357\\
-0.115626037197525	-1.83753940088357\\
-0.121825447759862	1.12896619492967\\
-0.121825447759862	-1.12896619492967\\
-0.142172495072296	1.81993587585063\\
-0.142172495072296	-1.81993587585063\\
-0.041171641855489	1.92147188169978\\
-0.041171641855489	-1.92147188169978\\
-0.0469921898617003	1.19107018416164\\
-0.0469921898617003	-1.19107018416164\\
-0.0203853645011128	1.12611460058355\\
-0.0203853645011128	-1.12611460058355\\
-0.147047284346179	1.97781102525556\\
-0.147047284346179	-1.97781102525556\\
-0.14969806302656	-1.97590122506075\\
-0.0312954905990681	1.11559871534292\\
-0.0312954905990681	-1.11559871534292\\
-0.11319483802092	1.07647753572544\\
-0.11319483802092	-1.07647753572544\\
-0.0504799903609755	1.11830981629827\\
-0.0504799903609755	-1.11830981629827\\
-0.0904393150561018	1.87288540968729\\
-0.0904393150561018	-1.87288540968729\\
-0.0889692786226363	1.15631469442994\\
-0.0889692786226363	-1.15631469442994\\
-0.0849411838906039	1.23671025929657\\
-0.0849411838906039	-1.23671025929657\\
-0.128858130951727	1.88550752545209\\
-0.128858130951727	-1.88550752545209\\
-0.070343194723421	1.98967550794925\\
-0.070343194723421	-1.98967550794925\\
-0.0480449809796352	1.83239882256036\\
-0.0480449809796352	-1.83239882256036\\
-0.179098527505454	1.85089562774339\\
-0.179098527505454	-1.85089562774339\\
-0.114662377004934	1.97642438544895\\
-0.114662377004934	-1.97642438544895\\
-0.0527268188966216	1.85837533950687\\
-0.0527268188966216	-1.85837533950687\\
-0.066600286532554	1.85956094167799\\
-0.066600286532554	-1.85956094167799\\
-0.167942644482084	1.85246109494838\\
-0.167942644482084	-1.85246109494838\\
-0.0643333793030332	1.15588382830756\\
-0.0643333793030332	-1.15588382830756\\
-0.1753340910601	1.90779993761942\\
-0.1753340910601	-1.90779993761942\\
-0.11438747872171	1.93191968615033\\
-0.11438747872171	-1.93191968615033\\
-0.11256598406868	1.19301667896956\\
-0.11256598406868	-1.19301667896956\\
-0.0603319182449694	2.02303126720576\\
-0.0603319182449694	-2.02303126720576\\
-0.0656672805549411	1.92960616847399\\
-0.0656672805549411	-1.92960616847399\\
-0.0868638477763568	1.08197273886441\\
-0.0868638477763568	-1.08197273886441\\
-0.149243610651067	2.0062760519138\\
-0.149243610651067	-2.0062760519138\\
-0.0618588540443397	1.8932959185632\\
-0.0618588540443397	-1.8932959185632\\
-0.0349918868280075	1.9798114362111\\
-0.0349918868280075	-1.9798114362111\\
-0.0371656095861654	1.87716212180114\\
-0.0371656095861654	-1.87716212180114\\
-0.115958262248804	1.21559561132263\\
-0.115958262248804	-1.21559561132263\\
-0.0564319699959403	1.14142415078046\\
-0.0564319699959403	-1.14142415078046\\
-0.123639229389313	1.84485213775968\\
-0.123639229389313	-1.84485213775968\\
-0.17493763289343	1.94020141141629\\
-0.17493763289343	-1.94020141141629\\
-0.0587198241966342	1.82384862383587\\
-0.0587198241966342	-1.82384862383587\\
-0.0931119534933847	1.13996017172167\\
-0.0931119534933847	-1.13996017172167\\
-0.160719957197594	1.94420119770246\\
-0.160719957197594	-1.94420119770246\\
-0.125353038594366	1.16418354538885\\
-0.125353038594366	-1.16418354538885\\
-0.145858752254796	1.14528975340933\\
-0.145858752254796	-1.14528975340933\\
-0.172951733613977	1.87379706668905\\
-0.172951733613977	-1.87379706668905\\
-0.0559344580602881	1.97618418454041\\
-0.0559344580602881	-1.97618418454041\\
-0.077431109669739	2.00819183728981\\
-0.077431109669739	-2.00819183728981\\
-0.100492982612235	1.86531539290835\\
-0.100492982612235	-1.86531539290835\\
-0.100591914094731	1.150674532771\\
-0.100591914094731	-1.150674532771\\
-0.0869372668751298	1.9264738004204\\
-0.0869372668751298	-1.9264738004204\\
-0.196118211216488	1.93076180820564\\
-0.196118211216488	-1.93076180820564\\
-0.123382016617941	1.98366797115817\\
-0.123382016617941	-1.98366797115817\\
-0.139038700206331	1.96701449337001\\
-0.139038700206331	-1.96701449337001\\
-0.142294925575816	1.13558525388777\\
-0.142294925575816	-1.13558525388777\\
-0.0306656762860772	1.12659697335201\\
-0.0306656762860772	-1.12659697335201\\
-0.159721672952518	1.97125537488896\\
-0.159721672952518	-1.97125537488896\\
-0.155932859149347	1.12100478669819\\
-0.155932859149347	-1.12100478669819\\
-0.101655094708645	1.93306860560634\\
-0.101655094708645	-1.93306860560634\\
-0.112373761514158	1.96211676903704\\
-0.112373761514158	-1.96211676903704\\
-0.145165676088217	1.10106943550436\\
-0.145165676088217	-1.10106943550436\\
-0.175938761324663	1.96696722833546\\
-0.175938761324663	-1.96696722833546\\
-0.101406453906325	1.16091405758045\\
-0.101406453906325	-1.16091405758045\\
-0.101260960457751	1.87964743108975\\
-0.101260960457751	-1.87964743108975\\
-0.0771248639245843	1.86343678545334\\
-0.0771248639245843	-1.86343678545334\\
-0.0550616807706414	1.90467017570665\\
-0.0550616807706414	-1.90467017570665\\
-0.102443904841014	1.91813745244061\\
-0.102443904841014	-1.91813745244061\\
-0.0863425662623837	1.97582347249269\\
-0.0863425662623837	-1.97582347249269\\
-0.0588315588395563	1.94064265079757\\
-0.0588315588395563	-1.94064265079757\\
-0.100015140988473	1.10148116049428\\
-0.100015140988473	-1.10148116049428\\
-0.141422291049219	1.91413065026976\\
-0.141422291049219	-1.91413065026976\\
-0.108207492269905	1.12983194789128\\
-0.108207492269905	-1.12983194789128\\
-0.0612122222928037	1.10411354410761\\
-0.0612122222928037	-1.10411354410761\\
-0.181308772821923	1.12872265709583\\
-0.181308772821923	-1.12872265709583\\
-0.0142816220486136	1.92543123944269\\
-0.0142816220486136	-1.92543123944269\\
-0.0816038697103623	1.80939254652608\\
-0.0816038697103623	-1.80939254652608\\
-0.15793540671627	1.83795921676075\\
-0.15793540671627	-1.83795921676075\\
-0.0588142398463194	1.84680231224919\\
-0.0588142398463194	-1.84680231224919\\
-0.0763294272726589	1.12944867177402\\
-0.0763294272726589	-1.12944867177402\\
-0.152431041948084	1.16778237074972\\
-0.152431041948084	-1.16778237074972\\
-0.176164961923067	1.16223718522927\\
-0.176164961923067	-1.16223718522927\\
-0.150835264958685	1.89594962704173\\
-0.150835264958685	-1.89594962704173\\
-0.0197822959662879	1.85809250547282\\
-0.0197822959662879	-1.85809250547282\\
-0.103952763346826	1.89250983566478\\
-0.103952763346826	-1.89250983566478\\
-0.0144347038581468	1.16670077030724\\
-0.0144347038581468	-1.16670077030724\\
-0.0559094249772984	1.1282520841488\\
-0.0559094249772984	-1.1282520841488\\
-0.0959774146265857	2.00779573159966\\
-0.0959774146265857	-2.00779573159966\\
-0.0855087908442973	1.11885681756069\\
-0.0855087908442973	-1.11885681756069\\
-0.0693416601065742	2.00150408836801\\
-0.0693416601065742	-2.00150408836801\\
-0.00555668429500003	1.94538085827338\\
-0.00555668429500003	-1.94538085827338\\
-0.100921772520104	2.03359753998067\\
-0.100921772520104	-2.03359753998067\\
-0.11981267697722	1.17547518544578\\
-0.11981267697722	-1.17547518544578\\
-0.10061011905123	1.94450464157208\\
-0.10061011905123	-1.94450464157208\\
-0.0696142591263388	1.1462322460844\\
-0.0696142591263388	-1.1462322460844\\
-0.121871110405863	1.82926422582997\\
-0.121871110405863	-1.82926422582997\\
-0.041421516982432	1.16395724467161\\
-0.041421516982432	-1.16395724467161\\
-0.100341454179518	1.18842870277502\\
-0.100341454179518	-1.18842870277502\\
-0.135195382559726	1.84174056403894\\
-0.135195382559726	-1.84174056403894\\
-0.113878754702985	1.92031593824554\\
-0.113878754702985	-1.92031593824554\\
-0.0343750947858687	1.95641617143328\\
-0.0343750947858687	-1.95641617143328\\
-0.0957119567190495	1.85145700854962\\
-0.0957119567190495	-1.85145700854962\\
-0.132441324482589	1.85601765012084\\
-0.132441324482589	-1.85601765012084\\
-0.103914082953511	1.95551477760076\\
-0.103914082953511	-1.95551477760076\\
-0.133371723679676	1.87646960372416\\
-0.133371723679676	-1.87646960372416\\
-0.0822413057673089	1.99305261384886\\
-0.0822413057673089	-1.99305261384886\\
-0.118796345818275	1.8768299274362\\
-0.118796345818275	-1.8768299274362\\
-0.115536525351167	1.9007260177303\\
-0.115536525351167	-1.9007260177303\\
-0.0658903348384987	1.87574384489112\\
-0.0658903348384987	-1.87574384489112\\
-0.0682711582182197	-1.87519785609926\\
-0.118087062853752	1.14074153120111\\
-0.118087062853752	-1.14074153120111\\
-0.0884257547270362	1.78850612249539\\
-0.0884257547270362	-1.78850612249539\\
-0.0581810219361254	1.17123805094961\\
-0.0581810219361254	-1.17123805094961\\
-0.131578433580393	1.97762279683511\\
-0.131578433580393	-1.97762279683511\\
-0.127515435529908	1.21167669831998\\
-0.127515435529908	-1.21167669831998\\
-0.110146469088536	1.22762672774892\\
-0.110146469088536	-1.22762672774892\\
-0.132916475707119	1.95538652382216\\
-0.132916475707119	-1.95538652382216\\
-0.123437664189734	1.9963157534183\\
-0.123437664189734	-1.9963157534183\\
-0.102532579612426	1.83600927289871\\
-0.102532579612426	-1.83600927289871\\
-0.135322740937407	1.16316062040144\\
-0.135322740937407	-1.16316062040144\\
-0.179654665132544	1.18309977498979\\
-0.179654665132544	-1.18309977498979\\
-0.150740337304488	1.96455910256247\\
-0.150740337304488	-1.96455910256247\\
-0.0493600844374954	1.08486396225642\\
-0.0493600844374954	-1.08486396225642\\
-0.0818945267392759	1.83654118346588\\
-0.0818945267392759	-1.83654118346588\\
-0.0909318780602822	1.8289969909403\\
-0.0909318780602822	-1.8289969909403\\
-0.135020243127873	1.07933578480625\\
-0.135020243127873	-1.07933578480625\\
-0.0514452004791365	1.10701415462682\\
-0.0514452004791365	-1.10701415462682\\
-0.145940593791377	1.11401238851586\\
-0.145940593791377	-1.11401238851586\\
-0.184946983029951	1.91145887561261\\
-0.184946983029951	-1.91145887561261\\
-0.0662476014562721	1.97638894278676\\
-0.0662476014562721	-1.97638894278676\\
-0.0714003364170722	1.10025380377264\\
-0.0714003364170722	-1.10025380377264\\
-0.030369097190209	1.16817404450234\\
-0.030369097190209	-1.16817404450234\\
-0.0731629875198979	1.81773331938988\\
-0.0731629875198979	-1.81773331938988\\
-0.0438829006657345	1.90471637908913\\
-0.0438829006657345	-1.90471637908913\\
-0.0697366889970133	1.19497369754003\\
-0.0697366889970133	-1.19497369754003\\
-0.0741513104778263	1.95894283657995\\
-0.0741513104778263	-1.95894283657995\\
-0.042349125213777	1.9727715635224\\
};

\end{axis}
\end{tikzpicture}%

%% file: figure/1_agent_ps_c.tex
%
%
\begin{tikzpicture}

\begin{axis}[%
width=4.521in,
height=3.566in,
at={(0.758in,0.481in)},
scale only axis,
xmin=-1,
xmax=0,
ymin=-2.5,
ymax=2.5,
xtick={-1,-0.5,0},
ytick={-2.5,0,2.5},
axis background/.style={fill=white},
legend style={legend cell align=left, align=left, draw=white!15!black}
]
\addplot [color=black, draw=none, mark=asterisk, mark options={solid, black}]
  table[row sep=crcr]{%
-0.155029061951074	0\\
-0.100000000000002	0\\
-0.155059798366778	1.17095774443225\\
-0.155059798366778	-1.17095774443225\\
-0.0999999999999949	1.17557050458495\\
-0.0999999999999949	-1.17557050458495\\
-0.15242259958997	1.88883004740328\\
-0.15242259958997	-1.88883004740328\\
-0.0999999999999998	1.90211303259031\\
-0.0999999999999998	-1.90211303259031\\
};

\addplot+[color=blue, draw=none, mark=*, mark size=1pt,mark options={solid, blue}]
  table[row sep=crcr]{%
-0.1	0\\
-0.1	1.17557050458495\\
-0.1	-1.17557050458495\\
-0.1	1.90211303259031\\
-0.1	-1.90211303259031\\
-0.1	1.90211303259031\\
-0.1	-1.90211303259031\\
-0.0999999999999999	1.17557050458495\\
-0.0999999999999999	-1.17557050458495\\
-0.1	0\\
-0.111678183905619	1.74458071023161\\
-0.111678183905619	-1.74458071023161\\
-0.139491814150965	0\\
-0.131263436111302	1.02081642987705\\
-0.131263436111302	-1.02081642987705\\
-0.0849092769303921	1.82381238829845\\
-0.0849092769303921	-1.82381238829845\\
-0.0759033938445468	0\\
-0.0771390261473335	1.11669878982002\\
-0.0771390261473335	-1.11669878982002\\
-0.0716079415986299	2.27552678431841\\
-0.0716079415986299	-2.27552678431841\\
-0.0913721774155055	1.29840255513789\\
-0.0913721774155055	-1.29840255513789\\
-0.135286444067743	1.67402446894499\\
-0.135286444067743	-1.67402446894499\\
-0.278285672440315	0.826246483196807\\
-0.278285672440315	-0.826246483196807\\
-0.525829458442205	0\\
-0.191437745897542	1.86949353798293\\
-0.191437745897542	-1.86949353798293\\
-0.201850714325772	0\\
-0.202337403517395	1.16665682459904\\
-0.202337403517395	-1.16665682459904\\
-0.201828451069315	2.0163841364775\\
-0.201828451069315	-2.0163841364775\\
-0.162217639825215	0\\
-0.162062692966641	1.23852437663316\\
-0.162062692966641	-1.23852437663316\\
-0.151716607149882	1.77138860434203\\
-0.151716607149882	-1.77138860434203\\
-0.223342986466138	0\\
-0.208438809971333	1.0646946730178\\
-0.208438809971333	-1.0646946730178\\
-0.0740805279114742	2.02039819120233\\
-0.0740805279114742	-2.02039819120233\\
-0.0839706749974107	1.23323325561514\\
-0.0839706749974107	-1.23323325561514\\
-0.096468859260293	1.66923509909658\\
-0.096468859260293	-1.66923509909658\\
-0.0781968457633201	0.852663786996411\\
-0.0781968457633201	-0.852663786996411\\
-0.050668589952774	0\\
-0.306392978898488	2.17520384841001\\
-0.306392978898488	-2.17520384841001\\
-0.174602097882871	1.29238389805303\\
-0.174602097882871	-1.29238389805303\\
-0.189993493052897	2.04736863362603\\
-0.189993493052897	-2.04736863362603\\
-0.0872968427124092	1.78998316989187\\
-0.0872968427124092	-1.78998316989187\\
-0.0760314035737708	1.08176278171321\\
-0.0760314035737708	-1.08176278171321\\
-0.0823543882158998	1.86264375837601\\
-0.0823543882158998	-1.86264375837601\\
-0.0785534303689452	1.14918816488863\\
-0.0785534303689452	-1.14918816488863\\
-0.0707407177634631	2.10775426203582\\
-0.0707407177634631	-2.10775426203582\\
-0.0863282376441112	1.26159530570933\\
-0.0863282376441112	-1.26159530570933\\
-0.213896701855497	2.26083940003636\\
-0.213896701855497	-2.26083940003636\\
-0.135204076533076	1.29888574847011\\
-0.135204076533076	-1.29888574847011\\
-0.161603837606241	1.92323087578093\\
-0.161603837606241	-1.92323087578093\\
-0.155002180210927	1.19240623650018\\
-0.155002180210927	-1.19240623650018\\
-0.335746162107177	1.9146545104156\\
-0.335746162107177	-1.9146545104156\\
-0.269979795891675	0\\
-0.276612915399524	1.24147641061656\\
-0.276612915399524	-1.24147641061656\\
-0.107992644737249	1.63819948626998\\
-0.107992644737249	-1.63819948626998\\
-0.550289431715043	0\\
-0.168020468980674	0.702698525252595\\
-0.168020468980674	-0.702698525252595\\
-0.0930997875809578	1.71337535422594\\
-0.0930997875809578	-1.71337535422594\\
-0.0751841815445402	0.964828493699764\\
-0.0751841815445402	-0.964828493699764\\
-0.0820113153274389	1.20697352904737\\
-0.0820113153274389	-1.20697352904737\\
-0.0769312187221491	1.95874819424868\\
-0.0769312187221491	-1.95874819424868\\
-0.184944857507661	1.70296248983811\\
-0.184944857507661	-1.70296248983811\\
-0.385763071880377	0.934523896261171\\
-0.385763071880377	-0.934523896261171\\
-0.192360837486433	1.75181648885311\\
-0.192360837486433	-1.75181648885311\\
-0.308736246815835	1.05213391642442\\
-0.308736246815835	-1.05213391642442\\
-0.337017772560014	0\\
-0.0911087299646308	1.73915676680441\\
-0.0911087299646308	-1.73915676680441\\
-0.0750213060999873	1.01202369596491\\
-0.0750213060999873	-1.01202369596491\\
-0.0881980573750247	0.9012495492316\\
-0.0881980573750247	-0.9012495492316\\
-0.272024245120411	2.18081717275207\\
-0.272024245120411	-2.18081717275207\\
-0.132804711401073	1.725576102489\\
-0.132804711401073	-1.725576102489\\
-0.202354639963229	0.987990516972137\\
-0.202354639963229	-0.987990516972137\\
-0.108683470868096	0.845848575440034\\
-0.108683470868096	-0.845848575440034\\
-0.357032104225587	1.99752525563131\\
-0.357032104225587	-1.99752525563131\\
-0.242598077656505	1.26826740806419\\
-0.242598077656505	-1.26826740806419\\
-0.0784855029125041	1.92881714477825\\
-0.0784855029125041	-1.92881714477825\\
-0.14770319630536	1.85648243585652\\
-0.14770319630536	-1.85648243585652\\
-0.159178514068669	1.14779513602583\\
-0.159178514068669	-1.14779513602583\\
-0.286759423768274	0.887344724846297\\
-0.286759423768274	-0.887344724846297\\
-0.441357093791989	0\\
-0.0757318884901364	0.929557817225646\\
-0.0757318884901364	-0.929557817225646\\
-0.068059540061884	2.19507795857073\\
-0.068059540061884	-2.19507795857073\\
-0.495599583994131	0\\
-0.402610233712067	0.970680454030118\\
-0.402610233712067	-0.970680454030118\\
-0.0912859219560882	1.03822981024968\\
-0.0912859219560882	-1.03822981024968\\
-0.141634367869047	1.80164299133302\\
-0.141634367869047	-1.80164299133302\\
-0.171576745452438	1.0979487512939\\
-0.171576745452438	-1.0979487512939\\
-0.0738706842505179	2.29853554667046\\
-0.0738706842505179	-2.29853554667046\\
-0.225494383159329	1.77517478299705\\
-0.225494383159329	-1.77517478299705\\
-0.324791024007162	1.1036750579823\\
-0.324791024007162	-1.1036750579823\\
-0.0715656311459583	2.08430262931876\\
-0.0715656311459583	-2.08430262931876\\
-0.213635091220073	2.06517131268555\\
-0.213635091220073	-2.06517131268555\\
-0.168342243654874	0.868217922413875\\
-0.168342243654874	-0.868217922413875\\
-0.131909162410191	1.88751827896636\\
-0.131909162410191	-1.88751827896636\\
-0.13403979056102	1.16767630706802\\
-0.13403979056102	-1.16767630706802\\
-0.0893699800291458	1.76198568041996\\
-0.0893699800291458	-1.76198568041996\\
-0.288244927653622	1.80746299860315\\
-0.288244927653622	-1.80746299860315\\
-0.342203066722646	1.17457832637416\\
-0.342203066722646	-1.17457832637416\\
-0.324783474108562	1.98503232014176\\
-0.324783474108562	-1.98503232014176\\
-0.365581983470191	0\\
-0.317898711668866	1.01749267317839\\
-0.317898711668866	-1.01749267317839\\
-0.0698613252544669	2.13437646465846\\
-0.0698613252544669	-2.13437646465846\\
-0.081042713076642	1.88393718022761\\
-0.081042713076642	-1.88393718022761\\
-0.126293768118989	0.765270678918229\\
-0.126293768118989	-0.765270678918229\\
-0.393051814077195	2.01099255625508\\
-0.393051814077195	-2.01099255625508\\
-0.0799217059508822	1.17453674091023\\
-0.0799217059508822	-1.17453674091023\\
-0.130693942483543	1.75302470108778\\
-0.130693942483543	-1.75302470108778\\
-0.175688715124375	1.03509892887917\\
-0.175688715124375	-1.03509892887917\\
-0.133074734761204	2.00066179276605\\
-0.133074734761204	-2.00066179276605\\
-0.121919595961067	1.2258655406606\\
-0.121919595961067	-1.2258655406606\\
-0.296830912529033	1.7734572254995\\
-0.296830912529033	-1.7734572254995\\
-0.3932564002752	1.16185149681933\\
-0.3932564002752	-1.16185149681933\\
-0.341245960773726	2.02597297785295\\
-0.341245960773726	-2.02597297785295\\
-0.361709974030074	2.0772470361359\\
-0.361709974030074	-2.0772470361359\\
-0.215572772209658	1.28442650984543\\
-0.215572772209658	-1.28442650984543\\
-0.233312779478365	0.778739694402471\\
-0.233312779478365	-0.778739694402471\\
-0.223857635493501	1.96860626567877\\
-0.223857635493501	-1.96860626567877\\
-0.188097739738149	1.22465780693306\\
-0.188097739738149	-1.22465780693306\\
-0.0865490507547314	0.726739633328064\\
-0.0865490507547314	-0.726739633328064\\
-0.0293779396436537	0\\
-0.249978646237934	1.95713243799764\\
-0.249978646237934	-1.95713243799764\\
-0.208569252778678	1.22637189778109\\
-0.208569252778678	-1.22637189778109\\
-0.104867272877212	2.09294798932533\\
-0.104867272877212	-2.09294798932533\\
-0.243213417521018	2.12153624015903\\
-0.243213417521018	-2.12153624015903\\
-0.16179830856607	1.27374396865178\\
-0.16179830856607	-1.27374396865178\\
-0.162348342078883	1.8797513528669\\
-0.162348342078883	-1.8797513528669\\
-0.168021466128691	1.1667115356467\\
-0.168021466128691	-1.1667115356467\\
-0.39405878453207	1.98001556978677\\
-0.39405878453207	-1.98001556978677\\
-0.262721120623525	1.27654601085747\\
-0.262721120623525	-1.27654601085747\\
-0.256156913015394	1.74798731743708\\
-0.256156913015394	-1.74798731743708\\
-0.404885369817407	1.09942456796111\\
-0.404885369817407	-1.09942456796111\\
-0.406442058875425	0\\
-0.16897486269695	1.07663420160309\\
-0.16897486269695	-1.07663420160309\\
-0.0673627251070747	2.2211638279804\\
-0.0673627251070747	-2.2211638279804\\
-0.0754522447389843	1.05451827213353\\
-0.0754522447389843	-1.05451827213353\\
-0.114815110217059	1.68217831891853\\
-0.114815110217059	-1.68217831891853\\
-0.245856666533361	0\\
-0.167756161334561	1.72762664710831\\
-0.167756161334561	-1.72762664710831\\
-0.295120090972606	0.9964851009018\\
-0.295120090972606	-0.9964851009018\\
-0.26266380418529	1.03338020374064\\
-0.26266380418529	-1.03338020374064\\
-0.294972806879908	0\\
-0.0728331890982751	2.05078762986751\\
-0.0728331890982751	-2.05078762986751\\
-0.2461589473327	1.9913378676822\\
-0.2461589473327	-1.9913378676822\\
-0.211536630357558	2.04274457805182\\
-0.211536630357558	-2.04274457805182\\
-0.200706951643613	1.08820100597621\\
-0.200706951643613	-1.08820100597621\\
-0.270424112191291	1.77177760933143\\
-0.270424112191291	-1.77177760933143\\
-0.377034103202389	1.1363774316438\\
-0.377034103202389	-1.1363774316438\\
-0.175671589136687	1.8972955503607\\
-0.175671589136687	-1.8972955503607\\
-0.154175856825635	2.02312353452401\\
-0.154175856825635	-2.02312353452401\\
-0.0690435546984007	2.16081779365625\\
-0.0690435546984007	-2.16081779365625\\
-0.309831278558318	2.04132564442896\\
-0.309831278558318	-2.04132564442896\\
-0.208084936911808	1.26586919568261\\
-0.208084936911808	-1.26586919568261\\
-0.296625839450397	2.14127186925696\\
-0.296625839450397	-2.14127186925696\\
-0.11171679106259	1.76994484907932\\
-0.11171679106259	-1.76994484907932\\
-0.125566653474936	1.0572352991444\\
-0.125566653474936	-1.0572352991444\\
-0.137931777554677	1.94612601828605\\
-0.137931777554677	-1.94612601828605\\
-0.131035134150781	1.20170459299559\\
-0.131035134150781	-1.20170459299559\\
-0.366707326007639	1.87552050413326\\
-0.366707326007639	-1.87552050413326\\
-0.315973215343149	1.2463442913812\\
-0.315973215343149	-1.2463442913812\\
-0.321357179828006	2.10258172179287\\
-0.321357179828006	-2.10258172179287\\
-0.195053623135947	1.28124378062317\\
-0.195053623135947	-1.28124378062317\\
-0.133574803323356	2.0426387769824\\
-0.133574803323356	-2.0426387769824\\
-0.094719187402346	1.69237450343543\\
-0.094719187402346	-1.69237450343543\\
-0.380763596310017	1.84819155936775\\
-0.380763596310017	-1.84819155936775\\
-0.342365074967777	1.24888252588477\\
-0.342365074967777	-1.24888252588477\\
-0.227729078487698	1.05360113402157\\
-0.227729078487698	-1.05360113402157\\
-0.271759064647772	0.983862513846686\\
-0.271759064647772	-0.983862513846686\\
-0.223932164319589	0.953337981865536\\
-0.223932164319589	-0.953337981865536\\
-0.127944636450606	1.84828864981131\\
-0.127944636450606	-1.84828864981131\\
-0.136596959915374	1.13885229623959\\
-0.136596959915374	-1.13885229623959\\
-0.0903173864155384	2.06837937769971\\
-0.0903173864155384	-2.06837937769971\\
-0.145307458966314	2.20080777783562\\
-0.145307458966314	-2.20080777783562\\
-0.116410863293301	1.28453090456335\\
-0.116410863293301	-1.28453090456335\\
-0.107426020061549	2.12151282489622\\
-0.107426020061549	-2.12151282489622\\
-0.110952085746596	2.00439663598211\\
-0.110952085746596	-2.00439663598211\\
-0.0813837190396107	0.793319182762551\\
-0.0813837190396107	-0.793319182762551\\
-0.125373784234484	2.06671617761261\\
-0.125373784234484	-2.06671617761261\\
-0.113447688650586	1.24925494775439\\
-0.113447688650586	-1.24925494775439\\
-0.150744457079723	2.10410877112326\\
-0.150744457079723	-2.10410877112326\\
-0.202766274489769	1.73177490248322\\
-0.202766274489769	-1.73177490248322\\
-0.362610107872317	1.02345623444692\\
-0.362610107872317	-1.02345623444692\\
-0.131236694686391	2.26005548850867\\
-0.131236694686391	-2.26005548850867\\
-0.134628236880553	1.69985908085064\\
-0.134628236880553	-1.69985908085064\\
-0.238313438597041	0.926602320921497\\
-0.238313438597041	-0.926602320921497\\
-0.152337966390495	0.788753603877727\\
-0.152337966390495	-0.788753603877727\\
-0.230404246607697	1.81241585242395\\
-0.230404246607697	-1.81241585242395\\
-0.285879056824534	1.14105288344196\\
-0.285879056824534	-1.14105288344196\\
-0.187675444658861	1.97192372070292\\
-0.187675444658861	-1.97192372070292\\
-0.22179496810062	1.84348149627657\\
-0.22179496810062	-1.84348149627657\\
-0.249444590894239	1.15975468819542\\
-0.249444590894239	-1.15975468819542\\
-0.102427228470201	2.16297323001471\\
-0.102427228470201	-2.16297323001471\\
-0.122307009677377	0.993433604966549\\
-0.122307009677377	-0.993433604966549\\
-0.157149675649321	1.82062245694072\\
-0.157149675649321	-1.82062245694072\\
-0.186189636559678	1.11980232254508\\
-0.186189636559678	-1.11980232254508\\
-0.273381215257438	1.86145420376485\\
-0.273381215257438	-1.86145420376485\\
-0.279094772355424	1.19382027885013\\
-0.279094772355424	-1.19382027885013\\
-0.187843339712073	1.7857042199327\\
-0.187843339712073	-1.7857042199327\\
-0.25857176978741	1.0944377951199\\
-0.25857176978741	-1.0944377951199\\
-0.394169968617352	1.11898595903577\\
-0.394169968617352	-1.11898595903577\\
-0.139622717291327	0.688135392084062\\
-0.139622717291327	-0.688135392084062\\
-0.13888377900304	1.96697064776162\\
-0.13888377900304	-1.96697064776162\\
-0.0772143350395	0.877256006479905\\
-0.0772143350395	-0.877256006479905\\
-0.166542885546645	1.84303628078866\\
-0.166542885546645	-1.84303628078866\\
-0.18754084839449	1.14112084841331\\
-0.18754084839449	-1.14112084841331\\
-0.168309353895555	0.846021852727697\\
-0.168309353895555	-0.846021852727697\\
-0.336203256086233	1.04204604425391\\
-0.336203256086233	-1.04204604425391\\
-0.101504091665942	1.93775211526831\\
-0.101504091665942	-1.93775211526831\\
-0.101282286461946	1.19589605874986\\
-0.101282286461946	-1.19589605874986\\
-0.119414328864573	0.740686346715977\\
-0.119414328864573	-0.740686346715977\\
-0.196982782390048	2.12873622098403\\
-0.196982782390048	-2.12873622098403\\
-0.141933655489813	1.27086403728032\\
-0.141933655489813	-1.27086403728032\\
-0.114560522120458	1.87133879735087\\
-0.114560522120458	-1.87133879735087\\
-0.116921492370399	1.15539671667295\\
-0.116921492370399	-1.15539671667295\\
-0.251946352611165	1.87563967010827\\
-0.251946352611165	-1.87563967010827\\
-0.252733855408757	1.19132547285034\\
-0.252733855408757	-1.19132547285034\\
-0.15967320963079	1.69413671091375\\
-0.15967320963079	-1.69413671091375\\
-0.333192349622344	0.896975653551976\\
-0.333192349622344	-0.896975653551976\\
-0.610795547857568	0\\
-0.39605921550094	0.879939522858822\\
-0.39605921550094	-0.879939522858822\\
-0.114589153159085	1.82250143146079\\
-0.114589153159085	-1.82250143146079\\
-0.122280361491449	1.11544949425692\\
-0.122280361491449	-1.11544949425692\\
-0.26091364240029	2.13240050927456\\
-0.26091364240029	-2.13240050927456\\
-0.0798299435353436	1.90461450557752\\
-0.0798299435353436	-1.90461450557752\\
-0.154248452152059	1.11504168263458\\
-0.154248452152059	-1.11504168263458\\
-0.262693348490648	2.22754008942351\\
-0.262693348490648	-2.22754008942351\\
-0.343040476582284	0.878321020530615\\
-0.343040476582284	-0.878321020530615\\
-0.293422436606107	1.9446978429367\\
-0.293422436606107	-1.9446978429367\\
-0.239014964075013	1.23477007136584\\
-0.239014964075013	-1.23477007136584\\
-0.28248362391626	2.01690048552356\\
-0.28248362391626	-2.01690048552356\\
-0.162667690048424	2.14493937756837\\
-0.162667690048424	-2.14493937756837\\
-0.17770550230208	0.961965336684817\\
-0.17770550230208	-0.961965336684817\\
-0.263182629818492	0.761217146441379\\
-0.263182629818492	-0.761217146441379\\
-0.673588011283716	0\\
-0.360884535254252	1.16268838702539\\
-0.360884535254252	-1.16268838702539\\
-0.105844409601414	1.01992257782365\\
-0.105844409601414	-1.01992257782365\\
-0.187853637090068	0.85731977831588\\
-0.187853637090068	-0.85731977831588\\
-0.186106772258875	0.93370805697247\\
-0.186106772258875	-0.93370805697247\\
-0.356294083196564	1.80847949612938\\
-0.356294083196564	-1.80847949612938\\
-0.373700570621676	1.22512163130685\\
-0.373700570621676	-1.22512163130685\\
-0.152221741230798	2.07436758234897\\
-0.152221741230798	-2.07436758234897\\
-0.197915824634835	1.90954340027064\\
-0.197915824634835	-1.90954340027064\\
-0.190416313366721	1.19189238550739\\
-0.190416313366721	-1.19189238550739\\
-0.336716556462468	1.86321778649225\\
-0.336716556462468	-1.86321778649225\\
-0.284160601314895	1.88033803280364\\
-0.284160601314895	-1.88033803280364\\
-0.0758801979367725	1.98035895868092\\
-0.0758801979367725	-1.98035895868092\\
-0.199321350605827	2.20719461629059\\
-0.199321350605827	-2.20719461629059\\
-0.236751041358556	0.861374103563006\\
-0.236751041358556	-0.861374103563006\\
-0.386056463507621	0\\
-0.165942503594096	2.17923735234849\\
-0.165942503594096	-2.17923735234849\\
-0.0668277567609332	2.24233989237235\\
-0.0668277567609332	-2.24233989237235\\
-0.196611644707685	1.80927478567877\\
-0.196611644707685	-1.80927478567877\\
-0.248621035154886	1.12205427533874\\
-0.248621035154886	-1.12205427533874\\
-0.15639231627809	0.898377394831495\\
-0.15639231627809	-0.898377394831495\\
-0.256663327916298	2.02292261919315\\
-0.256663327916298	-2.02292261919315\\
-0.189780813547379	1.24988733663836\\
-0.189780813547379	-1.24988733663836\\
-0.181403151833231	2.22462787734892\\
-0.181403151833231	-2.22462787734892\\
-0.245286307123979	1.72750460657026\\
-0.245286307123979	-1.72750460657026\\
-0.435793089037994	1.05870033062014\\
-0.435793089037994	-1.05870033062014\\
-0.120440123541784	1.6581622903903\\
-0.120440123541784	-1.6581622903903\\
-0.279363536015209	1.1194197124657\\
-0.279363536015209	-1.1194197124657\\
-0.173422297799623	1.76854197494273\\
-0.173422297799623	-1.76854197494273\\
-0.252485427592083	1.06756714407757\\
-0.252485427592083	-1.06756714407757\\
-0.163431121767788	0.993492327529839\\
-0.163431121767788	-0.993492327529839\\
-0.328812576477646	0.936311863567536\\
-0.328812576477646	-0.936311863567536\\
-0.176656945940591	2.27302268442231\\
-0.176656945940591	-2.27302268442231\\
-0.32558153781751	1.77477256100612\\
-0.32558153781751	-1.77477256100612\\
-0.40533841137814	1.18940792290302\\
-0.40533841137814	-1.18940792290302\\
-0.366568221897756	1.97643229995774\\
-0.366568221897756	-1.97643229995774\\
-0.391323083049661	1.02538868033638\\
-0.391323083049661	-1.02538868033638\\
-0.113832339817914	1.95390788512608\\
-0.113832339817914	-1.95390788512608\\
-0.0750230363236873	0.987660970715798\\
-0.0750230363236873	-0.987660970715798\\
-0.227660517560749	1.90069840418863\\
-0.227660517560749	-1.90069840418863\\
-0.218779773746464	1.1954030935515\\
-0.218779773746464	-1.1954030935515\\
-0.10158786312521	1.80713730754546\\
-0.10158786312521	-1.80713730754546\\
-0.102675008381565	1.09998324628327\\
-0.102675008381565	-1.09998324628327\\
-0.245866756941438	1.77677091944699\\
-0.245866756941438	-1.77677091944699\\
-0.34623841255517	1.12020803845528\\
-0.34623841255517	-1.12020803845528\\
-0.222020239316738	1.93861084025635\\
-0.222020239316738	-1.93861084025635\\
-0.217397191546346	2.15511897681685\\
-0.217397191546346	-2.15511897681685\\
-0.281235706413447	2.24180935377766\\
-0.281235706413447	-2.24180935377766\\
-0.244628658916305	2.08310305427866\\
-0.244628658916305	-2.08310305427866\\
-0.358293353153762	0.939474022590219\\
-0.358293353153762	-0.939474022590219\\
-0.148157452875907	1.74213249436266\\
-0.148157452875907	-1.74213249436266\\
-0.227311978368352	1.02007103613157\\
-0.227311978368352	-1.02007103613157\\
-0.258933625865455	0.947060775950031\\
-0.258933625865455	-0.947060775950031\\
-0.340750251827896	2.07264561682696\\
-0.340750251827896	-2.07264561682696\\
-0.147530183890384	1.90061296974325\\
-0.147530183890384	-1.90061296974325\\
-0.283858461214362	2.04380007434718\\
-0.283858461214362	-2.04380007434718\\
-0.186449143043859	2.152328992123\\
-0.186449143043859	-2.152328992123\\
-0.0836036138849719	0.761713498169921\\
-0.0836036138849719	-0.761713498169921\\
-0.26657491691699	2.00435118307992\\
-0.26657491691699	-2.00435118307992\\
-0.256062884098694	0.883959632894653\\
-0.256062884098694	-0.883959632894653\\
-0.357524345146827	2.04024295274314\\
-0.357524345146827	-2.04024295274314\\
-0.214880417622738	1.74946664879809\\
-0.214880417622738	-1.74946664879809\\
-0.35014490450625	1.06329635162638\\
-0.35014490450625	-1.06329635162638\\
-0.0818056250316276	2.00140899262997\\
-0.0818056250316276	-2.00140899262997\\
-0.130322263792465	1.91382377075886\\
-0.130322263792465	-1.91382377075886\\
-0.325487360107478	0.982000282665332\\
-0.325487360107478	-0.982000282665332\\
-0.177627325298054	0.750321066390484\\
-0.177627325298054	-0.750321066390484\\
-0.114388204145416	0.947580611382343\\
-0.114388204145416	-0.947580611382343\\
-0.247349314950147	2.14971306879991\\
-0.247349314950147	-2.14971306879991\\
-0.329034678426934	1.07206219433314\\
-0.329034678426934	-1.07206219433314\\
-0.307356208040083	0.955290079047773\\
-0.307356208040083	-0.955290079047773\\
-0.10099167691711	2.29630068098289\\
-0.10099167691711	-2.29630068098289\\
-0.238731210025656	2.05841462389479\\
-0.238731210025656	-2.05841462389479\\
-0.263635431826518	1.90820919384414\\
-0.263635431826518	-1.90820919384414\\
-0.240810044822935	1.21096719438517\\
-0.240810044822935	-1.21096719438517\\
-0.279413839384284	0.957054631147625\\
-0.279413839384284	-0.957054631147625\\
-0.405473365954381	1.05490472273727\\
-0.405473365954381	-1.05490472273727\\
-0.108146934744642	2.02980076928857\\
-0.108146934744642	-2.02980076928857\\
-0.3390400315464	1.79119616953407\\
-0.3390400315464	-1.79119616953407\\
-0.388341079593706	1.207021575919\\
-0.388341079593706	-1.207021575919\\
-0.215904892096183	1.87583178418196\\
-0.215904892096183	-1.87583178418196\\
-0.270007060198403	1.15972317064504\\
-0.270007060198403	-1.15972317064504\\
-0.139233005944391	2.14109399816465\\
-0.139233005944391	-2.14109399816465\\
-0.275751870820847	0.91253728598985\\
-0.275751870820847	-0.91253728598985\\
-0.244888451843835	0.974476155828496\\
-0.244888451843835	-0.974476155828496\\
-0.199768975688961	1.83709766383325\\
-0.199768975688961	-1.83709766383325\\
-0.230628943868368	1.14650184662125\\
-0.230628943868368	-1.14650184662125\\
-0.106145938867226	2.2549058420353\\
-0.106145938867226	-2.2549058420353\\
-0.418942530524591	1.93869400692442\\
-0.418942530524591	-1.93869400692442\\
-0.290566077770075	1.27955785348528\\
-0.290566077770075	-1.27955785348528\\
-0.147776295767284	1.07056933363097\\
-0.147776295767284	-1.07056933363097\\
-0.315082245657465	1.13196589398522\\
-0.315082245657465	-1.13196589398522\\
-0.218971213559869	2.09972119571942\\
-0.218971213559869	-2.09972119571942\\
-0.274444372484098	2.11336375786362\\
-0.274444372484098	-2.11336375786362\\
-0.370943568697075	1.91542603587012\\
-0.370943568697075	-1.91542603587012\\
-0.290005528651353	1.25648319488423\\
-0.290005528651353	-1.25648319488423\\
-0.23458757923109	0.828790283494081\\
-0.23458757923109	-0.828790283494081\\
-0.279030708302536	0.867711238762158\\
-0.279030708302536	-0.867711238762158\\
-0.288811533505976	1.84297415450997\\
-0.288811533505976	-1.84297415450997\\
-0.305523613518851	1.19307540535981\\
-0.305523613518851	-1.19307540535981\\
-0.382061924198093	1.08113545386451\\
-0.382061924198093	-1.08113545386451\\
-0.0983893454476921	2.19117801050171\\
-0.0983893454476921	-2.19117801050171\\
-0.358741335870916	1.8445116868501\\
-0.358741335870916	-1.8445116868501\\
-0.294142277428995	0.92413335678571\\
-0.294142277428995	-0.92413335678571\\
-0.277392984343509	1.07292356377372\\
-0.277392984343509	-1.07292356377372\\
-0.178138423232269	0.887365785408068\\
-0.178138423232269	-0.887365785408068\\
-0.0988819833964725	1.8453609103408\\
-0.0988819833964725	-1.8453609103408\\
-0.0985038651341472	1.13509559953262\\
-0.0985038651341472	-1.13509559953262\\
-0.193096402025079	0.808507699338665\\
-0.193096402025079	-0.808507699338665\\
-0.129291260336567	0.848826609604259\\
-0.129291260336567	-0.848826609604259\\
-0.183267241699636	2.1126397854747\\
-0.183267241699636	-2.1126397854747\\
-0.276568489909662	1.96547541857788\\
-0.276568489909662	-1.96547541857788\\
-0.0926182162587569	0.677827424179731\\
-0.0926182162587569	-0.677827424179731\\
-0.115512707415405	1.79107724245725\\
-0.115512707415405	-1.79107724245725\\
-0.12902648169791	1.08317774795915\\
-0.12902648169791	-1.08317774795915\\
-0.0920050817064915	2.04465388670202\\
-0.0920050817064915	-2.04465388670202\\
-0.0791804759352829	0.831811971828691\\
-0.0791804759352829	-0.831811971828691\\
-0.280601091617152	0.781535413794201\\
-0.280601091617152	-0.781535413794201\\
-0.646960751553667	0\\
-0.143239909923884	0.864299921561713\\
-0.143239909923884	-0.864299921561713\\
-0.264676476050209	0.810988156519428\\
-0.264676476050209	-0.810988156519428\\
-0.253303675394182	0.838692546906838\\
-0.253303675394182	-0.838692546906838\\
-0.162122629534859	2.00393244525225\\
-0.162122629534859	-2.00393244525225\\
-0.268193993180192	1.80305519337708\\
-0.268193993180192	-1.80305519337708\\
-0.227308105985747	2.18765020118972\\
-0.227308105985747	-2.18765020118972\\
-0.13788588304101	2.23391544623006\\
-0.13788588304101	-2.23391544623006\\
-0.203043588652797	1.98502399441512\\
-0.203043588652797	-1.98502399441512\\
-0.304373765308477	1.08072043102383\\
-0.304373765308477	-1.08072043102383\\
-0.260008451934746	1.82908535247658\\
-0.260008451934746	-1.82908535247658\\
-0.297270451604195	1.16920622053883\\
-0.297270451604195	-1.16920622053883\\
-0.217891337535873	0.974723942883815\\
-0.217891337535873	-0.974723942883815\\
-0.176184354160771	2.01967243466667\\
-0.176184354160771	-2.01967243466667\\
-0.247414058279463	2.19701940977888\\
-0.247414058279463	-2.19701940977888\\
-0.306134003323975	2.0100218502504\\
-0.306134003323975	-2.0100218502504\\
-0.584383202519102	0\\
-0.302647836324242	0.821620568962243\\
-0.302647836324242	-0.821620568962243\\
-0.326599247559817	2.19038975061997\\
-0.326599247559817	-2.19038975061997\\
-0.310526927081828	1.79257866042939\\
-0.310526927081828	-1.79257866042939\\
-0.374200576034154	1.18403442379885\\
-0.374200576034154	-1.18403442379885\\
-0.298969055055396	2.21139966228409\\
-0.298969055055396	-2.21139966228409\\
-0.160388770452104	1.79225890474891\\
-0.160388770452104	-1.79225890474891\\
-0.233818146401163	1.09725521744315\\
-0.233818146401163	-1.09725521744315\\
-0.274073031131817	2.06133108431151\\
-0.274073031131817	-2.06133108431151\\
-0.302627194025467	0.905711396323753\\
-0.302627194025467	-0.905711396323753\\
-0.217274834096819	0.730143959021871\\
-0.217274834096819	-0.730143959021871\\
-0.411188598371313	1.99811550690042\\
-0.411188598371313	-1.99811550690042\\
-0.301836080264334	1.98301003500806\\
-0.301836080264334	-1.98301003500806\\
-0.474315507606748	0\\
-0.396336022956798	1.91683827288657\\
-0.396336022956798	-1.91683827288657\\
-0.318444801734486	1.17622538928337\\
-0.318444801734486	-1.17622538928337\\
-0.297210565971034	0.841687481892111\\
-0.297210565971034	-0.841687481892111\\
-0.167841053562605	1.94811446843981\\
-0.167841053562605	-1.94811446843981\\
-0.277570973812203	2.15211351829585\\
-0.277570973812203	-2.15211351829585\\
-0.180035902790024	2.29284234457028\\
-0.180035902790024	-2.29284234457028\\
-0.201701070499681	0.888396478897046\\
-0.201701070499681	-0.888396478897046\\
-0.207284781309543	1.11482044570799\\
-0.207284781309543	-1.11482044570799\\
-0.147923712916984	0.755240379184231\\
-0.147923712916984	-0.755240379184231\\
-0.369767700388176	0.95891436445718\\
-0.369767700388176	-0.95891436445718\\
-0.299301869662906	1.10789427869975\\
-0.299301869662906	-1.10789427869975\\
-0.126468752760887	0.903649263565501\\
-0.126468752760887	-0.903649263565501\\
-0.323823994075312	1.88143474974977\\
-0.323823994075312	-1.88143474974977\\
-0.29363514200919	1.22640379224603\\
-0.29363514200919	-1.22640379224603\\
-0.2038248647216	0.94463560324838\\
-0.2038248647216	-0.94463560324838\\
-0.252036893987081	1.00134493775187\\
-0.252036893987081	-1.00134493775187\\
-0.147351432007518	0.710116344193173\\
-0.147351432007518	-0.710116344193173\\
-0.284394860523878	1.91888102370668\\
-0.284394860523878	-1.91888102370668\\
-0.140022362766175	0.968879151699482\\
-0.140022362766175	-0.968879151699482\\
-0.18434614465231	1.05579177422331\\
-0.18434614465231	-1.05579177422331\\
-0.222141145409633	0.879930439948177\\
-0.222141145409633	-0.879930439948177\\
-0.405556523287192	0.944265570014934\\
-0.405556523287192	-0.944265570014934\\
-0.190120992539854	0.714000914438818\\
-0.190120992539854	-0.714000914438818\\
-0.106763505469726	1.97530178103673\\
-0.106763505469726	-1.97530178103673\\
-0.19812808598011	2.24656347933142\\
-0.19812808598011	-2.24656347933142\\
-0.277870734956806	1.75139060345433\\
-0.277870734956806	-1.75139060345433\\
-0.417294182505708	1.12641766378175\\
-0.417294182505708	-1.12641766378175\\
-0.094381624861479	0.952699518262915\\
-0.094381624861479	-0.952699518262915\\
-0.106096851021714	2.2333647970102\\
-0.106096851021714	-2.2333647970102\\
-0.356593485631804	0.831815479685496\\
-0.356593485631804	-0.831815479685496\\
-0.15871342643837	1.05103919565233\\
-0.15871342643837	-1.05103919565233\\
-0.176309621658374	2.08967390296261\\
-0.176309621658374	-2.08967390296261\\
-0.332331782511147	1.93762807369675\\
-0.332331782511147	-1.93762807369675\\
-0.172910081557675	2.05896960600762\\
-0.172910081557675	-2.05896960600762\\
-0.13914501073577	1.24946212854366\\
-0.13914501073577	-1.24946212854366\\
-0.17070371228223	1.21363938333505\\
-0.17070371228223	-1.21363938333505\\
-0.151180688444206	1.71244601059526\\
-0.151180688444206	-1.71244601059526\\
-0.434294799723896	1.0987997407717\\
-0.434294799723896	-1.0987997407717\\
-0.159260308380002	2.22664788455923\\
-0.159260308380002	-2.22664788455923\\
-0.328255642822053	2.15440501251483\\
-0.328255642822053	-2.15440501251483\\
-0.262809723408437	2.09606352143213\\
-0.262809723408437	-2.09606352143213\\
-0.0989437255731092	0.634567841416984\\
-0.0989437255731092	-0.634567841416984\\
-0.131359015217511	2.16992464655305\\
-0.131359015217511	-2.16992464655305\\
-0.17587688912989	1.81052096830176\\
-0.17587688912989	-1.81052096830176\\
-0.216551027239439	0.912988442804884\\
-0.216551027239439	-0.912988442804884\\
-0.154132524915427	0.923057195815572\\
-0.154132524915427	-0.923057195815572\\
-0.214224534903851	0.82627767593588\\
-0.214224534903851	-0.82627767593588\\
-0.336852322325483	1.81664972126841\\
-0.336852322325483	-1.81664972126841\\
-0.203868070716449	1.04207431012158\\
-0.203868070716449	-1.04207431012158\\
-0.255811213538302	0.908048771675186\\
-0.255811213538302	-0.908048771675186\\
-0.226280490882426	2.2151312178652\\
-0.226280490882426	-2.2151312178652\\
-0.138314578620777	1.10096111475657\\
-0.138314578620777	-1.10096111475657\\
-0.363825705699247	2.13785987013475\\
-0.363825705699247	-2.13785987013475\\
-0.0894041807536524	0.813820638629044\\
-0.0894041807536524	-0.813820638629044\\
-0.252875154824605	2.24928710554291\\
-0.252875154824605	-2.24928710554291\\
-0.133117840366822	2.11888586447212\\
-0.133117840366822	-2.11888586447212\\
-0.247884403339852	1.79946040066925\\
-0.247884403339852	-1.79946040066925\\
-0.315323795427443	0\\
-0.209373091339776	0.775520625855705\\
-0.209373091339776	-0.775520625855705\\
-0.0807038159307	2.25692865073333\\
-0.0807038159307	-2.25692865073333\\
-0.377945440248008	2.10302640685372\\
-0.377945440248008	-2.10302640685372\\
-0.0910671529632948	2.21572908331021\\
-0.0910671529632948	-2.21572908331021\\
-0.102500627916786	1.21979386079322\\
-0.102500627916786	-1.21979386079322\\
-0.0994892105490881	1.06166193323046\\
-0.0994892105490881	-1.06166193323046\\
-0.210125036635334	1.70744502437571\\
-0.210125036635334	-1.70744502437571\\
-0.430682043106463	0.970291179451009\\
-0.430682043106463	-0.970291179451009\\
-0.100577004343788	0.980821953570156\\
-0.100577004343788	-0.980821953570156\\
-0.242653597438195	1.92628830747148\\
-0.242653597438195	-1.92628830747148\\
-0.356261551493437	1.90175601731557\\
-0.356261551493437	-1.90175601731557\\
-0.361859345008451	1.95132356562736\\
-0.361859345008451	-1.95132356562736\\
-0.307943601541012	2.06292145645393\\
-0.307943601541012	-2.06292145645393\\
-0.223984252039913	1.17209376554256\\
-0.223984252039913	-1.17209376554256\\
-0.323808348195282	1.963485248085\\
-0.323808348195282	-1.963485248085\\
-0.380413063891887	1.04271129456699\\
-0.380413063891887	-1.04271129456699\\
-0.332927663340707	1.83800186531769\\
-0.332927663340707	-1.83800186531769\\
-0.334593736686551	1.21752529551939\\
-0.334593736686551	-1.21752529551939\\
-0.248072252067221	1.85276822541115\\
-0.248072252067221	-1.85276822541115\\
-0.351435687716819	1.00598709380341\\
-0.351435687716819	-1.00598709380341\\
-0.192105977941009	0.835307231959576\\
-0.192105977941009	-0.835307231959576\\
-0.161846422534678	1.98110543686057\\
-0.161846422534678	-1.98110543686057\\
-0.143675773050971	1.22000455849899\\
-0.143675773050971	-1.22000455849899\\
-0.256775913697201	2.04336591519303\\
-0.256775913697201	-2.04336591519303\\
-0.349057848206897	1.09111663110978\\
-0.349057848206897	-1.09111663110978\\
-0.34465786287	2.09972272580791\\
-0.34465786287	-2.09972272580791\\
-0.311252247130104	1.21536700027315\\
-0.311252247130104	-1.21536700027315\\
-0.159160036097504	2.04431602255158\\
-0.159160036097504	-2.04431602255158\\
-0.316304984160979	1.9037407112486\\
-0.316304984160979	-1.9037407112486\\
-0.089159197535976	0.700338751943473\\
-0.089159197535976	-0.700338751943473\\
-0.313679861257375	1.84567244910411\\
-0.313679861257375	-1.84567244910411\\
-0.190499416132063	0.910964886040553\\
-0.190499416132063	-0.910964886040553\\
-0.417606391402174	1.01752267852137\\
-0.417606391402174	-1.01752267852137\\
-0.154647030448855	0.828425422040644\\
-0.154647030448855	-0.828425422040644\\
-0.264700585868376	1.21118159101095\\
-0.264700585868376	-1.21118159101095\\
-0.131458967770856	0.883964851813663\\
-0.131458967770856	-0.883964851813663\\
};

\end{axis}
\end{tikzpicture}%

%% file: figure/agents_vs_sr.tex
%
%
\definecolor{mycolor1}{rgb}{0.00000,0.44700,0.74100}%
\definecolor{mycolor2}{rgb}{0.85000,0.32500,0.09800}%
\begin{tikzpicture}

\begin{axis}[%
width=4.521in,
height=3.566in,
at={(0.758in,0.481in)},
scale only axis,
xmin=1,
xmax=10,
ymin=0,
ymax=1.5,
axis background/.style={fill=white},
axis x line*=bottom,
axis y line*=left,
legend style={legend cell align=left, align=left, draw=white!15!black, at={(axis cs:9.5,1.4)}, anchor=north east}
]

\addplot[only marks, mark=square, mark options={}, mark size=5.0pt, draw=mycolor2] table[row sep=crcr]{%
x	y\\
1	0.472668716962924\\
2	0.348795355366569\\
3	0.300345839745955\\
4	0.278852552493227\\
5	0.257675912044122\\
6	0.235463746584936\\
7	0.218893611883018\\
8	0.204109124501717\\
9	0.192447445872489\\
10	0.182591080368647\\
};
\addlegendentry{\Large Unconstrained adversaries}

\addplot[only marks, mark=asterisk, mark options={}, mark size=5.000pt, draw=mycolor1] table[row sep=crcr]{%
x	y\\
1	1.41774388847238\\
2	0.413293046490177\\
3	0.33412331991404\\
4	0.2871880951696\\
5	0.258642027966532\\
6	0.235537723951546\\
7	0.218202269158359\\
8	0.204124347549686\\
9	0.192460050522755\\
10	0.182574216142382\\
};
\addlegendentry{\Large Sign preserving adversaries}

\end{axis}
\end{tikzpicture}%

%% file: figure/10_agent_ps.tex
%
%
\begin{tikzpicture}

\begin{axis}[%
width=4.521in,
height=3.566in,
at={(0.758in,0.481in)},
scale only axis,
xmin=-0.2,
xmax=0,
ymin=-2,
ymax=2,
xtick={-0.2,-0.1,0},
ytick={-2,0,2},
axis background/.style={fill=white},
legend style={legend cell align=left, align=left, draw=white!15!black}
]
\addplot [color=black, draw=none, mark=*, mark options={solid, black}]
  table[row sep=crcr]{%
-0.0980660571282232	1.90203213819325\\
-0.0980660571282232	-1.90203213819325\\
-0.120502418134253	1.89912523461907\\
-0.120502418134253	-1.89912523461907\\
-0.0892666577071702	1.17199018886793\\
-0.0892666577071702	-1.17199018886793\\
-0.12911653901895	1.17784053320535\\
-0.12911653901895	-1.17784053320535\\
-0.108921494780629	0.0149629304790786\\
-0.108921494780629	-0.0149629304790786\\
};

\addplot [color=blue, draw=none, mark=*, mark size=1pt, mark options={solid, blue}]
  table[row sep=crcr]{%
-0.1	0\\
-0.1	1.17557050458495\\
-0.1	-1.17557050458495\\
-0.1	1.90211303259031\\
-0.1	-1.90211303259031\\
-0.1	1.90211303259031\\
-0.1	-1.90211303259031\\
-0.0999999999999999	1.17557050458495\\
-0.0999999999999999	-1.17557050458495\\
-0.1	0\\
-0.104964076459124	1.90683673512965\\
-0.104964076459124	-1.90683673512965\\
-0.118516383148175	1.92365309429719\\
-0.118516383148175	-1.92365309429719\\
-0.0997957717704414	1.19083402696785\\
-0.0997957717704414	-1.19083402696785\\
-0.116989690925868	1.17631769991694\\
-0.116989690925868	-1.17631769991694\\
-0.119321869317369	0\\
-0.110448001870559	1.91361470609675\\
-0.110448001870559	-1.91361470609675\\
-0.106217101843473	1.87220654030179\\
-0.106217101843473	-1.87220654030179\\
-0.111137051315855	1.182444397084\\
-0.111137051315855	-1.182444397084\\
-0.106524541961824	1.15367351675222\\
-0.106524541961824	-1.15367351675222\\
-0.108330860956557	0.0180025348831158\\
-0.108330860956557	-0.0180025348831158\\
-0.0873445256880509	1.88211215020362\\
-0.0873445256880509	-1.88211215020362\\
-0.0632551219245244	0\\
-0.0697565536516318	1.16833405342255\\
-0.0697565536516318	-1.16833405342255\\
-0.099861506905528	1.91299659952347\\
-0.099861506905528	-1.91299659952347\\
-0.0929814036052923	1.88635982835857\\
-0.0929814036052923	-1.88635982835857\\
-0.0951610496391947	0.0140356619440836\\
-0.0951610496391947	-0.0140356619440836\\
-0.0916533755969264	1.15706195991495\\
-0.0916533755969264	-1.15706195991495\\
-0.0931861564470167	1.91249544401079\\
-0.0931861564470167	-1.91249544401079\\
-0.096505151829247	1.86771031255184\\
-0.096505151829247	-1.86771031255184\\
-0.100861779185511	1.17027541188198\\
-0.100861779185511	-1.17027541188198\\
-0.0892355869309527	1.16514682528335\\
-0.0892355869309527	-1.16514682528335\\
-0.0939774084055842	0\\
-0.0870769181422945	1.87634724524986\\
-0.0870769181422945	-1.87634724524986\\
-0.106122901419636	0\\
-0.0837207178905022	0\\
-0.0789729021969951	1.16636881373984\\
-0.0789729021969951	-1.16636881373984\\
-0.10951184908657	1.17359166210836\\
-0.10951184908657	-1.17359166210836\\
-0.0981768070990558	1.88403137554146\\
-0.0981768070990558	-1.88403137554146\\
-0.0969902418364923	1.18220080758474\\
-0.0969902418364923	-1.18220080758474\\
-0.108108966525407	1.16225441904961\\
-0.108108966525407	-1.16225441904961\\
-0.102542533167945	0.00792487237042118\\
-0.102542533167945	-0.00792487237042118\\
-0.0809634846001778	1.92565156954063\\
-0.0809634846001778	-1.92565156954063\\
-0.0792184771128556	1.19295413765212\\
-0.0792184771128556	-1.19295413765212\\
-0.0698851306133201	1.87316281437393\\
-0.0698851306133201	-1.87316281437393\\
-0.0986061167705384	1.89701604138249\\
-0.0986061167705384	-1.89701604138249\\
-0.0833991811894445	0.0177935222335808\\
-0.0833991811894445	-0.0177935222335808\\
-0.0729976426673786	1.16300846975716\\
-0.0729976426673786	-1.16300846975716\\
-0.108042807938286	1.9201509976082\\
-0.108042807938286	-1.9201509976082\\
-0.127284677208443	1.89968057046196\\
-0.127284677208443	-1.89968057046196\\
-0.11799496384261	1.19039368505304\\
-0.11799496384261	-1.19039368505304\\
-0.119082512921118	1.17050441470982\\
-0.119082512921118	-1.17050441470982\\
-0.129250021181444	0\\
-0.111125812890295	1.89607816400048\\
-0.111125812890295	-1.89607816400048\\
-0.105131530706964	1.92511953840396\\
-0.105131530706964	-1.92511953840396\\
-0.0944442477094708	1.20299191558882\\
-0.0944442477094708	-1.20299191558882\\
-0.114279931384325	1.1584859912731\\
-0.114279931384325	-1.1584859912731\\
-0.0846599673290172	1.18323049255969\\
-0.0846599673290172	-1.18323049255969\\
-0.12579216358348	1.91625949804959\\
-0.12579216358348	-1.91625949804959\\
-0.139723427179376	0\\
-0.128221076366725	1.19318232161232\\
-0.128221076366725	-1.19318232161232\\
-0.12322276786234	1.91092001724272\\
-0.12322276786234	-1.91092001724272\\
-0.11490772565962	0.0153910506935142\\
-0.11490772565962	-0.0153910506935142\\
-0.0932727854345136	1.89385548742385\\
-0.0932727854345136	-1.89385548742385\\
-0.115914930734992	1.88146084260343\\
-0.115914930734992	-1.88146084260343\\
-0.126294210159135	1.16692729152569\\
-0.126294210159135	-1.16692729152569\\
-0.0822514010659346	1.90396157620103\\
-0.0822514010659346	-1.90396157620103\\
-0.0699800625284278	1.86733907513193\\
-0.0699800625284278	-1.86733907513193\\
-0.0510649084810351	0\\
-0.09500601651143	1.16651473469266\\
-0.09500601651143	-1.16651473469266\\
-0.0540053028727175	1.16306038474499\\
-0.0540053028727175	-1.16306038474499\\
-0.111187227411676	1.87066025611128\\
-0.111187227411676	-1.87066025611128\\
-0.0912834451152527	1.17508098165605\\
-0.0912834451152527	-1.17508098165605\\
-0.114514114483105	1.15070459490683\\
-0.114514114483105	-1.15070459490683\\
-0.097405091598269	1.92255666814418\\
-0.097405091598269	-1.92255666814418\\
-0.122625837251433	1.18641996188284\\
-0.122625837251433	-1.18641996188284\\
-0.0737695736302989	1.17602176997527\\
-0.0737695736302989	-1.17602176997527\\
-0.0763187536245229	0\\
-0.103834341612204	1.87877440300685\\
-0.103834341612204	-1.87877440300685\\
-0.112757593175888	0\\
-0.132113114807563	1.91242733074223\\
-0.132113114807563	-1.91242733074223\\
-0.119232395135877	1.15738043785318\\
-0.119232395135877	-1.15738043785318\\
-0.103592389464604	1.86734303443989\\
-0.103592389464604	-1.86734303443989\\
-0.0857386923761519	1.93538961256126\\
-0.0857386923761519	-1.93538961256126\\
-0.0853110734657637	1.20168444793241\\
-0.0853110734657637	-1.20168444793241\\
-0.0728200438155189	1.91236484527806\\
-0.0728200438155189	-1.91236484527806\\
-0.104860196565295	1.17930607865051\\
-0.104860196565295	-1.17930607865051\\
-0.0691303722578997	0\\
-0.10895558549547	1.88583417911223\\
-0.10895558549547	-1.88583417911223\\
-0.104012773379144	1.14220985589159\\
-0.104012773379144	-1.14220985589159\\
-0.117635989945244	1.90417079340212\\
-0.117635989945244	-1.90417079340212\\
-0.0903687492941478	1.94478847522943\\
-0.0903687492941478	-1.94478847522943\\
-0.121929768836637	1.88963123980305\\
-0.121929768836637	-1.88963123980305\\
-0.101763037709031	1.19853405783763\\
-0.101763037709031	-1.19853405783763\\
-0.0864648899887383	1.89718098899034\\
-0.0864648899887383	-1.89718098899034\\
-0.073338904932027	1.89712452928803\\
-0.073338904932027	-1.89712452928803\\
-0.0791217028774845	0.0256301978189545\\
-0.0791217028774845	-0.0256301978189545\\
-0.0733022680135897	1.15676554147553\\
-0.0733022680135897	-1.15676554147553\\
-0.112701892811726	1.90219543166474\\
-0.112701892811726	-1.90219543166474\\
-0.111375881587669	1.16679479042626\\
-0.111375881587669	-1.16679479042626\\
-0.0917180692173115	1.90006501133786\\
-0.0917180692173115	-1.90006501133786\\
-0.106430676884995	0.011509202623431\\
-0.106430676884995	-0.011509202623431\\
-0.134880345773277	1.16427528557527\\
-0.134880345773277	-1.16427528557527\\
-0.145610367843755	0\\
-0.103044690328606	1.89054061227748\\
-0.103044690328606	-1.89054061227748\\
-0.0864044064127599	1.14882184628081\\
-0.0864044064127599	-1.14882184628081\\
-0.117002263047486	1.16515589407955\\
-0.117002263047486	-1.16515589407955\\
-0.0918493502591905	1.18544980650141\\
-0.0918493502591905	-1.18544980650141\\
-0.0809330119995458	1.88629119436419\\
-0.0809330119995458	-1.88629119436419\\
-0.0901474787661906	0.0138555748133974\\
-0.0901474787661906	-0.0138555748133974\\
-0.0695306105448565	1.90784956577911\\
-0.0695306105448565	-1.90784956577911\\
-0.0760857667529126	1.18434186125659\\
-0.0760857667529126	-1.18434186125659\\
-0.122123967351793	1.88308092178865\\
-0.122123967351793	-1.88308092178865\\
-0.118811732266489	1.14062573734592\\
-0.118811732266489	-1.14062573734592\\
-0.105216014836889	0.0244474412499\\
-0.105216014836889	-0.0244474412499\\
-0.12111605669084	1.14773568779947\\
-0.12111605669084	-1.14773568779947\\
-0.104578689818915	1.93495035455863\\
-0.104578689818915	-1.93495035455863\\
-0.116000202047484	1.91426065364311\\
-0.116000202047484	-1.91426065364311\\
-0.124060295979268	1.1763542543598\\
-0.124060295979268	-1.1763542543598\\
-0.112167956564937	0.00704044333746073\\
-0.112167956564937	-0.00704044333746073\\
-0.0988129593752775	1.87730593395221\\
-0.0988129593752775	-1.87730593395221\\
-0.103412425344077	1.88473569152708\\
-0.103412425344077	-1.88473569152708\\
-0.0976820092080092	0.0204528059959422\\
-0.0976820092080092	-0.0204528059959422\\
-0.0981687234160908	1.14928947636543\\
-0.0981687234160908	-1.14928947636543\\
-0.076856878097973	1.92228599731274\\
-0.076856878097973	-1.92228599731274\\
-0.0817759982182711	1.87398579749008\\
-0.0817759982182711	-1.87398579749008\\
-0.0923748545203154	1.87581785371915\\
-0.0923748545203154	-1.87581785371915\\
-0.0859011933226912	1.16996135450114\\
-0.0859011933226912	-1.16996135450114\\
-0.130711705530548	1.18502764615497\\
-0.130711705530548	-1.18502764615497\\
-0.109382918468696	1.14670012642259\\
-0.109382918468696	-1.14670012642259\\
-0.0976397297871646	1.94086304235604\\
-0.0976397297871646	-1.94086304235604\\
-0.0874325141207567	1.91973103056831\\
-0.0874325141207567	-1.91973103056831\\
-0.101343394777331	1.18552786142533\\
-0.101343394777331	-1.18552786142533\\
-0.0635686655416419	1.15806250133481\\
-0.0635686655416419	-1.15806250133481\\
-0.0941460171671285	0.00828864268818415\\
-0.0941460171671285	-0.00828864268818415\\
-0.0818798267079539	1.17376129978114\\
-0.0818798267079539	-1.17376129978114\\
-0.128014317431922	1.16137225218927\\
-0.128014317431922	-1.16137225218927\\
-0.109744968327062	1.20224251050588\\
-0.109744968327062	-1.20224251050588\\
-0.0695331962025262	1.92137079164181\\
-0.0695331962025262	-1.92137079164181\\
-0.077541532760581	1.8919309069871\\
-0.077541532760581	-1.8919309069871\\
-0.0580131435900098	0\\
-0.0617922081370825	1.17385240773841\\
-0.0617922081370825	-1.17385240773841\\
-0.0981749311293042	1.21087232394751\\
-0.0981749311293042	-1.21087232394751\\
-0.0923360810313255	1.90644217853627\\
-0.0923360810313255	-1.90644217853627\\
-0.0839179744076619	1.89062637925696\\
-0.0839179744076619	-1.89062637925696\\
-0.0841412734350939	1.16337471968729\\
-0.0841412734350939	-1.16337471968729\\
-0.0919890705846778	1.15045277150375\\
-0.0919890705846778	-1.15045277150375\\
-0.0901177632203821	1.92857515744774\\
-0.0901177632203821	-1.92857515744774\\
-0.0773262532333175	1.90169318786721\\
-0.0773262532333175	-1.90169318786721\\
-0.0671762211830653	1.17835067561765\\
-0.0671762211830653	-1.17835067561765\\
-0.0887647667077023	1.85190344550255\\
-0.0887647667077023	-1.85190344550255\\
-0.0941845030203147	1.16157045576823\\
-0.0941845030203147	-1.16157045576823\\
-0.0884093197955126	1.91434114667029\\
-0.0884093197955126	-1.91434114667029\\
-0.0909816075843672	1.86653299181305\\
-0.0909816075843672	-1.86653299181305\\
-0.101323919761696	1.15597278762756\\
-0.101323919761696	-1.15597278762756\\
-0.119401283153705	1.18135794012727\\
-0.119401283153705	-1.18135794012727\\
-0.106197417160939	1.16939462590333\\
-0.106197417160939	-1.16939462590333\\
-0.132488414895113	1.87093404259736\\
-0.132488414895113	-1.87093404259736\\
-0.0975953569899553	1.8907139770368\\
-0.0975953569899553	-1.8907139770368\\
-0.114087455021918	1.86094659343485\\
-0.114087455021918	-1.86094659343485\\
-0.099098782134243	1.93581607693193\\
-0.099098782134243	-1.93581607693193\\
-0.0826668038189307	1.94227061897868\\
-0.0826668038189307	-1.94227061897868\\
-0.108839495617272	1.19510993567846\\
-0.108839495617272	-1.19510993567846\\
-0.139107070563123	1.17216974645176\\
-0.139107070563123	-1.17216974645176\\
-0.09785371871813	1.8482708455928\\
-0.09785371871813	-1.8482708455928\\
-0.0753496358837396	1.15105590486161\\
-0.0753496358837396	-1.15105590486161\\
-0.0810306924290003	1.88035730296369\\
-0.0810306924290003	-1.88035730296369\\
-0.100796123563516	1.13412367881226\\
-0.100796123563516	-1.13412367881226\\
-0.100176176793276	0.0147118561304238\\
-0.100176176793276	-0.0147118561304238\\
-0.112540551528848	0.0339375946356799\\
-0.112540551528848	-0.0339375946356799\\
-0.110334920676626	1.90751748548808\\
-0.110334920676626	-1.90751748548808\\
-0.0903667915924398	1.1914622157399\\
-0.0903667915924398	-1.1914622157399\\
-0.0946893200988834	1.95362244906618\\
-0.0946893200988834	-1.95362244906618\\
-0.0943540355061273	1.21433088176982\\
-0.0943540355061273	-1.21433088176982\\
-0.0635283659245214	1.92929831628094\\
-0.0635283659245214	-1.92929831628094\\
-0.0697479272303352	1.20351429158507\\
-0.0697479272303352	-1.20351429158507\\
-0.106998933248363	1.86206252113007\\
-0.106998933248363	-1.86206252113007\\
-0.0857003065530116	0.0267284667895209\\
-0.0857003065530116	-0.0267284667895209\\
-0.103844446643896	1.89712561698291\\
-0.103844446643896	-1.89712561698291\\
-0.13177814583155	1.14824391339921\\
-0.13177814583155	-1.14824391339921\\
-0.128291004490438	1.93725758089913\\
-0.128291004490438	-1.93725758089913\\
-0.120098523563452	1.87665147991465\\
-0.120098523563452	-1.87665147991465\\
-0.113538940515173	1.9384223880967\\
-0.113538940515173	-1.9384223880967\\
-0.111585079585245	1.89009509038731\\
-0.111585079585245	-1.89009509038731\\
-0.116998649279434	0.0231673317984023\\
-0.116998649279434	-0.0231673317984023\\
-0.079897078252644	1.90843726402325\\
-0.079897078252644	-1.90843726402325\\
-0.0924284178898063	1.1961213989834\\
-0.0924284178898063	-1.1961213989834\\
-0.137261842840736	1.92730538561715\\
-0.137261842840736	-1.92730538561715\\
-0.128788529598628	1.173048213292\\
-0.128788529598628	-1.173048213292\\
-0.124823209693198	0.0137435168973944\\
-0.124823209693198	-0.0137435168973944\\
-0.099967948242296	1.16458359374463\\
-0.099967948242296	-1.16458359374463\\
-0.110533523080723	1.13597729299435\\
-0.110533523080723	-1.13597729299435\\
-0.0667336722177494	1.19178414691469\\
-0.0667336722177494	-1.19178414691469\\
-0.106041055008179	1.94834616920695\\
-0.106041055008179	-1.94834616920695\\
-0.0926283263431766	0.0246529975357314\\
-0.0926283263431766	-0.0246529975357314\\
-0.127625042390785	1.92681643660336\\
-0.127625042390785	-1.92681643660336\\
-0.0611632324291169	1.89366149934311\\
-0.0611632324291169	-1.89366149934311\\
-0.11168933625252	1.8756955976386\\
-0.11168933625252	-1.8756955976386\\
-0.138524567243157	1.16080548240866\\
-0.138524567243157	-1.16080548240866\\
-0.111272926368874	1.95413863769635\\
-0.111272926368874	-1.95413863769635\\
-0.0958689620300677	1.22512897230773\\
-0.0958689620300677	-1.22512897230773\\
-0.109093345851109	1.95923467940565\\
-0.109093345851109	-1.95923467940565\\
-0.110201451489249	1.20776117298726\\
-0.110201451489249	-1.20776117298726\\
-0.109303212035899	1.18872386412729\\
-0.109303212035899	-1.18872386412729\\
-0.104383844059112	1.91544993793828\\
-0.104383844059112	-1.91544993793828\\
-0.113223857944715	1.9296132143276\\
-0.113223857944715	-1.9296132143276\\
-0.122663039945418	1.90222731096307\\
-0.122663039945418	-1.90222731096307\\
-0.0685448612022287	1.9018243919766\\
-0.0685448612022287	-1.9018243919766\\
-0.0877369792728449	1.90429588661366\\
-0.0877369792728449	-1.90429588661366\\
-0.100542671971229	1.92729389376518\\
-0.100542671971229	-1.92729389376518\\
-0.0798448391461125	1.1787702546291\\
-0.0798448391461125	-1.1787702546291\\
-0.0738138498050059	1.88486911609138\\
-0.0738138498050059	-1.88486911609138\\
-0.119783963737291	0.00958373236711843\\
-0.119783963737291	-0.00958373236711843\\
-0.143358877371537	1.89952335832155\\
-0.143358877371537	-1.89952335832155\\
-0.113016239909334	1.94800987739741\\
-0.113016239909334	-1.94800987739741\\
-0.0734271634556068	1.94353052221718\\
-0.0734271634556068	-1.94353052221718\\
-0.0727521787960479	1.20856956970471\\
-0.0727521787960479	-1.20856956970471\\
-0.0815657830460833	0.00668583676552036\\
-0.0815657830460833	-0.00668583676552036\\
-0.106575169549119	1.94089647022298\\
-0.106575169549119	-1.94089647022298\\
-0.134259058990437	1.88059507481139\\
-0.134259058990437	-1.88059507481139\\
-0.101076722535453	0.0320976256157481\\
-0.101076722535453	-0.0320976256157481\\
-0.13605951296663	1.92241770948283\\
-0.13605951296663	-1.92241770948283\\
-0.116355229306117	1.20905458288208\\
-0.116355229306117	-1.20905458288208\\
-0.143435261972522	1.1771019673768\\
-0.143435261972522	-1.1771019673768\\
-0.0917765200807595	1.93892895168908\\
-0.0917765200807595	-1.93892895168908\\
-0.0841100819404906	0.0354108741246026\\
-0.0841100819404906	-0.0354108741246026\\
-0.0815747708307377	1.21789819490485\\
-0.0815747708307377	-1.21789819490485\\
-0.131584537832742	1.89534060166083\\
-0.131584537832742	-1.89534060166083\\
-0.123676057370899	0.0276744382246481\\
-0.123676057370899	-0.0276744382246481\\
-0.135428844748239	1.18327125258997\\
-0.135428844748239	-1.18327125258997\\
-0.0812364866742678	1.15849625783138\\
-0.0812364866742678	-1.15849625783138\\
-0.121626816926731	1.16297894782418\\
-0.121626816926731	-1.16297894782418\\
-0.0585258397741462	1.16787830009935\\
-0.0585258397741462	-1.16787830009935\\
-0.127265182527827	1.86980771301582\\
-0.127265182527827	-1.86980771301582\\
-0.0871661995522332	0.00808470789051341\\
-0.0871661995522332	-0.00808470789051341\\
-0.0815101197389185	1.86838137261346\\
-0.0815101197389185	-1.86838137261346\\
-0.0793153697650006	1.91751098418094\\
-0.0793153697650006	-1.91751098418094\\
-0.085576663316861	1.19628053704444\\
-0.085576663316861	-1.19628053704444\\
-0.0821083618756045	1.18885816683725\\
-0.0821083618756045	-1.18885816683725\\
-0.0799436490819481	1.14886003162827\\
-0.0799436490819481	-1.14886003162827\\
-0.11515097404445	1.18558015609771\\
-0.11515097404445	-1.18558015609771\\
-0.134514152773874	1.19031358104993\\
-0.134514152773874	-1.19031358104993\\
-0.126233762391991	1.13846091015479\\
-0.126233762391991	-1.13846091015479\\
-0.112210630294438	1.92302569861921\\
-0.112210630294438	-1.92302569861921\\
-0.14567237661381	1.19124110293111\\
-0.14567237661381	-1.19124110293111\\
-0.152548637915946	0\\
-0.0896589564594941	1.89026553156107\\
-0.0896589564594941	-1.89026553156107\\
-0.11475681580713	1.14468729671152\\
-0.11475681580713	-1.14468729671152\\
-0.0974611226832688	1.14267659899347\\
-0.0974611226832688	-1.14267659899347\\
-0.12731416355059	1.20131932279823\\
-0.12731416355059	-1.20131932279823\\
-0.101216813722863	1.20660152784396\\
-0.101216813722863	-1.20660152784396\\
-0.0785099765835839	0.0151836206708971\\
-0.0785099765835839	-0.0151836206708971\\
-0.0618274727010415	1.90399945409483\\
-0.0618274727010415	-1.90399945409483\\
-0.125548471974447	1.89406837512706\\
-0.125548471974447	-1.89406837512706\\
-0.126439452338457	1.18161526196335\\
-0.126439452338457	-1.18161526196335\\
-0.0804932494367484	1.20332240643469\\
-0.0804932494367484	-1.20332240643469\\
-0.0761467983421599	1.93382789094705\\
-0.0761467983421599	-1.93382789094705\\
-0.0711392692542783	1.18831945593544\\
-0.0711392692542783	-1.18831945593544\\
-0.126136264508595	1.14705233415227\\
-0.126136264508595	-1.14705233415227\\
-0.0936270311882278	1.13779497915981\\
-0.0936270311882278	-1.13779497915981\\
-0.107036482730558	0.00569903784106474\\
-0.107036482730558	-0.00569903784106474\\
-0.0711755899430585	1.18163679192189\\
-0.0711755899430585	-1.18163679192189\\
-0.133378802692035	1.90524998965236\\
-0.133378802692035	-1.90524998965236\\
-0.1175302779977	0.0298391950530479\\
-0.1175302779977	-0.0298391950530479\\
-0.120340478735676	1.19584747919271\\
-0.120340478735676	-1.19584747919271\\
-0.0921764099439517	1.92330161514532\\
-0.0921764099439517	-1.92330161514532\\
-0.134204436439024	1.15550783272696\\
-0.134204436439024	-1.15550783272696\\
-0.133828464159257	1.88881843980629\\
-0.133828464159257	-1.88881843980629\\
-0.103548772885623	1.21315500707988\\
-0.103548772885623	-1.21315500707988\\
-0.0978745480838574	1.90714565122624\\
-0.0978745480838574	-1.90714565122624\\
-0.118809823953842	1.94395531448919\\
-0.118809823953842	-1.94395531448919\\
-0.076490692164511	1.86818850095472\\
-0.076490692164511	-1.86818850095472\\
-0.063861121517869	1.16561511766335\\
-0.063861121517869	-1.16561511766335\\
-0.12251630147314	1.92007240306617\\
-0.12251630147314	-1.92007240306617\\
-0.144635942556672	1.16838548518429\\
-0.144635942556672	-1.16838548518429\\
-0.117775586601456	1.2020560736983\\
-0.117775586601456	-1.2020560736983\\
-0.116489664193755	1.93368069299226\\
-0.116489664193755	-1.93368069299226\\
-0.124482548727095	1.94177538314255\\
-0.124482548727095	-1.94177538314255\\
-0.0651816697104995	1.88816332651236\\
-0.0651816697104995	-1.88816332651236\\
-0.0638187398794768	1.18237787054346\\
-0.0638187398794768	-1.18237787054346\\
-0.0905260826974353	1.85808369423981\\
-0.0905260826974353	-1.85808369423981\\
-0.125358026586799	1.87830849034208\\
-0.125358026586799	-1.87830849034208\\
-0.110515591703937	0.0230954427171492\\
-0.110515591703937	-0.0230954427171492\\
-0.127963868972924	1.88931529345355\\
-0.127963868972924	-1.88931529345355\\
-0.0989042132317114	1.95027305544591\\
-0.0989042132317114	-1.95027305544591\\
-0.129342289731213	1.92136697861707\\
-0.129342289731213	-1.92136697861707\\
-0.116754685077258	1.89676884910322\\
-0.116754685077258	-1.89676884910322\\
-0.114026460680101	1.85222603111354\\
-0.114026460680101	-1.85222603111354\\
-0.0824472643085365	1.8548898277106\\
-0.0824472643085365	-1.8548898277106\\
-0.0758848369586044	1.21985841603579\\
-0.0758848369586044	-1.21985841603579\\
-0.0966694425272235	1.86256142117285\\
-0.0966694425272235	-1.86256142117285\\
-0.107777102673851	1.90043685943573\\
-0.107777102673851	-1.90043685943573\\
-0.0867717546302349	1.15820154648425\\
-0.0867717546302349	-1.15820154648425\\
-0.079417588619542	1.95029613662437\\
-0.079417588619542	-1.95029613662437\\
-0.0769117378523678	1.1711472793143\\
-0.0769117378523678	-1.1711472793143\\
-0.0682899806933943	1.88172548418735\\
-0.0682899806933943	-1.88172548418735\\
-0.0917382775950185	1.18026269705233\\
-0.0917382775950185	-1.18026269705233\\
-0.0944337702779128	1.9324486594624\\
-0.0944337702779128	-1.9324486594624\\
-0.0899540796327536	0.0196062637286668\\
-0.0899540796327536	-0.0196062637286668\\
-0.0895450811492754	1.20918000477982\\
-0.0895450811492754	-1.20918000477982\\
-0.0691599362430362	1.9320440777727\\
-0.0691599362430362	-1.9320440777727\\
-0.0842408672104198	1.93023360265816\\
-0.0842408672104198	-1.93023360265816\\
-0.109028991958895	1.88005866826655\\
-0.109028991958895	-1.88005866826655\\
-0.0553130855451072	1.88804110671352\\
-0.0553130855451072	-1.88804110671352\\
-0.0741465232525205	1.19242524571679\\
-0.0741465232525205	-1.19242524571679\\
-0.0962331969716115	1.91756165208383\\
-0.0962331969716115	-1.91756165208383\\
-0.134982579469106	1.93210012814754\\
-0.134982579469106	-1.93210012814754\\
-0.0861045708182294	1.90935720654915\\
-0.0861045708182294	-1.90935720654915\\
-0.134394528969895	1.17782177665901\\
-0.134394528969895	-1.17782177665901\\
-0.0761699477332228	1.87454494024484\\
-0.0761699477332228	-1.87454494024484\\
-0.126693760280473	1.15359651671133\\
-0.126693760280473	-1.15359651671133\\
-0.088744497815716	0\\
-0.0769476594201517	1.85314513279115\\
-0.0769476594201517	-1.85314513279115\\
-0.0629425805029732	1.15235320076372\\
-0.0629425805029732	-1.15235320076372\\
-0.0848434279517822	1.95142006941371\\
-0.0848434279517822	-1.95142006941371\\
-0.106438055006244	1.18427252642385\\
-0.106438055006244	-1.18427252642385\\
-0.118552400944398	1.86540783792023\\
-0.118552400944398	-1.86540783792023\\
-0.101333981393846	1.85859862233979\\
-0.101333981393846	-1.85859862233979\\
-0.121358996799719	1.92968631513452\\
-0.121358996799719	-1.92968631513452\\
-0.134453069045548	0\\
-0.112186830467629	1.2140641757175\\
-0.112186830467629	-1.2140641757175\\
-0.10709364618771	0.0354036698169783\\
-0.10709364618771	-0.0354036698169783\\
-0.103302761185142	1.14743023930611\\
-0.103302761185142	-1.14743023930611\\
-0.140954148618803	1.14883060699528\\
-0.140954148618803	-1.14883060699528\\
-0.0662001051875882	1.20724306857462\\
-0.0662001051875882	-1.20724306857462\\
-0.0883305778594659	1.13016131581165\\
-0.0883305778594659	-1.13016131581165\\
-0.0708471202233192	1.89060248425627\\
-0.0708471202233192	-1.89060248425627\\
-0.13236125158308	1.16936284043193\\
-0.13236125158308	-1.16936284043193\\
-0.0936070576155683	1.88080107423485\\
-0.0936070576155683	-1.88080107423485\\
-0.0601734225144413	1.18613692113985\\
-0.0601734225144413	-1.18613692113985\\
-0.108977403785723	0.0289739523390265\\
-0.108977403785723	-0.0289739523390265\\
-0.0869013787036564	1.13811137327994\\
-0.0869013787036564	-1.13811137327994\\
-0.104806900072642	1.19114059877764\\
-0.104806900072642	-1.19114059877764\\
-0.102283451584987	0.0473432918888906\\
-0.102283451584987	-0.0473432918888906\\
-0.117260804988789	1.87118199555854\\
-0.117260804988789	-1.87118199555854\\
-0.10914400015294	1.85703257889119\\
-0.10914400015294	-1.85703257889119\\
-0.109960362952265	1.93475399832374\\
-0.109960362952265	-1.93475399832374\\
-0.0950772969192317	1.85558497866892\\
-0.0950772969192317	-1.85558497866892\\
-0.0948849369868602	1.17154654467468\\
-0.0948849369868602	-1.17154654467468\\
-0.145510849480765	1.88882323713388\\
-0.145510849480765	-1.88882323713388\\
-0.135822046095625	0.0140477401954004\\
-0.135822046095625	-0.0140477401954004\\
-0.148583230122373	1.17812901165177\\
-0.148583230122373	-1.17812901165177\\
-0.0860874321557356	1.92500360308027\\
-0.0860874321557356	-1.92500360308027\\
-0.0839541025064821	1.86214686327969\\
-0.0839541025064821	-1.86214686327969\\
-0.129303477603056	1.88411661726122\\
-0.129303477603056	-1.88411661726122\\
-0.103029194852851	0.0196720960561543\\
-0.103029194852851	-0.0196720960561543\\
-0.0931024263836052	0.0358113865564458\\
-0.0931024263836052	-0.0358113865564458\\
-0.053870342275634	1.92242824177348\\
-0.053870342275634	-1.92242824177348\\
-0.0702056984341672	0.012372884379332\\
-0.0702056984341672	-0.012372884379332\\
-0.14548128923636	1.89471276204216\\
-0.14548128923636	-1.89471276204216\\
-0.100050298377933	0.0251522384982404\\
-0.100050298377933	-0.0251522384982404\\
-0.114892145845524	1.19501883830854\\
-0.114892145845524	-1.19501883830854\\
-0.120715655723712	0.0401462156602239\\
-0.120715655723712	-0.0401462156602239\\
-0.0730106713457604	1.92558841697252\\
-0.0730106713457604	-1.92558841697252\\
-0.0953478035534495	1.94674601875444\\
-0.0953478035534495	-1.94674601875444\\
-0.0876988785215076	1.95685202981414\\
-0.0876988785215076	-1.95685202981414\\
-0.0628472741197677	1.91440100856588\\
-0.0628472741197677	-1.91440100856588\\
-0.06876556397315	1.17349788405811\\
-0.06876556397315	-1.17349788405811\\
-0.0693031574882138	1.14946514981326\\
-0.0693031574882138	-1.14946514981326\\
-0.138494252840467	1.91592693234749\\
-0.138494252840467	-1.91592693234749\\
-0.0725832361852194	0.0178550590089977\\
-0.0725832361852194	-0.0178550590089977\\
-0.116634484200888	1.95498335444262\\
-0.116634484200888	-1.95498335444262\\
-0.111410951051846	1.8653004649002\\
-0.111410951051846	-1.8653004649002\\
-0.0750643038473855	1.1997452077461\\
-0.0750643038473855	-1.1997452077461\\
-0.0858951770566043	1.1776725854352\\
-0.0858951770566043	-1.1776725854352\\
-0.0767198081812135	0.0406404002984624\\
-0.0767198081812135	-0.0406404002984624\\
-0.0735973140628984	1.86363911368705\\
-0.0735973140628984	-1.86363911368705\\
-0.0868108202082468	0.043458549791696\\
-0.0868108202082468	-0.043458549791696\\
-0.0899563991962373	1.21844622676449\\
-0.0899563991962373	-1.21844622676449\\
-0.0898064857703439	1.8441740379116\\
-0.0898064857703439	-1.8441740379116\\
-0.158314781887207	0\\
-0.05887870514154	1.18081681852221\\
-0.05887870514154	-1.18081681852221\\
-0.107373907456386	1.9304623212545\\
-0.107373907456386	-1.9304623212545\\
-0.10426929651523	1.95516695949793\\
-0.10426929651523	-1.95516695949793\\
-0.11864623398256	1.21409491515763\\
-0.11864623398256	-1.21409491515763\\
-0.0823977713540308	1.14439566228776\\
-0.0823977713540308	-1.14439566228776\\
-0.0592435729284274	1.8739908632093\\
-0.0592435729284274	-1.8739908632093\\
-0.0600911904843	1.19741992637942\\
-0.0600911904843	-1.19741992637942\\
-0.0988497837113344	0.0384381870530612\\
-0.0988497837113344	-0.0384381870530612\\
-0.090221927223048	1.14451100819125\\
-0.090221927223048	-1.14451100819125\\
-0.0681135583032411	1.93709020766132\\
-0.0681135583032411	-1.93709020766132\\
-0.0751268537348297	1.14188375383934\\
-0.0751268537348297	-1.14188375383934\\
-0.0679194029777467	1.89505675235699\\
-0.0679194029777467	-1.89505675235699\\
-0.0645221259898939	1.14658706248605\\
-0.0645221259898939	-1.14658706248605\\
-0.126113991875236	0.00409993277344948\\
-0.126113991875236	-0.00409993277344948\\
-0.131493379173923	1.94683175611373\\
-0.131493379173923	-1.94683175611373\\
-0.110585829551038	0.043106371801735\\
-0.110585829551038	-0.043106371801735\\
-0.0746081237105162	0.00768350083430265\\
-0.0746081237105162	-0.00768350083430265\\
-0.124277989208317	1.86485763605432\\
-0.124277989208317	-1.86485763605432\\
-0.117226875280756	1.88695533730948\\
-0.117226875280756	-1.88695533730948\\
-0.101958142209492	1.92007568432519\\
-0.101958142209492	-1.92007568432519\\
-0.0699084704282118	1.19789520984756\\
-0.0699084704282118	-1.19789520984756\\
-0.0449370757068146	0\\
-0.0852189799865227	1.21250343997808\\
-0.0852189799865227	-1.21250343997808\\
-0.0947290288051442	0.0308517465610785\\
-0.0947290288051442	-0.0308517465610785\\
-0.0494523524825987	1.18270084208258\\
-0.0494523524825987	-1.18270084208258\\
-0.050332141075343	1.89550877546294\\
-0.050332141075343	-1.89550877546294\\
-0.140092093185804	1.18616880288482\\
-0.140092093185804	-1.18616880288482\\
-0.13266076651128	0.00604631494432766\\
-0.13266076651128	-0.00604631494432766\\
-0.122194363043667	1.85279149260115\\
-0.122194363043667	-1.85279149260115\\
};

\end{axis}
\end{tikzpicture}%